\documentclass[aps,pra,superscriptaddress,showpacs,showkeys,notitlepage, numbers,onecolumn,longbibliography,nofootinbib]{revtex4-2}

\pdfoutput=1
\usepackage[utf8]{inputenc}
\usepackage[english]{babel}
\usepackage[T1]{fontenc}

\usepackage{amsmath,amssymb,amsthm}
\usepackage{hyperref}
\usepackage{tikz}
\usepackage{float}
\usepackage{subcaption}
\usepackage{graphicx}
\usepackage{physics}

\hypersetup{
     colorlinks=true,
     linkcolor=blue,
     filecolor=blue,
     citecolor = magenta,      
     urlcolor=black,
}

\usetikzlibrary{quantikz}

\newcommand{\mycomment}[1]{}

\newtheorem{theorem}{Theorem}[section]
\newtheorem{corollary}{Corollary}[section]
\newtheorem{lemma}[theorem]{Lemma}

\theoremstyle{definition}

\begin{document}

\title{Efficient Quantum Circuits for Non-Unitary and Unitary Diagonal Operators with Space-Time-Accuracy trade-offs}

\author{Julien Zylberman}
\affiliation{Sorbonne Université, Observatoire de Paris, Université PSL, CNRS, LUX, 75005 Paris, France}

%
\author{Ugo Nzongani}
\affiliation{CNRS, LIS, Aix-Marseille Université, Université de Toulon, Marseille, France}
\affiliation{Unité de Mathématiques Appliquées, ENSTA Paris, Institut Polytechnique de Paris, 91120 Palaiseau, France}

\author{Andrea Simonetto}
\affiliation{Unité de Mathématiques Appliquées, ENSTA Paris, Institut Polytechnique de Paris, 91120 Palaiseau, France}

\author{Fabrice Debbasch}
\affiliation{Sorbonne Université, Observatoire de Paris, Université PSL, CNRS, LUX, 75005 Paris, France}

\begin{abstract}

Unitary and non-unitary diagonal operators are fundamental building blocks in quantum algorithms with applications in the resolution of partial differential equations, Hamiltonian simulations, the loading of classical data on quantum computers (quantum state preparation) and many others. In this paper, we introduce a general approach to implement unitary and non-unitary diagonal operators with efficient-adjustable-depth quantum circuits. The depth, {\sl i.e.}, the number of layers of quantum gates of the quantum circuit, is reducible with respect either to the width, {\sl i.e.}, the number of ancilla qubits, or to the accuracy between the implemented operator and the target one. While exact methods have an optimal exponential scaling either in terms of size, {\sl i.e.}, the total number of primitive quantum gates, or width, approximate methods prove to be efficient for the class of diagonal operators depending on smooth, at least differentiable, functions. Our approach is general enough to allow any method for diagonal operators to become adjustable-depth or approximate, decreasing the depth of the circuit by increasing its width or its approximation level. This feature offers flexibility and can match with the hardware limitations in coherence time or cumulative gate error. We illustrate these methods by performing quantum state preparation and non-unitary-real-space simulation of the diffusion equation. This simulation paves the way to efficient implementations of stochastic models useful in physics, chemistry, biology, image processing and finance.

\end{abstract}


\maketitle


\section{Introduction}

Gate-based quantum computing leverages the principles of quantum mechanics to perform computations that would be prohibitively challenging on classical computers. The advantage of quantum algorithms lies in the possible decomposition of $n$-qubit operations into a reasonable number of hardware feasible single- and two-qubit operations. In this context, the development of new decompositions which improve on the current ones or widen the range of operations efficiently implementable on quantum computers, to include, {\sl e.g}, non-unitary operations, is crucial to achieve novel applications and to lower the amount of resources.

In this paper, we focus on implementing one of the most fundamental operations in quantum computing, {\sl i.e}, the diagonal operation, be it unitary or non-unitary. In fact, the computational cost of implementing diagonal operators is a limiting factor in many quantum algorithms with applications as varied as real-space simulation of unitary and non-unitary dynamics \cite{zalka1998simulating, wiesner1996simulations, kassal2008polynomial, jordan2012quantum,budinski2021quantum,linden2022quantum,seneviratne2024exact,mangin2024efficient}, quantum optimization \cite{farhi2001quantum,10.1145/3569095} or quantum state preparation \cite{grover2002creating, PhysRevA.109.042401, zhang2022quantum}.

The development of canonical decomposition of diagonal unitaries was first mentioned in the context of general unitary synthesis by A.Barenco et al. in 1995 \cite{barenco1995elementary}. Considering an $n$-qubit diagonal unitary, the authors propose to implement each eigenvalue sequentially with $2^n$ multi-controlled phase gates, each of these gates being in turn implementable with $\mathcal{O}(n^2)$ primitive quantum gates (or only $\mathcal{O}(n)$ if one uses one ancilla qubit). In 2004, S.Bullock et al. \cite{bullock2004asymptotically} were the first to present a quantum circuit related to the Walsh-Hadamard transform of the phases of the diagonal unitaries.
Their quantum circuits, made only of controlled-NOT gate and one qubit $z$-axis rotations, implement any diagonal unitary with $2^{n+1} -3$ gates, and are asymptotically optimal in terms of size for exact implementation. It is only in 2014 that Welch et al. \cite{welch2014efficient} developed the first approximate quantum circuits with no exponential scalings for diagonal unitaries depending on smooth, at least once differentiable, functions. The complexity of their quantum circuits, also based on Walsh-Hadamard transform, is directly related to the maximum value of the derivative of the function, making the synthesis efficient as long as the function has small variations. The same year, Welch et al. developed efficient quantum circuits in the Clifford+T basis for diagonal unitaries with $k\ll 2^n$ distinct eigenvalues \cite{10.5555/3179320.3179326}. In 2021, approximate quantum compiling have been developed to train diagonal ansatze \cite{geller2021experimental}. Then, the work of Zhang et al. \cite{zhang2024depth} improved the work of Welch et al. \cite{welch2014efficient} by reducing the depth of the quantum circuits by a factor of $2$ while maintaining an optimal size scaling. In 2023, Sun et al. \cite{sun2023asymptotically} developed a recursive scheme for exact implementation of diagonal unitaries based on a particular ordering of the Walsh operators. Their method has an optimal depth scaling as $\mathcal{O}(2^n/n)$ and can be implemented with ancilla qubits to reduce the depth of the quantum circuits. In 2024, Zylberman et al. \cite{PhysRevA.109.042401} implemented diagonal unitaries using sparse Walsh Series: by considering the smallest number of Walsh coefficients approximating a smooth, at least differentiable, function, they were able to reduce the depth of the associated quantum circuits. In 2024 also, Nzongani et al. \cite{nzongani2024adjustable} considered adjustable-depth methods to parallelize block-diagonal operators. Recently, Claudon et al. \cite{claudon2024polylogarithmic} were able to revisit the initial results of Barenco et al., by developing a poly-logarithmic depth quantum circuit for multi-controlled operations, improving the circuit complexity of implementing one eigenvalue of a diagonal operator from $\mathcal{O}(n)$ to $\mathcal{O}(\log(n)^3)$ using one ancilla qubit. The work of D.Motlagh et al. \cite{motlagh2024generalized} uses a quantum signal processing protocol to approximate efficiently diagonal operators through their Fourier decompositions, using one ancilla qubit and at the cost of a ressource-intensive classical optimization.

All above methods can be grouped into three families: (i) the methods that use a sequential decomposition and implement one eigenvalue at a time \cite{barenco1995elementary,bullock2004asymptotically,claudon2024polylogarithmic} (ii) the methods that use a Walsh (or Walsh-Hadamard) decomposition \cite{bullock2004asymptotically,welch2014efficient,sun2023asymptotically,PhysRevA.109.042401} and (iii) the methods based on a classical optimization \cite{motlagh2024generalized}. All methods in (i) and (ii) express diagonal unitaries as a product of simpler, efficiently implementable, diagonal unitaries\footnote{and do not rely on ressource intensive classical preprocessing}. For generic diagonal unitary, exact decompositions involve an exponential amount of these simpler diagonal unitaries, leading to non-efficient quantum circuits. 
The hope is then that, for many applications, the relevant diagonal unitaries are actually non generic. For example, the relevant diagonal unitaries can be sparse, {\sl i.e.}, only have a certain “small” number of eigenvalues different from unity, or the phases of the eigenvalues can be computed from functions with some smoothness properties, as for real-space simulations based on the implementation of evolution operators which are diagonal in a specific basis \cite{welch2014efficient,childs2022quantum}.

In this work, we present an adjustable-depth framework to implement diagonal unitaries: from any decomposition into primitive diagonal operators, be it sparse or dense, exact or approximate, the depth of the associated quantum circuits can be reduced using an arbitrary number of ancilla qubits. The ancilla qubits are used to parallelize the implementation of the primitive diagonal operators, giving a first method to adjust the depth of the circuits. For the class of diagonal unitary depending on smooth, at least differentiable, functions, we prove that one can also adjust the depth of the circuits with respect to the approximation, giving a second tunable feature of the circuit. The complexities of the constructed quantum circuits are summarized in different theorems:
\begin{itemize}
\item Full parallelization Theorem \ref{thm : full parallelization} for any diagonal unitary decomposed into a product of $p$ simpler, efficiently implementable, diagonal unitaries $\hat{U}=\prod_{j=0}^{p-1}\hat{U}_j$ given a sufficient number of ancilla qubits.
\item Adjustable-depth Theorem \ref{thm : adjustable-depth with ancilla} for any diagonal unitary decomposed into a product of $p$ simpler, efficiently implementable, diagonal unitaries $\hat{U}=\prod_{j=0}^{p-1}\hat{U}_j$ given a limited number of ancilla qubits.
\item Approximate Theorem \refeq{thm: Approximation theorem} to modify any exact methods into an approximate one for any diagonal unitary associated to a smooth, at least differentiable, function.
\end{itemize}

These first results (Theorems \ref{thm : full parallelization}, \ref{thm : adjustable-depth with ancilla}) are similar in spirit to those in \cite{moore2001parallel}, but they were developed independently and they are here presented in full generality and with numerical simulations.

From now on we focus on diagonal non-unitaries. Implementing non-unitary operators with quantum computers has been often discussed in the last decades \cite{gingrich2004non, terashima2005nonunitary,Briegel_2009, daskin2017ancilla} and is crucial to broaden the range of applications accessible to quantum computing. Some algorithms of quantum state preparation (QSP) involve non-unitary operations \cite{gingrich2004non, PhysRevA.106.022414}. Linear combination of unitaries (LCU) achieves Hamiltonian simulation by implementing the Taylor expansion of the evolution operator into a series of non-unitary operators \cite{10.5555/2481569.2481570, berry2015simulating}. Solving ordinary differential equations is also possible using non-unitary operations within the LCU framework \cite{childs2017quantum, berry2017quantum} and many quantum partial differential equations solvers are based on non-unitary operations \cite{gilyen2019quantum,low2019hamiltonian}. The quantum singular value transform (QSVT) makes it possible to implement polynomials of a block encoded operator and quantum algorithms such as Shor's algorithm \cite{Shor_1997}, quantum search \cite{grover1996fast,giri2017review} or quantum phase estimation \cite{Brassard_2002} can be recast as a QSVT \cite{gilyen2019quantum, martyn2021grand}. Most of these algorithms are based on the block-encoding of a non-unitary operator into a larger unitary one by adding ancilla qubits and the computational resources are often given as the number of queries to a block-encoded operator. In particular, the development of quantum circuits for non-unitary diagonal operators has attracted much attention because these operators appear naturally in the context of stochastic processes \cite{oksendal2013stochastic} and are therefore useful in a wide variety of contexts which range from physics, chemistry and biology \cite{van2004stochastic, freund2000stochastic} to image processing (see section 3.6 in \cite{gonzalez2009digital}) and finance \cite{oksendal2013stochastic, BLACK1976167,pironneau2009partial}.

In this article, we show how to block-encode a non-unitary diagonal operator from a unitary one, keeping the adjustable properties of the quantum circuits. We present the different strategies that one can use starting from a block-encoded diagonal operator: the block-encoding may be enough for the target application, but one may also need to measure an ancilla qubit to get a renormalized qubit state on which the diagonal non-unitary operator has been performed. The probability of success depends directly on the diagonal operator and on the qubit state. Two distincts approaches are presented: a non-destructive repeat-until-success scheme and an amplitude amplification scheme for which the depth of the quantum circuit is adjustable with respect to the probability of success. Finally, we illustrate these methods with two applications: a new quantum state preparation protocols of smooth, at least differentiable, functions is presented and the real-space simulation of the diffusion equation is performed by preparing an initial Gaussian distribution which evolves through the non-unitary evolution operator of the diffusion equation. These contributions can be summarize by the following points:
\begin{itemize}
\item Block-encoding Theorem \ref{thm: block-encoding} of non-unitary diagonal operators $\hat{D}$ using adjustable-depth quantum circuits.
\item Table of complexities \ref{Table : approximate case} for unitary and non-unitary diagonal operators depending on smooth, at least differentiable, functions.
\item Table of complexities \ref{Table : sparse case} for sparse unitary and non-unitary diagonal operators.
\item Non-destructive repeat-until-success protocol for non-unitary diagonal operators when amplitude amplification fails.
\item Efficient quantum state preparation Theorem \ref{thm: qsp} with space-time-accuracy trade-offs to prepare quantum states depending on differentiable functions. 
\item Resolution of the diffusion equation, see Fig. \ref{fig: heat equation resolution}, with high fidelity, few qubits and low depth quantum circuits.
\end{itemize}

The remainder of the paper is organized as follows. Section \ref{sec:preli} presents the sequential and Walsh-Hadamard decompositions and how to copy quantum registers. Section \ref{sec:adj_framework} introduces the general approach to adjust the depth of the quantum circuits with respect to the number of ancilla qubits and Section \ref{sec:adj_framework_approx} with respect to the accuracy. In Section \ref{sec:non_unitary}, the adjustable-depth quantum circuits for non-unitary diagonal operators are presented. Section \ref{sec:adj_depth_methods} summarizes the asymptotic scalings of the different methods for diagonal operators. Section \ref{section amplitude amplification} details the amplitude amplification and the non-destructive repeat-until-success protocols. Finally, Section \ref{sec:applications} presents two applications: quantum state preparation and the resolution of the diffusion equation generated by Brownian motion. Additional technical details are provided in the appendices, and some code files for constructing quantum circuits are available at \cite{github_ugo}.

\section{Preliminaries}\label{sec:preli}

This section introduces preliminary notions used across the paper. We present the metrics chosen to estimate the computational resources of the quantum circuits, as well as the two decompositions of diagonal unitaries, {\sl i.e.}, the sequential and the Walsh-Hadamard decomposition. We also introduce the copy unitary used to prepare the ancilla registers to parallelize the computations.

\subsection{Depth, Size, Width, and Accuracy}

Different metrics exist to evaluate the computational resources needed to implement a given quantum circuit. 
Size is the number of single-qubit and two-qubit gates. Depth is the number of layers of quantum gates, and it is related to the time needed to perform the desired operation. Width is the number of ancilla qubits. A metric for accuracy for us, is the error in spectral norm between the implemented operator and the target operator. The set of primitive single-qubit and two-qubit gates considered in this article is the set of all single-qubit gates complemented by the control NOT gate (CNOT). 
Other metrics such as the CNOT count or the T-count can also be estimated from the scalings provided in this paper. 
A unitary operator acting on $n$ qubits is said to be efficiently implementable if there is a quantum circuit implementing it up to an error $\epsilon>0$ in spectral norm with a size, depth and width scaling at worst as $\mathcal{O}(\text{poly}(n,1/\epsilon))$. By extension, the quantum circuit is said efficient. A non-unitary operator is said to be efficiently implementable on a set  $\mathcal{H}$ of qubit states if the associated quantum circuit is efficient and the probability of successfully implementing the non-unitary operator on any $n$-qubit state $\ket{\psi} \in \mathcal{H}$ scales at worst as $\Omega(\text{poly}(1/n,\epsilon))$.

Some of the quantum circuits presented in this article have a depth and/or a size independent of $n$, making them particularly relevant for large scale fault-tolerant algorithms.

\subsection{Sequential decomposition}
The sequential decomposition expresses the diagonal unitary $\hat{U}=\sum_{x=0}^{N-1}e^{i\theta_x}\ket{x}\bra{x}$ acting on $n$ qubits (with $N=2^n$ eigenvalues) as a product of simple diagonal unitaries $\hat{U}_j=e^{i\theta_{j}}\ket{j}\bra{j}+\sum_{x\neq j} \ket{x}\bra{x}$ with $j\in\{0,...,N-1\}$, $x=\sum_{i=0}^{n-1}x_i 2^i \in \{0,...,N-1\}$, $x_i\in\{0,1\}$ and $\ket{x}=\ket{x_0}\otimes \cdots \otimes\ket{x_{n-1}}$ being the computational basis of the Hilbert space $\mathcal{H}_2^{\otimes n}$. For each $j$, the operator ${\hat U}_j$ has by construction $2^n - 1$ eigenvalues equal to unity and the other eigenvalue is $\exp(i \theta_{j})$, which is one of the eigenvalues of $\hat U$:
\begin{equation}
    \hat{U}=\prod_{j=0}^{N-1}\hat{U}_j,
\end{equation}

and with the notation $\hat{U}=\text{diag}(u_0,\dots,u_{N-1})$:
\begin{equation}
    \text{diag}(u_0,...,u_{N-1})=\text{diag}(u_0,1,\dots,1)\times\text{diag}(1,u_1,\dots,1)\times\cdots\times\text{diag}(1,\dots,1,u_{N-1}).
\end{equation}

Thus, each $\hat{U}_j$ is a $n$-qubit unitary which changes the phase of one of the eigenvectors of the computational basis. By denoting $j=\sum_{i=0}^{n-1}j_i2^i$ the binary decomposition of $j$  (where $ j_i\in\{0,1\}$ for all values of $i$), one can show that $\hat{U}_j$ corresponds to applying the phase $e^{i\theta_j}$ on one of the $n$ qubits if all the qubits are in state $\ket{j_i}$. The operator $\hat{U}_j$ can therefore be written as a $(n-1)$-controlled phase gate $\hat{P}(\theta_j)=\ket{0}\bra{0}+e^{i\theta_j}\ket{1}\bra{1}$ if $j_{n-1}=1$, or $e^{i\theta_j}\ket{0}\bra{0}+\ket{1}\bra{1}=\hat{X}\hat{P}(\theta_j)\hat{X}$ if $j_{n-1}=0$, which is controlled by the qubits $i\in \{0,...,n-2\}$ for which $j_i=1$ and anti-controlled by the qubits $i$ for which $j_i=0$. The anti-controls can be rewritten as control operations using $\hat{X}$-Pauli gates:
\begin{equation}
    \hat{U}_j=(\hat{X}_0^{j_0}\otimes...\otimes\hat{X}_{n-1}^{j_{n-1}})\Lambda_{\{0,...,n-2\}}(\hat{P}(\theta_j))(\hat{X}_0^{j_0}\otimes...\otimes\hat{X}_{n-1}^{j_{n-1}}).
\label{sequential decomposition}
\end{equation}
with $\Lambda_{\{1,...,n-1\}}(\hat{P}(\theta_j))$ a $(n-1)$-controlled $\hat{P}(\theta_j)$ gates and $\hat{X}_i^{j_i}$ the $\hat{X}$-Pauli gate applied on qubit $i$ if $j_i=1$.

Let $k_j$ be the number of $1$'s in the binary decomposition of $j$. Each $\hat{U}_j$ is implementable using $2k_j$ $\hat{X}$-Pauli gates and a $(n-1)$-multi-controlled gate.  The recent work of Claudon et al. \cite{claudon2024polylogarithmic} introduces polylogarithmic-depth quantum circuits to implement $(n-1)$-multi-controlled gates by breaking it down into smaller multi-controlled gates which are implemented in parallel. Their approximate method implements the gate $\Lambda_{\{1,...,n-1\}}(\hat{P}(\theta_j))$  up to an error $\epsilon>0$ with a quantum circuit of depth $\mathcal{O}(\log(n)^3\log(1/\epsilon))$, size $\mathcal{O}(n\log(n)^4\log(1/\epsilon))$ and without ancilla qubits (Proposition 2 in \cite{claudon2024polylogarithmic}). One can also implement exactly the $n$-control gate using one ancilla qubit (zeroed or borrowed\footnote{ Zeroed ancilla qubits are initially in state $\ket{0}$ and reset to $\ket{0}$ at the end of the computation. Borrowed ancilla qubits are in an arbitrary state $\ket{\psi}$ potentially entangled with other qubits and are restored to their initial state afterwards.}) with a depth $\mathcal{O}(\log(n)^3)$ and a size $\mathcal{O}(n\log(n)^4)$ (Corollary 1 in \cite{claudon2024polylogarithmic}) or without ancilla qubits with a linear depth $\mathcal{O}(n)$ and a linear size $\mathcal{O}(n)$ \cite{Craig}. 

Being diagonal operators, the ${\hat U}_j$'s commute with one another. It thus makes sense to wonder if there is an optimal order in which to implement these operators. 
It is possible to use a Gray code order \cite{gray-pulse-code-communication-1953,nielsen2002quantum, beauchamp1984applications}, {\sl i.e.}, from two following operators $\hat{U}_j$ and $\hat{U}_{j'}$ only one bit differs from the binary decomposition of $j$ and $j'$, which cancels a maximum number of $\hat{X}$-Pauli gates from two consecutive operators. Thus, only one $\hat{X}$-Pauli gate remains between each multi-controlled gate, reducing the overall cost of the implementation (see Appendix \ref{example sequential decomposition} for examples of quantum circuits with and without Gray code ordering). 

\subsection{Walsh-Hadamard decomposition}

The exact synthesis of quantum circuits for diagonal unitaries has been proven asymptotically optimal in terms of size (up to a factor $2$) using $2^{n+1}-3$ CNOT and $\hat{R}_Z$ quantum gates \cite{bullock2004asymptotically}. As shown by Welch et al. \cite{welch2014efficient}, this construction is directly related to the Walsh-Hadamard representation of a set of $N=2^n$ data points: a $N$-Walsh Series can represent exactly a set $\{\theta_k\}$ of $N$ data points through the Walsh-Hadamard transform.

Walsh Series and Walsh functions were first introduced by Walsh in 1923 \cite{walsh1923closed}. The Walsh function of order $j \in \{0,1,...,N-1\}$ is defined on $[0,1]$ by
\begin{equation}
    w_j(x)=(-1)^{\sum_{i=0}^{n-1}j_{i}x_{i}},
\label{Walsh function definition}
\end{equation}
where $j_i$ is the $i$-th bit in the binary expansion $j=\sum_{i=0}^{n-1}j_i2^{i}$ and $x_i$ is the $i$-th bit in the dyadic expansion $x=\sum_{i=0}^{\infty} x_i/ 2^{i+1}$.

Walsh functions only take $\pm1$ values depending on their order $j$ and are naturally implementable on quantum computers through their associated $\hat{Z}$-Pauli operators:

\begin{equation}
    \hat{w}_j = \bigotimes_{i=0}^{n-1}(\hat{Z}_i)^{j_i}.
\end{equation}

To implement the diagonal unitary $\hat{U}=\sum_{k=0}^{N-1} e^{i\theta_k}\ket{k}\bra{k}$, which depends on the $N$ real values $\{\theta_k\}$, one first defines the Walsh coefficients:
\begin{equation}
    a_j=\frac{1}{N}\sum_{k=0}^{N-1}\theta_k w_{j}(k/N),
\end{equation}
and write 
\begin{equation}
    \hat{U} = \prod_{j=0}^{N-1}\hat{W}_{j},
\label{walsh decomposition}
\end{equation}
where $\hat{W}_j=e^{ia_j\hat{w}_j}$.

Using the fact that a tensor product of $\hat{Z}$-Pauli gates can be rewritten using two CNOT stairs and one $\hat{Z}$-Pauli gate, one can show that each of the operators $\hat{W}_j$ can be implemented using one $\hat{R}_Z$ gate and $2k_j$ CNOT gates, with $k_j$ the number of bit equals to unity in the binary decomposition of $j$. For $j=\sum_{i=0}^{p-1}j_i2^{i}$ with $1\leq p\leq n$ and $j_{p-1}=1$, the operator $\hat{W}_j$ acts on the $p$-th qubits $0,...,p-1$ as:
\begin{equation}
   \hat{W}_j=\left(\prod_{i=0}^{p-2}\widehat{CNOT}_{i\rightarrow p-1}^{j_i}\right)\left(\hat{I}_2^{\otimes(p-1)} \otimes \hat{R}_Z(a_j))
   \right)\left(\prod_{i=0}^{p-2}\widehat{CNOT}_{i\rightarrow p-1}^{j_i}\right)^{-1},
\label{eq: exp of walsh operator}
\end{equation}
where $\widehat{CNOT}_{i\rightarrow p-1}^{j_i}$ is the CNOT gate controlled by the $i$-th qubit and applied on the $(p-1)$-th qubit if $j_i=1$, the symbol $\prod_{i=0}^{p-2}\hat{A}_i=\hat{A}_0...\hat{A}_{p-2}$ is the product of operator $\hat{A}_i$ with indexes in increasing order, $\hat{I}_2^{\otimes (p-1)}$ is the identity operator on the qubits $0$ to $p-2$ and $\hat{R}_Z(a_j)=e^{i a_j}\ket{0}\bra{0}+e^{-i a_j}\ket{1}\bra{1}$ acts on the qubit $p-1$. The operator $\hat{W}_0(a_0)=e^{ia_0}\hat{I}_2^{\otimes n}$ is a phase encoding the average value of the $\{\theta_k\}$.

Being diagonal operators, the $\hat{W}_j$ commute with one another. Similarly to the sequential decomposition, one can choose a Gray code ordering to cancel a maximum of CNOT gates from two consecutive $\hat{W}_j$ operator (see Appendix \ref{example Walsh-Hadamard decomposition} for examples of quantum circuit with and without Gray code ordering).

Note that sparse Walsh-Hadamard decomposition, {\sl i.e.}, diagonal unitary expressed with only a number $s < N$ of  $\hat{W}_j$ operators, are particularly efficient to approximate diagonal unitaries depending on smooth, at least once differentiable, functions \cite{PhysRevA.109.042401}. In this case, the Gray code ordering may not be efficient to cancel a maximum of CNOT gates.

\subsection{Preparation of the ancilla registers}

The main results of this article use ancilla qubits to adjust the depth of the quantum circuits by parallelizing the implementation of the diagonal unitaries $\hat{U}_j$ or $\hat{W}_j$. This is possible by “copying”\footnote{The name “copy” does not refer to the cloning operation $\ket{\psi}\otimes\ket{0}\rightarrow \ket{\psi}\otimes \ket{\psi}$ which is forbidden by the no-cloning theorem.} the bit value of each of the $n$ qubits into other ancilla qubits. Consider the $i$-th qubit in an arbitrary state $\ket{\psi}_i=\alpha\ket{0}_i+\beta\ket{1}_i$ and a register of $m_i$ zeroed ancilla qubit in state $\ket{0}^{\otimes m_i}$. The copy unitary $\hat{C}_i$ copies acts as:
\begin{equation}
\hat{C}_i(\alpha\ket{0}_i+\beta\ket{1}_i) \otimes \ket{0}^{\otimes m_i}=\alpha\ket{0}^{\otimes (m_i+1)}+\beta\ket{1}^{\otimes (m_i+1)}.
\end{equation}
This operation has already been introduced for other quantum circuit synthesis \cite{Cruz_2019, nzongani2023quantum}. We state and prove clearly the complexity associated to the copy operation with the following lemma.
\begin{lemma}
The “copy” unitary $\hat{C}_i$ acting on $m_i+1$ qubits can be implemented using $m_i$ CNOT gates with a depth $\lceil \log_2(m_i+1) \rceil$.
\label{copy lemma}
\end{lemma}
\begin{proof}
Note first that a set of $q$ CNOT gates applied to different qubits can be implemented in parallel and therefore has a depth $1$. Define $k=\lceil \log_2(m_i+1) \rceil$. A direct recursion on $k'=1,...,k-1$ proves that the copy unitary doubles the number of “copied” qubit at each step, so each newly copied qubit can be used to copy another one at the next step as 
\begin{equation}
\begin{split}
(\alpha\ket{0}+\beta&\ket{1})\otimes \ket{0}^{\otimes m_i} \xrightarrow{\widehat{CNOT}_{0\rightarrow1}} (\alpha\ket{0}^{\otimes 2}+\beta\ket{1}^{\otimes 2})\otimes \ket{0}^{\otimes m_i-1} \\ & \xrightarrow{\widehat{CNOT}_{0\rightarrow2}\otimes \widehat{CNOT}_{1\rightarrow3}} (\alpha\ket{0}^{\otimes 4}+\beta\ket{1}^{\otimes 4})\otimes \ket{0}^{\otimes m_i-3} \\ & \vdots 
\\ & \xrightarrow{\otimes_{j=0}^{2^{k'}-1}\widehat{CNOT}_{j\rightarrow (j+2^{k'})}} \alpha\ket{0}^{\otimes2^{k'+1}}+\beta\ket{1}^{\otimes 2^{k'+1}}\otimes \ket{0}^{(m_i-(2^{k'+1}-1))},
\end{split}
\end{equation}

where $\widehat{CNOT}_{j\rightarrow (j+2^{k'})}$ is the CNOT gate controlled by the $j$-th qubit and applied on $(j+2^{k'})$-th zeroed qubit. After $k-1$ steps, the qubit state is $(\alpha\ket{0}^{\otimes2^{k-1}}+\beta\ket{1}^{\otimes 2^{k-1}})\otimes\ket{0}^{\otimes(m_i-(2^{k-1}-1))}$. A last layer of CNOT gates suffices to produce the state $\alpha\ket{0}^{\otimes(m_i+1)}+\beta\ket{1}^{\otimes (m_i+1)}$. 
\end{proof}

\begin{figure}[ht]
    \centering
\begin{quantikz}
  \lstick[2]{\text{Main register}} & \ctrl{2}\gategroup[4,steps=2,style={dashed,rounded corners,color=blue,inner xsep=2pt},background]{{depth=1}} & \qw & \ctrl{4}\gategroup[8,steps=4,style={dashed,rounded corners,color=blue,inner xsep=2pt},background]{{depth=1}} & \qw & \qw & \qw & \qw \\
  \qw & \qw & \ctrl{2} & \qw & \ctrl{4} & \qw & \qw & \qw \\
  \lstick[2]{\text{1st ancilla register}} & \targ{} & \qw & \qw & \qw & \ctrl{4} & \qw & \qw \\
  \qw & \qw & \targ{} & \qw & \qw & \qw & \ctrl{4} & \qw \\
  \lstick[2]{\text{2nd ancilla register}} & \qw & \qw & \targ{} & \qw & \qw & \qw & \qw \\
  \qw & \qw & \qw & \qw & \targ{} & \qw & \qw & \qw \\
  \lstick[2]{\text{3rd ancilla register}} & \qw & \qw & \qw & \qw & \targ{} & \qw & \qw \\
  \qw & \qw & \qw & \qw & \qw & \qw & \targ{} & \qw 
\end{quantikz} 
\caption{Quantum circuit to copy a register of 2 qubits into 3 ancilla registers. The qubits of the ancilla registers are in the $\ket{0}$ state. The operations framed together can be executed in parallel as they act on different qubits. The depth of the circuit scales logarithmically in the number of copies as the number of copying qubits doubles at each step.}
\label{fig:copy_unitary}
\end{figure}
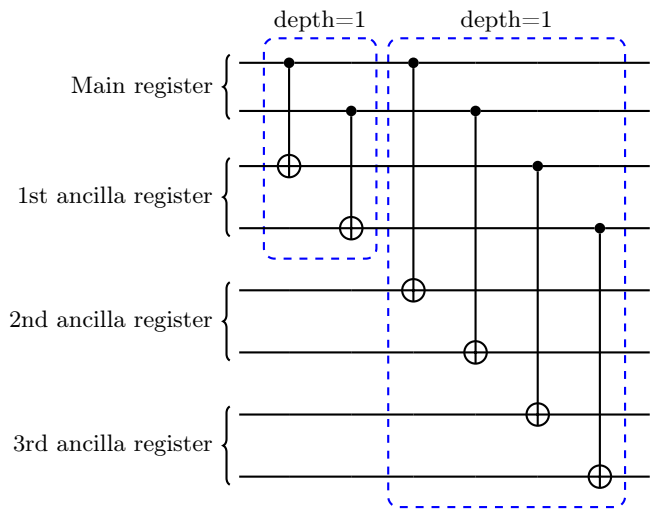

A direct corollary of this lemma adresses the copy of a register of $n$ qubits, where each of the $n$ qubits is copied on $m_i$ other qubits. This operation is performed by implementing each $\hat{C}_i$ in parallel through the operator $\widehat{\text{copy}}=\bigotimes_{i=0}^{n-1} \hat{C}_i$. Its size is given by $\sum_{i=0}^{n-1}m_i$ CNOT gates and its depth is $\lceil \log_2(m+1) \rceil$ where $m=\max_{i=0,...,n-1} m_i$.  A small instance of the copy unitary is shown in Fig. \ref{fig:copy_unitary} where a register of 2 qubits gets copied into 3 ancilla registers of 2 qubits each.

\section{Adjustable-depth framework using ancilla qubits}\label{sec:adj_framework}

In this section, we present an approach to adjust the depth of a quantum circuit implementing a diagonal unitary. Let $\hat{U}$ be a diagonal unitary acting on $n$ qubits with a decomposition $\hat{U}=\prod_{j=0}^{p-1}\hat{U}_j$ where each $\hat{U}_j$ acts on $k_j\leq n$ qubits and is implementable with a quantum circuit of size $s_j$ and depth $d_j$. The $\hat{U}_j$ can be sequential operators (\ref{sequential decomposition}) as well as exponential of Walsh operators (\ref{eq: exp of walsh operator}). Without ancilla qubits, one can implement $\hat{U}$ by performing each $\hat{U}_j$ one by one, leading to a quantum circuit of size $s=\sum_{j=0}^{p-1} s_j$ and depth $d=\sum_{j=0}^{p-1} d_j$.  With ancilla qubits, one can parallelize the implementation of a maximum number of $\hat{U}_j$ operators. The parallelization can be partial or total, reducing the depth proportionally to the number of ancilla qubits available at the cost of preparing the ancilla registers with copy unitaries.

 First, suppose one has enough available ancilla qubits to prepare $p-1$ registers:
\begin{equation}
\begin{split}
    \ket{x}\ket{0}^{\otimes k} \xrightarrow{\widehat{\text{copy}}}  \ket{x}\ket{\tilde{x}}_1...\ket{\tilde{x}}_{p-1},
\end{split}
\end{equation}
where each register $\ket{\tilde{x}}_j$ contains only a copy of the $k_j$ qubits on which $\hat{U}_j$ is acting non-trivially and $k=\sum_{j=1}^{p-1} k_j$. More formally, if one defines the support of $\hat{U}_j$ as the set of qubits on which $\hat{U}_j$ does not act like the identity, the state $\ket{\tilde{x}}_j$ contains only a copy of the qubits in the support of $\hat{U}_j$.

One can then implement each of the $p$ diagonal unitaries $\hat{U}_j=\sum_{x=0}^{N-1}e^{i\theta_j(x)}\ket{x}\bra{x}$ on a different register of qubits:

\begin{equation}
\begin{split}
  \ket{x}\ket{\tilde{x}}_1...\ket{\tilde{x}}_{p-1}  \xrightarrow{ \bigotimes _{j=0}^{p-1} \hat{U}_j } &\hat{U}_0\ket{x}\hat{U}_1\ket{\tilde{x}}_{1}...\hat{U}_{p-1}\ket{\tilde{x}}_{p-1}
 \\ =& e^{i\theta_{0}(x)}e^{i\theta_1(x)}...e^{i\theta_{p-1}(x)}\ket{x}\ket{\tilde{x}}_{1}...\ket{\tilde{x}}_{p-1}
\\=&e^{i\theta(x)}\ket{x}\ket{\tilde{x}}_{1}...\ket{\tilde{x}}_{p-1}
\end{split}
\end{equation}
with $e^{i\theta(x)}=e^{i\sum_{j=0}^{p-1}\theta_j(x)}$ being the eigenvalue of $\hat{U}$ associated to the $\ket{x}$ eigenvector.
Finally, the inverse copy operation resets the ancilla qubits in the state $\ket{0}$:
\begin{equation}
\begin{split}
 e^{i\theta(x)}\ket{x}\ket{\tilde{x}}_{1}...\ket{\tilde{x}}_{p-1} 
  \xrightarrow{ \widehat{\text{copy}}^{-1}}  e^{i\theta(x)}\ket{x}\ket{0}^{\otimes k}.
\end{split}
\end{equation}
The full parallelization protocol has the following complexity:
\begin{theorem} [Full parallelization]
\label{thm : full parallelization}
The implementation of the diagonal operator $\hat{U}=\prod_{j=0}^{p-1}\hat{U}_j$ can be fully parallelized with a quantum circuit using $k=\sum_{j=1}^{p-1}k_j$ ancilla qubits, with a size $s+2k$ and a depth bounded by $\max_{j}(d_j)+2\lceil \log_2(p) \rceil$.
\end{theorem}

\begin{proof}
The size of the quantum circuit is given by the sum of the size of each $\hat{U}_j$ and the size of the $\widehat{\text{copy}}$ and $\widehat{\text{copy}}^{-1}$ unitaries: one controlled-NOT gate is applied on each of the $k$ ancilla qubits for the $\widehat{\text{copy}}$ and a second time for $\widehat{\text{copy}}^{-1}$, giving a final size of $s+2k$. The depth corresponds to the sum of the depths of the two copy unitaries which, in the worst case of a qubit copied $p$ times, is bounded by $2\lceil \log_2(p) \rceil$, and by the depth of the parallelized $\hat{U}_j$, which is given by the maximum of the depths of the $\hat{U}_j$'s.
\end{proof}

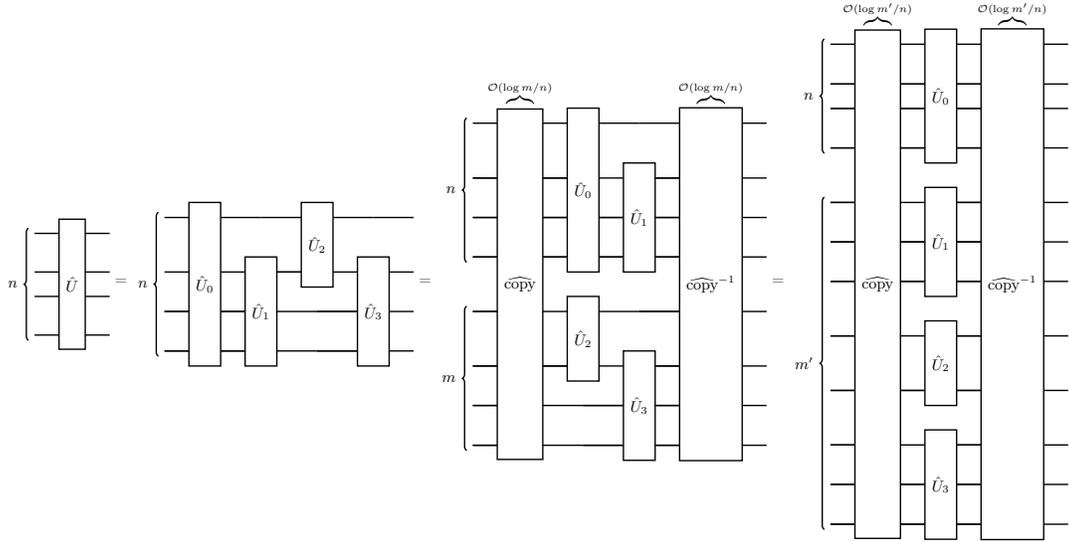
\begin{figure}
    \centering
    \scalebox{0.65}{
    \begin{quantikz}
        \lstick[4]{$n$} & \gate[4]{\hat{U}} & \qw \\
        & & \qw \\
        & & \qw \\
        & & \qw
    \end{quantikz}
    =
    \begin{quantikz}
        \lstick[5]{$n$} & \gate[4]{\hat{U}_0} & \qw & \gate[2]{\hat{U}_2} & \qw & \qw \\
        & & \gate[3]{\hat{U}_1} & \qw & \gate[3]{\hat{U}_3} & \qw \\
        & \qw & \qw & \qw & \qw & \qw \\
        & \qw & \qw & \qw & \qw & \qw
    \end{quantikz}
    =
    \begin{quantikz}
        \lstick[4]{$n$} & \gate[8]{ \widehat{\text{copy}}}\gategroup[8,steps=1,style={dashed,rounded corners,color=white, inner xsep=-10pt, inner ysep=-6pt},background]{{\sc
$\overbrace{}^{\mathcal{O}(\log m/n)}$}} & \gate[4]{\hat{U}_0} & \qw & \gate[8]{ \widehat{\text{copy}}^{-1}}\gategroup[8,steps=1,style={dashed,rounded corners,color=white, inner xsep=-10pt, inner ysep=-6pt},background]{{\sc
$\overbrace{}^{\mathcal{O}(\log m/n)}$}} & \qw  \\
        & \qw & \qw & \gate[3]{\hat{U}_1} & \qw & \qw \\
        & \qw & \qw & \qw & \qw & \qw \\
        & \qw & \qw & \qw & \qw & \qw \\
        \lstick[4]{$m$} & \qw & \gate[2]{\hat{U}_2} & \qw & \qw & \qw \\
        & \qw & \qw & \gate[3]{\hat{U}_3} & \qw & \qw \\
        & \qw & \qw & \qw & \qw & \qw \\
        & \qw& \qw & \qw & \qw & \qw 
    \end{quantikz}
    =
    \begin{quantikz}
        \lstick[4]{$n$} & \gate[12]{ \widehat{\text{copy}}}\gategroup[12,steps=1,style={dashed,rounded corners,color=white, inner xsep=-10pt, inner ysep=-6pt},background]{{\sc
$\overbrace{}^{\mathcal{O}(\log m'/n)}$}} & \gate[4]{\hat{U}_0} & \gate[12]{ \widehat{\text{copy}}^{-1}}\gategroup[12,steps=1,style={dashed,rounded corners,color=white, inner xsep=-10pt, inner ysep=-6pt},background]{{\sc
$\overbrace{}^{\mathcal{O}(\log m'/n)}$}} & \qw \\
        & \qw & \qw & \qw & \qw \\
        & \qw & \qw & \qw & \qw \\
        & \qw & \qw & \qw & \qw \\
        \lstick[8]{$m'$} & \qw & \gate[3]{\hat{U}_1} & \qw & \qw \\
        & \qw & \qw & \qw & \qw \\
        & \qw & \qw & \qw & \qw \\
        & \qw & \gate[2]{\hat{U}_2} & \qw & \qw \\
        & \qw & \qw & \qw & \qw \\
        & \qw & \gate[3]{\hat{U}_3} & \qw & \qw \\
        & \qw & \qw & \qw & \qw \\
        & \qw & \qw & \qw & \qw
    \end{quantikz}
    }
    \caption{Adjustable-depth quantum circuit for a diagonal unitary $\hat{U}=\prod_{j=0}^3\hat{U}_j$ using ancilla qubits. The main register is composed of $n$ qubits, the depth of the circuit is reduced by adding ancilla qubits. The “copy” operations are performed with a depth logarithmic in $m/n$.}
    \label{fig:adjustable_depth_qc}
\end{figure}

If the number $m$ of available ancilla qubits is smaller than $k$, one can gather the $\hat{U}_j$'s into $p'<p$ groups and implement the groups in parallel to each other.
Figure \Ref{fig:adjustable_depth_qc} gives an example of a diagonal unitary composed of $p=4$ efficiently implementable diagonal operators which are implementable sequentially or in parallel depending on the number of ancilla qubits available.  For $m=n$, one can implement half of the $\hat{U}_j$'s in parallel to the others, reducing at best by a factor of $2$ the depth of the circuit (but doubling its width).  The following Theorem \ref{thm : adjustable-depth with ancilla} gives an upper bound for the depth of the adjustable-depth protocol which decreases with the number of ancilla qubits, at the cost of performing the $ \widehat{\text{copy}}$ and $ \widehat{\text{copy}}^{-1}$ unitaries.

\begin{theorem}[Adjustable-depth]
\label{thm : adjustable-depth with ancilla}
Let $m\ge n$ be the number of ancilla qubits and $m'=\lceil m/n \rceil$. The diagonal unitary $\hat{U}=\prod_{j=0}^{p-1}\hat{U}_j$ can be implemented with a quantum circuit of size $s+2m$ and depth $d$ bounded as:
\begin{equation}
    d\leq \sum_{j=0}^{\lceil p/m' \rceil-1 } d_j +2\lceil \log_2(m') \rceil \leq \lceil \frac{p}{m'} \rceil \max_{j}(d_j) +2\lceil \log_2(m') \rceil,
\end{equation}
where the $\hat{U}_j$ are ordered to have decreasing depth $d_0\ge d_1\ge...\ge d_{p-1}$.
\end{theorem}

\begin{proof}
Similarly to the fully parallelized theorem, the size of the quantum circuit is given by the size $s$ of the circuits implementing the $\hat{U}_j$'s added to the number $2m$ of controlled-NOT gates needed to perform the $ \widehat{\text{copy}}$ and $\widehat{\text{copy}}^{-1}$ unitaries. The first inequality is given by the worst case scenario where the $\hat{U}_j$'s act non-trivially on $n$ qubits, in which case one needs at least $m\ge n$ ancilla qubits to reduce the depth of the quantum circuit. In this situation, the number of groups that can be implemented in parallel is $m'=\lfloor m/n \rfloor +1$ (the $+1$ term comes from the $n$ initial qubits, {\sl i.e.}, $m'=\lceil m/n\rceil$). It is always possible to create $m'$ random groups of at most $\lceil p/m' \rceil$ operators. In the worst case, the maximum depth of the groups is the sum from $j=0$ to $\lceil p/m' \rceil-1$ of the maximum depths of the $\hat{U}_j$'s: $\sum_{j=0}^{\lceil p/m' \rceil -1} d_j$. Since the depths of the two copy steps $2\lceil \log_2(m')\rceil$, one thus arrives at the first inequality. The second inequality trivially comes from the fact that $\forall j$, $d_j\leq \max_{j'}(d_{j'})$. 
\end{proof}

This adjustable-depth protocol does not deal with finding an optimal strategy to group the $\hat{U}_j$'s. Such a strategy needs to consider the depth $d_j$ of each operator and to group them in order to minimize the maximum depth of the groups. Finding this optimal solution or a “good” solution can take some classical time and it is in general a difficult combinatorial problem (even though some tailored algorithmic techniques may exist) \cite{bertsimas2005optimization}. We do not look into this here, however, the upper bound is reached for a diagonal unitary made only of sequential operators (\ref{sequential decomposition}) acting non-trivially on $n$ qubits with identical depth. For a diagonal unitary given by a Walsh-Hadamard decomposition, it is possible to group the unitaries with distinct support first, before parallelizing them. 

Overall, the parallelization approach reduces the depth of the quantum circuit while increasing the width, making it possible to adjust the quantum circuit to the constraints of the hardware, {\sl e.g.}, the decoherence time and the number of available qubits. Numerical results illustrating the trade-off between depth and width are given in the application section in Fig. \ref{fig: space_time_tradeoff} for the quantum state preparation of Gaussian states.

\section{Adjustable-depth framework using approximations}\label{sec:adj_framework_approx}
This section presents two efficient methods to approximate diagonal unitaries $\hat{U}_f=e^{i\hat{f}}$ which depend on smooth or, at least, once differentiable functions $f$. The first method transforms any exact quantum circuit for diagonal unitaries into an efficient approximate one. The second method computes a good approximation of $f$ using sparse Walsh Series. Both methods exploit the specific structure of $\hat{U}_f$ to get an efficient approximate quantum circuit, while exact methods suffer from exponential scalings either in terms of depth or width \cite{sun2023asymptotically}. These two approximating methods have efficient scalings with a depth and size asymptotically independent of the number of qubits $n$. By comparison, approximate quantum compiling \cite{geller2021experimental} does not provide any guarantees to reach a given accuracy in the large $n$ limit.

Let $f$ be a smooth, at least once differentiable, function defined on $[0,1]$, $\hat{U}_{f,n}=e^{i\hat{f}_n}=\sum_{x=0}^{N-1} e^{if(x/N)}\ket{x}\bra{x}$ be the target unitary acting on $n$ qubits and $\hat{U}_{f,m}=\sum_{x=0}^{M-1} e^{if(x/M)}\ket{x}\bra{x}$ be the restriction $\hat{U}_{f,n}$ to $m<n$ qubits (with $M=2^m$). The following theorem states that implementing $\hat{U}_{f,m}$ using an exact method gives an efficient quantum circuit approximating the target unitary $\hat{U}_{f,n}$ up to an error $\epsilon>0$ in spectral norm: 

\begin{theorem}
\label{thm: Approximation theorem}
Given a quantum circuit implementing exactly any $n$-qubit diagonal unitary with size $s(n)$, depth $d(n)$ and width $w(n)$, one can $\epsilon$-approximate $\hat{U}_{f,n}$ with a quantum circuit of size $s(m)$, depth $d(m)$ and width $w(m)$ with $m=\lceil \log_2(||f'||_{\infty,[0,1]}/\epsilon) \rceil$.
\end{theorem}

\begin{proof}
The proof is based on the Walsh-Hadamard representation of the target unitary, independently of its method of implementation. First, we prove that implementing exactly the diagonal unitary $\hat{U}_{f,m}$ on $m=\lceil \log_2(||f'||_{\infty,[0,1]}/\epsilon) \rceil$ qubits gives an $\epsilon$-approximation of $\hat{U}_{f,n}$ in spectral norm.   Define $S_{f,M}=\sum_{j=0}^{M-1}a_j^fw_j$ as the $M$-Walsh Series of $f$,  with $a_j^f=\frac{1}{M}\sum_{x=0}^{M-1}f(x/M)w_{j}(x/M)$ the $j$-th Walsh coefficients associated to the function $f$ and $M=2^m$. The $M$-Walsh Series $S_{f,M}$ of $f$ is a piecewise function taking at most $M$ different values and defined by $\forall x \in [k/M, (k+1)/M[$, $S_{f,M}(x)=S_{f,M}(k/2^m)=f(k/2^m)$ with $k\in\{0,...,M-1\}$ (see Lemma 1.1 of \cite{PhysRevA.109.042401}). Thus, the $n$-qubit operator $\hat{S}_{f,M}=\sum_{x=0}^{N-1}S_{f,M}(x/N)\ket{x}\bra{x}$ has at most $M$ distinct eigenvalues and can be rewritten as: 

\begin{equation}
\begin{split}
    \hat{S}_{f,M}&=\sum_{x_0,...,x_{n-1}=0}^1S_{f,M}\left(\sum_{i=0}^{n-1} x_i/ 2^{i+1}\right)\ket{x_0,...,x_{n-1}}\bra{x_0,...,x_{n-1}}\\
    &=\sum_{x_0,...,x_{m-1}=0}^1S_{f,M}\left(\sum_{i=0}^{m-1} x_i/ 2^{i+1}\right)\ket{x_0,...,x_{m-1}}\bra{x_0,...,x_{m-1}} \\& \otimes \sum_{x_m,...,x_{n-1}=0}^1\ket{x_m,...,x_{n-1}}\bra{x_m,...,x_{n-1}} 
    \\ &= \hat{f}_m \otimes  \hat{I}_2^{\otimes n-m},
\end{split}
\end{equation}
where  $\hat{f}_m=\sum_{x=0}^{M-1}f(x/M)\ket{x}\bra{x}$ is an $m$-qubit operator.

Thus, $\hat{U}_{f,m} \otimes \hat{I}_2^{\otimes n-m}=e^{i\hat{f}_m\otimes \hat{I}_2^{\otimes n-m}}=e^{i\hat{S}_{f,M}}=\hat{U}_{S_{f,M},n}$, which is the diagonal unitary with eigenvalues associated to the $M$-Walsh series of $f$. Furthermore, $M$-Walsh Series approximate the function itself up to an error depending on the maximum value of the derivative of $f$ and decreasing linearly with $M$ \cite{walsh1923closed,welch2014efficient}: $||f-S_{f,M}||_{\infty,[0,1]}\leq \frac{||f'||_{\infty,[0,1]}}{2^m}$. In terms of spectral norm $||.||_2$, this implies that $\hat{U}_{S_{M,f},n}$ is an approximation of the target diagonal unitary $\hat{U}_{f,n}$:
\begin{equation}
    ||\hat{U}_{f,n}-\hat{U}_{S_{f,M},n}||_{2}\leq \frac{||f'||_{\infty,[0,1]}}{2^m}.
\end{equation}
Finally, one can choose  $m=\lceil \log_2(||f'||_{\infty,[0,1]}/\epsilon) \rceil$ of qubits to implement the target unitary $\hat{U}_{f,n}$ up to a given error $\epsilon>0$ in spectral norm with a quantum circuit of size $s(m)$, depth $d(m)$ and width $w(m)$.
\end{proof}

In the particular case $m\geq n$, one implements exactly the diagonal unitary on $n$ qubits and the error vanishes. On the contrary, for a given $\epsilon>0$, this method gives a quantum circuit which approximates diagonal unitaries with complexities asymptotically independent of the number of qubits $n$.  

A second way to approximate efficiently a diagonal unitary defined through a given function $f$ is to consider a sparse Walsh series.
 For instance, by retaining only the largest coefficients in the $M$-Walsh Series, one can reach a given infidelity with quantum circuits 
 of smaller size and depth \cite{PhysRevA.106.022414}. More generally, the problem of finding the “best” Walsh series with the smallest number of terms approximating up to a given $\epsilon>0$ a function $f$ is called the minimax problem \cite{yuen1975function}. While the minimax problem has no known canonical solutions, sparse Walsh series appear numerically as the most efficient method to approximate smooth, at least differentiable, functions. 

Numerical results illustrating the trade-off between depth and accuracy are given in the application section, for the quantum state preparation of Gaussian states, see Fig. \ref{fig: tradeoff_time_accuracy}. The trade-off is presented for the approximate method using an $M$-Walsh series and for a sparse Walsh-Series.

\section{Implementation of non-unitary diagonal operators}\label{sec:non_unitary}

In this section, we first outline the process for implementing a non-unitary diagonal operator with real eigenvalues using two diagonal unitaries and one ancilla qubit. Then, we adapt the previously introduced methods to develop techniques for implementing non-unitary diagonal operators with adjustable-depth quantum circuits, applicable to both real and complex eigenvalues.

Consider the non-unitary diagonal operator $\hat{D}=\text{diag}(d_0,...,d_{N-1})$, $N=2^n$ with real eigenvalues $\{d_x\}_{x=0}^{N-1}$. The block encoding approach consists in implementing $\hat{D}$ in a larger unitary operator $\hat{U}$ using additional ancilla qubits as: 
\begin{equation}
    \hat{U}_D=\begin{pmatrix} \hat{D} & * \\ * & *
    \end{pmatrix},
\end{equation}
where the other block “$*$” are chosen to ensure that $\hat{U}_b$ is unitary.

More formally, $\hat{U}_D$ is said to be a $(\alpha,m,\epsilon)$-block encoding of $\hat{D}$ with $\alpha >0$ and $\epsilon>0$ if 
\begin{equation}
    \|\frac{1}{\alpha}\hat{D}-(\bra{0}^{\otimes m}\otimes \hat{I}_N ) \hat{U}_D (\ket{0}^{\otimes m}\otimes \hat{I}_N) \|_2 \leq \epsilon,
\end{equation}
in spectral norm. Note that $\alpha$ is often the largest eigenvalue of $\hat{D}$ and $m$ is the number of additional qubits used to block-encode $\hat{D}$ into a unitary operator. In the case where the block-encoding is exact, we note that $\hat{U}_D$ is a $(\alpha,m)$-block encoding of $\hat{D}$.

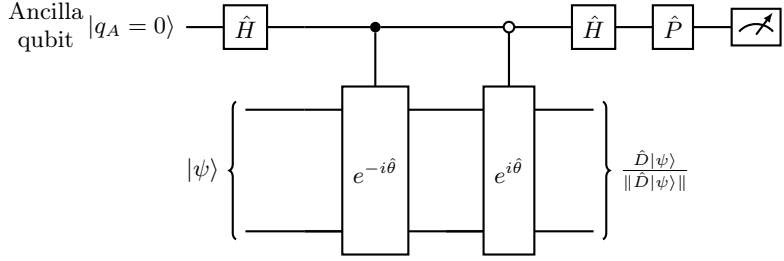
\begin{figure}[ht]
    \centering
\begin{quantikz}
  \lstick[wires=1]{$\begin{matrix} \text{Ancilla} \\ \text{qubit} \end{matrix} \ket{q_A=0}$}& \gate{\hat{H}} & \qw 
&\ctrl{1} &\qw& \octrl{1}& \gate{\hat{H}}& \gate{\hat{P}} & \meter{} 
\\
&\lstick[wires=3]{$\ket{\psi}$} 
   &  \qw
&\gate[3,nwires={2}]{e^{-i\hat{\theta}}} & \qw &\gate[3,nwires={2}]{e^{i\hat{\theta}}}& \qw \rstick[wires=3]{$\frac{\hat{D}\ket{\psi}}{\|\hat{D}\ket{\psi}\|}$} &&
\\ & & && &&&&&&
\\ &&\qw & \qw &\qw&\qw&\qw&& 
\end{quantikz} 
\caption{Quantum circuit to block-encode a non-unitary diagonal operator $\hat{D}=\text{diag}(d_0,...,d_{2^n-1})$ using $\hat{\theta}=\arcsin(\hat{D}/d_{\max})$, with $d_{\max}=\max_{i=0,...,2^n-1}|d_i|$.}

\label{quantum circuit scheme for non-unitary diagonal}
\end{figure}

Figure \ref{quantum circuit scheme for non-unitary diagonal} depicts the sequence of unitary operations needed to block-encode a non-unitary diagonal operator using two controlled diagonal unitaries and one ancilla qubit without using parallelization protocols. From a qubit state $\ket{\psi}$ and an ancilla qubit $\ket{0}$, one obtains:

\begin{equation}
\begin{split}
     \ket{\psi}\ket{0}\xrightarrow{\hat{H}} \ket{\psi}\frac{\ket{0}+\ket{1}}{\sqrt{2}} & \xrightarrow{e^{i\hat{\theta}\otimes\hat{Z}}} \frac{e^{i\hat{\theta}}\ket{\psi}\ket{0}+e^{-i\hat{\theta}}\ket{\psi}\ket{1}}{\sqrt{2}} \\& \xrightarrow{\hat{P}\hat{H}} -i\cos(\hat{\theta})\ket{\psi}\ket{0}+\sin(\hat{\theta})\ket{\psi}\ket{1}.
\end{split}
\end{equation} 

The gate $\hat{P}=\begin{pmatrix} 1 & 0 \\ 0 & -i \end{pmatrix}$ is optional and only serves to suppress a factor $i$ in the final state. The operator $\hat{\theta}$ encodes the eigenvalues of $\hat{D}$ as $\hat{\theta}=\arcsin(\hat{D}/d_{\max})$, where $d_{\max}=\max_{x=0,...,N-1}|d_x|$ is the largest eigenvalue of $\hat{D}$. Without measuring the ancilla, this quantum circuit is an exact $(d_{\max},1)$ block-encoding of $\hat{D}$. This block-encoding is sufficient for some applications but one may need to output only the state proportionate to $\hat{D}\ket{\psi}$ by measuring the ancilla qubit in state $\ket{1}$:
 \begin{equation}\label{eq:final_state_qsp}
\ket{\phi}=-i\sqrt{\hat{I}-\left(\frac{\hat{D}}{d_{\max}}\right)^2}\ket{\psi}\ket{0}+\frac{\hat{D}}{d_{\max}}\ket{\psi}\ket{1}  \xrightarrow{\text{if } \ket{q_A}=\ket{1}} \frac{\hat{D}\ket{\psi} }{\| \hat{D}\ket{\psi}\|},
\end{equation}
with a probability of success given by
\begin{equation}
\mathbb{P}(1)=\| \frac{\hat{D}}{d_{\max}}\ket{\psi}\|^2=|\bra{\psi}(\frac{\hat{D}}{d_{\max}})^2\ket{\psi}|=\sum_{x=0}^{2^n-1}d_x^2|\psi_x|^2/d_{\max}^2.
\label{eq: proba success}
\end{equation}
The specific structures of $\hat{D}$ and $\ket{\psi}$ directly affect the probability of success. When the diagonal operator $\hat{D}$ and the qubit state $\ket{\psi}$ depend on continuous real-valued functions $f$ and $g$ defined on $[0,1]$, {\sl i.e.}, $\hat{D}=\sum_{x=0}^{N-1} f(x/N)\ket{x}\bra{x}$ and $\ket{\psi}=\frac{1}{||g||_{2,N}}\sum_{x=0}^{N-1} g(x/N)\ket{x}$, the probability of success tends to an $n$-independent limit, enabling an efficient implementation:
\begin{equation}
\mathbb{P}(1)=\frac{\|fg\|_{2,[0,1]}^2}{\|g\|_{2,[0,1]}^2\|f\|_{\infty}^2}=\Theta(1),
\label{eq: proba of success continue}
\end{equation}
with $\|fg\|_{2,[0,1]}^2=\int_0^1f(x)^2g(x)^2\textrm{d}x$, $\|g\|_{2,[0,1]}^2=\int_0^1g(x)^2\textrm{d}x$ and $\|f\|_{\infty}=\sup_{x\in [0,1]}|f(x)|$. The notation $\mathbb{P}(1)=\Theta(1)$ means that the probability of success is bounded by two constants $0<C_1\leq \mathbb{P}(1)\leq C_2 \leq1$ in the large $n$ limit.
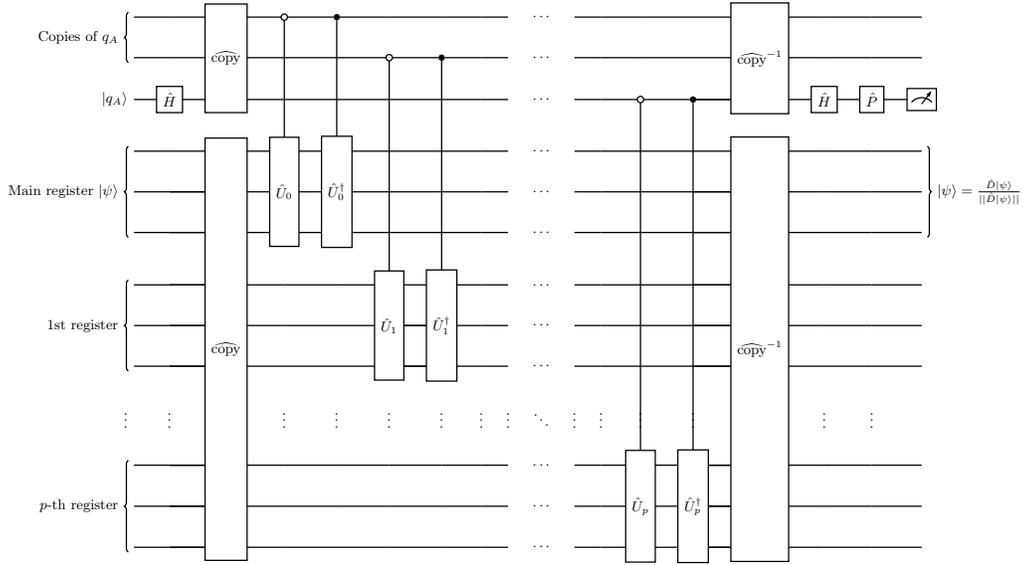
\begin{figure}
\scalebox{0.6}{
\begin{quantikz}
    \lstick[2]{Copies of $q_A$} & \qw & \gate[3]{ \widehat{\text{copy}}} & \octrl{4} & \ctrl{4} & \qw & \qw & \qw & \qw & \cdots & & \qw & \qw & \qw & \gate[3]{ \widehat{\text{copy}}^{-1}} & \qw & \qw & \qw \\
     & \qw & & \qw & \qw & \octrl{6} & \ctrl{6} & \qw & \qw & \cdots & & \qw & \qw & \qw & & \qw & \qw & \qw \\
    \lstick{$\ket{q_A}$} & \gate{\hat{H}} & & \qw & \qw & \qw & \qw & \qw & \qw & \cdots & & \qw & \octrl{8} & \ctrl{8} & \qw & \gate{\hat{H}} & \gate{\hat{P}} & \meter{} &  \\
    \lstick[3]{Main register $\ket{\psi}$} & \qw & \gate[10,nwires={7}]{ \widehat{\text{copy}}} & \gate[3]{\hat{U}_0} & \gate[3]{\hat{U}_0^\dagger} & \qw & \qw & \qw & \qw & \cdots & & \qw & \qw & \qw & \gate[10,nwires={7}]{ \widehat{\text{copy}}^{-1}} & \qw & \qw & \rstick[3]{$\ket{\psi}=\frac{\hat{D}\ket{\psi}}{||\hat{D}\ket{\psi}||}$}\qw \\
    & \qw & \qw & & & \qw & \qw & \qw & \qw & \cdots & & \qw & \qw & \qw & \qw & \qw & \qw & \qw \\
    & \qw & \qw & & & \qw & \qw & \qw & \qw & \cdots & & \qw & \qw & \qw & \qw & \qw & \qw & \qw \\
    \lstick[3]{1st register} & \qw & \qw & \qw & \qw & \gate[3]{\hat{U}_1} & \gate[3]{\hat{U}_1^\dagger} & \qw & \qw & \cdots & & \qw & \qw & \qw & \qw & \qw & \qw & \qw \\
    & \qw & \qw & \qw & \qw & & \qw & \qw & \qw & \cdots & & \qw & \qw & \qw & \qw & \qw & \qw & \qw \\
    & \qw & \qw & \qw & \qw & & \qw & \qw & \qw & \cdots & & \qw & \qw & \qw & \qw & \qw & \qw & \qw \\
    \lstick{$\vdots$} & \vdots & \vdots & \vdots & \vdots & \vdots & \vdots & \vdots & \vdots & \ddots & \vdots & \vdots & \vdots & \vdots & \vdots & \vdots & \vdots \\
    \lstick[3]{$p$-th register} & \qw & \qw & \qw & \qw & \qw & \qw & \qw & \qw & \cdots & & \qw & \gate[3]{\hat{U}_{p}} & \gate[3]{\hat{U}_{p}^\dagger} & \qw & \qw & \qw & \qw \\
    & \qw & \qw & \qw & \qw & \qw & \qw & \qw & \qw & \cdots & & \qw & & & \qw & \qw & \qw & \qw \\
    & \qw & \qw & \qw & \qw & \qw & \qw & \qw & \qw & \cdots & & \qw & & & \qw & \qw & \qw & \qw 
\end{quantikz} 
}
\caption{Adjustable-depth quantum circuit for the implementation of non-unitary diagonal operator $\hat{D}$. The main register $\ket{\psi}$ is composed of $n$ qubits. The ancilla qubit is $\ket{q_A}$. The upper register is composed of ancilla qubits that store the bit value of $\ket{q_A}$, it is used to perform several controls simultaneously. The $p$ ancilla registers below the main register are used to parallelize the implementation of $\hat{U}$. If the output of the measurement of $\ket{q_A}$ is 0, the target state has been successfuly implemented.}
\label{Adjustable depth framework for non unitary diagonal operator}
\end{figure}

When the diagonal operator $\hat{D}$ is $s$-sparse, i.e., $\hat{D}=\sum_{x \in S} d_x \ket{x}\bra{x}$, with $S \subseteq \{0,1,...,N-1\}$ and $|S|=s$, and the qubit state depends on a continuous function, {\sl i.e.}, $\ket{\psi}=\frac{1}{||g||_{2,N}}\sum_{x=0}^{N-1} g(x/N)\ket{x}$, the probability of success diminishes exponentially with the number of qubits:
\begin{equation}
    \mathbb{P}(1)=\frac{1}{||g||_{2,N}}\sum_{x \in S} \frac{d_x^2}{d_{\max}^2}g(x/N)^2=\Theta(1/N),
\label{eq: proba success sparse}
\end{equation}
with $N=2^n$.
In the other cases, for instance for a sparse qubit state and a non-sparse diagonal operator, no general formula garantees a significant probability of success in the large $n$ limit. The probability of success can also be improved using amplitude amplification as discussed in Section \ref{section amplitude amplification}.

Now, consider the problem of implementing a non-unitary diagonal operator $\hat{D}$ using an adjustable-depth quantum circuit associated to a Walsh or a sequential decomposition of $\hat{U}=e^{ i \arcsin(\hat{D}/d_{\max})}=\prod_{j=0}^{p-1}\hat{U}_j$.  In Fig. \ref{quantum circuit scheme for non-unitary diagonal}, the controls of $\hat{U}$ and $\hat{U}^{\dagger}$ do not preserve the adjustable-depth feature of the methods presented in the previous sections since one ancilla qubit cannot control several blocks at the same time. In order to overcome this issue, one can use $p$ copies of the ancilla qubit $\ket{q_A}$ to control $p$ blocks in parallel at the same time. That is why, one needs to prepare a copy register of $\ket{q_A}$ before the control of the diagonal operators. The following theorem states and proves that Fig. \ref{Adjustable depth framework for non unitary diagonal operator} represents an adjustable-depth quantum circuit for a block-encoding of $\hat{D}$ up to the measurement of the ancilla qubit $\ket{q_A}$.

\begin{theorem}[Block-encoding of $\hat{D}$]
\label{thm: block-encoding}
Let $\hat{D}$ be a non-unitary diagonal operator with real eigenvalues and $\hat{U}=e^{ i \arcsin(\hat{D}/(\alpha d_{\max}))}$ with $\alpha \geq 1$. Suppose there is a decomposition $\tilde{U}=\prod_{j=0}^{p-1} \hat{U}_j$ approximating $\hat{U}$ up to an error $\epsilon>0$ in spectral norm, where each $\hat{U}_j$ is a diagonal unitary acting on $k_j\leq n$ qubits. Then, the quantum circuit defined Fig. \ref{Adjustable depth framework for non unitary diagonal operator}\footnote{without the measurement.} is an $(\alpha d_{\max}, k+p, \epsilon)$-block-encoding of $\hat{D}$, where $k=\sum_{j=1}^{p-1}k_j$.\end{theorem}

The proof is presented in Appendix \ref{proof thm block encoding}.

In terms of complexity, copying the ancilla qubit $\ket{q_A}$ always requires less controlled-not gates than copying the main register.  Additionally, controlling an operator $\hat{U}_j$ has the same asymptotic cost than implementing $\hat{U}_j$: remark that for any operator of the form $\hat{V}^{-1}\hat{U}\hat{V}$, the controlled operation is given by $C(\hat{V}^{-1}\hat{U}\hat{V})=\hat{V}^{-1}C(\hat{U})\hat{V}$. Therefore, if $\hat{U}_j$ is a sequential operator given by Eq. \eqref{sequential decomposition}, one only needs to transform the $(n-1)$-controlled phase gate into a $n$-controlled phase gate and if $\hat{U_j}$ is a Walsh-Hadamard operator given by Eq. \eqref{eq: exp of walsh operator}, one only needs to control the $\hat{R}_{Z}$ gates. As a consequence, the asymptotic cost in terms of size and depth of implementing a diagonal operator $\hat{D}$ is directly given by the cost of implementing the unitary $e^{ i \arcsin(\hat{D}/d_{\max})}$ up to the same error.

When the non-unitary diagonal operator $\hat{D}$ has complex eigenvalues, one can separate the modulus part and the phase part by writing $\hat{D}=e^{i\arg(\hat{D})}\widehat{|D|}$, where $\arg{\hat{D}}$ is the diagonal operator encoding the argument of each eigenvalue and $\widehat{|D|}$ encodes the modulus. The operator  $e^{i\arg(\hat{D})}$ is unitary, while $\widehat{|D|}$ has real eigenvalues and can be block-encoded using Theorem \ref{thm: block-encoding} through the diagonal unitaries $e^{\pm i \arcsin(\widehat{|D|}/(\alpha d_{\max}))}$. It is possible to implement the phase operator with one method and the modulus part with another if relevant. It is also possible to implement directly the control $e^{i (\arg(\hat{D}) \pm \arcsin(\widehat{|D|}/(\alpha d_{\max})))}$ to avoid implementing the phase operator and the modulus separately.

\section{Efficient-adjustable-depth methods}\label{sec:adj_depth_methods}

In this section, we combine the previous theorems to obtain the overall complexities for the efficient-adjustable-depth quantum circuits for unitary and non-unitary diagonal operators without exponential scalings. 

\subsection{Unitary and non-unitary diagonal operators defined in terms of smooth functions}

Consider a diagonal unitary $\hat{U}=\sum_{x=0}^{N-1} e^{if(x/N)}\ket{x}\bra{x}$ defined in terms of a smooth, at least differentiable, function $f$. For each decomposition of $\hat{U}$, sequential or Walsh-Hadamard, the combination of the adjustable-depth Theorem \ref{thm : adjustable-depth with ancilla} or the fully parallelized Theorem \ref{thm : full parallelization} gives three possibilities: without ancilla, adjustable-depth with a given number $p>0$ of ancilla and fully parallelized. Then, the efficiency is achieved by approximating the diagonal unitary up to an error $\epsilon>0$ using Theorem \ref{thm: Approximation theorem} and the associated method.

For a non-unitary diagonal operator $\hat{D}=\sum_{x=0}^{N-1} f(x/N)\ket{x}\bra{x}$ defined in terms of a smooth, at least differentiable, function $f$, the previous section guarantees that the asymptotic cost of implementing $\hat{D}$ is given by the cost of implementing $\hat{U}=e^{i \arcsin(\hat{D}/(\alpha d_{\max}))}$. The parameter $\alpha>1$ has been introduced to ensure that the function $\arcsin(f/(\alpha \|f\|_{\infty}))$ is differentiable\footnote{ $\alpha$ can be taken constant equal to $1.1$ with the only backward of reducing the probability of success by a factor $\alpha^2$.} and that the approximation Theorem \ref{thm: Approximation theorem} can be applied.  

The complexities for each of these methods are presented in Table \ref{Table : approximate case} and proven in Appendix \ref{sec:approximation_methods}.

\begin{table*}[ht]
    \centering
\begin{tabular}{ |p{7.5cm}||p{2.7cm}|p{2.5cm}|p{2.8cm}|}
 \hline
 \multicolumn{4}{|c|}{Efficient quantum circuit for diagonal operators depending on differentiable functions} \\
 \hline
 Method & Depth & Size & Ancilla \\
 \hline

 Sequential (\ref{cor : Sequential no ancilla : approximé}, \ref{cor:  approximate sequential decomposition with one ancilla}, \ref{cor: approximate block-encoding using a sequential decomposition and one ancilla}, \ref{cor: approximate  block-encoding sequential decomposition  with two ancilla})  & $\tilde{\mathcal O}(1/\epsilon)$ & $\tilde{\mathcal O}(1/\epsilon)$ & $0$ or $1$ or $2$\\
 Sequential adjustable-depth (\ref{cor: partially parallelized approximate sequential decomposition}, \ref{cor: partially parallelized approximate sequential decomposition bis}, \ref{corollary: partially parallelized block-encoding with approximate sequential decomposition}, \ref{cor: partially parallelized approximate block-encoding sequential decomposition using polylog MCphase}) &   $\tilde{\mathcal O}(\frac{1}{\epsilon m}+\log(m))$  & $\tilde{\mathcal O}(1/\epsilon+m)$& $m=\Omega(\log(1/\epsilon))$\\
 Sequential fully parallelized (\ref{lemma: fully parallelized approximate sequential decomposition}, \ref{cor: Fully parallelized block-encoding with approximate sequential decomposition}) & $\mathcal{O}(\log(1/\epsilon))$ & $\tilde{\mathcal O}(1/\epsilon)$ & $\tilde{O}(1/\epsilon)$\\
\hline
 
Walsh-Hadamard (\cite{welch2014efficient},\ref{cor: approximate walsh-hadamard decomposition},\ref{cor: approximate walsh-hadamard block-encoding with one ancilla}) & $\mathcal{O}(1/\epsilon)$ & $\mathcal{O}(1/\epsilon)$ & $0$ or $1$ \\
 Walsh-Hadamard adjustable-depth (\ref{cor: partially parallelized approximate walsh-hadamard decomposition with one ancilla}, \ref{Partially parallelized block-encoding with approximate Walsh-Hadamard decomposition})   &   $\tilde{\mathcal O}(\frac{1}{\epsilon m}+\log(m))$  & $\tilde{\mathcal O}(1/\epsilon+m)$& $m=\Omega(\log(1/\epsilon))$ \\
 Walsh-Hadamard fully parallelized (\ref{lemma: fully parallelized approximate walsh-hadamard decomposition},\ref{lemma: fully parallelized approximate Walsh-Hadamard decomposition})  & $\mathcal{O}(\log(1/\epsilon))$ & $\tilde{\mathcal O}(1/\epsilon)$ & $\tilde{\mathcal O}(1/\epsilon)$\\
\hline
 
   Walsh-recursive  (\ref{cor: lemma11+approximate}, \ref{walsh-recursive nonunitary})  & $\mathcal{O}(\frac{1}{\epsilon\log(1/\epsilon)})$ & $\mathcal{O}(1/\epsilon)$&  $0$ or $\mathcal{O}(\log(1/\epsilon))$ \\
    Walsh-optimized adjustable-depth (\ref{cor:lemma20+approximate}, \ref{Walsh-optimized adjustable-depth non-unitary})  & $\mathcal{O}(\frac{1}{\epsilon m}+\log m)$ & $\tilde{\mathcal O}(1/\epsilon+m)$&  $ m=\Omega(\log(1/\epsilon)) $ \\
    Walsh-optimized fully parallelized (\ref{cor: lemma20 max qubit+approximate}, \ref{Walsh-optimized fully parallelized non-unitary}) & $\mathcal{O}(\log(1/\epsilon))$ & $\mathcal{O}(1/\epsilon)$&  $\mathcal{O}(\frac{1}{\epsilon\log(1/\epsilon)})$ \\
  \hline
 Fourier-GQSP (\ref{cor:GQSP approximate}) & $\mathcal{O}(n\log(1/\epsilon))$ & $\mathcal{O}(n\log(1/\epsilon))$ & $1$\\
 Fourier-GQSP with ancilla (\ref{cor:GQSP approximate parallel}) & $\mathcal{O}(\log(n)\log(1/\epsilon))$ & $\mathcal{O}(n\log(1/\epsilon))$ & $n$\\
\hline
\end{tabular}  
\caption{Table of depth, size and $m$-ancilla qubits for different efficient quantum circuits to implement, up to an error $\epsilon>0$ in spectral norm, $n$-qubit unitary and non-unitary diagonal operators associated to real-valued function $f$ with bounded first derivative. All the scalings are proven in Appendix \ref{For diagonal unitaries depending on differentiable functions} for diagonal unitaries and in Appendix \ref{sec: non-unitary diagonal operators depending on differentiable functions} for non-unitary diagonal operators. The notation $\tilde{\mathcal O}(\cdot)$ stands for $\mathcal{O}(\cdot)$ up to a polylogarithmic factor.}
\label{Table : approximate case}
\end{table*}

The complexities presented in the first  three lines of Table \ref{Table : approximate case} are derived from the exact implementation of diagonal operators using the sequential decomposition and from the approximation Theorem \ref{thm: Approximation theorem} (more details are given in Appendix \ref{sec:exact methods}). The number of ancilla qubits depends on the parallelization method (no parallelization, partial parallelization or full parallelization), the unitarity or non-unitarity of the diagonal operator (with one additional ancilla required for a non-unitary operator), and the method for implementing multi-controlled gates (where more ancilla simplify the implementation).  

The next three lines present the complexity scaling for the Walsh-Hadamard decomposition of diagonal operators. These scalings are derived from the exact methods described in Appendix \ref{sec:exact methods} and from the approximation Theorem \ref{thm: Approximation theorem}.  The number of ancilla qubits depends on the parallelization method (no parallelization, partial parallelization, or full parallelization), and on whether the diagonal operator is unitary or non-unitary. 

The next three lines of Table \ref{Table : approximate case} present the Walsh-recursive and the Walsh-optimized methods. These scalings are derived from the exact methods developed in \cite{sun2023asymptotically}. When no ancilla qubits are available, the authors of \cite{sun2023asymptotically} propose a recursive scheme for diagonal unitaries with an optimal depth complexity of $\mathcal{O}(2^n/n)$. However, the size is not optimal and is bounded by $2^{n+3}+\mathcal{O}(\log(n^2/\log(n)))$ (Lemma 24 of \cite{sun2023asymptotically}). A second method uses a number $2n \leq m \leq 2^n/n$ of ancilla qubits to parallelize the Walsh operators in an order which optimizes the number of CNOT gates, achieving a depth $\mathcal{O}(2^n/m+\log(m))$ and a size $\mathcal{O}(2^n)$ (Lemma 10 in \cite{sun2023asymptotically}). These exact methods, together with the approximation Theorem \ref{thm: Approximation theorem} and the block-encoding Theorem \ref{thm: block-encoding}, lead to the scalings presented in table \ref{Table : approximate case}.

The last two lines of Table \ref{Table : approximate case} use the Fourier-GQSP method developed in the recent article \cite{motlagh2024generalized}. These methods rely on a quantum signal processing protocol to implement diagonal operators through their Fourier expansions. From single-qubit rotations and controlled-diagonal operators based on exponential of linear functions $\hat{U}_{\omega}=\sum_{x=0}^{N-1}e^{i\frac{2\pi x}{N}}\ket{x}\bra{x}$, one can construct a polynomial $P$ of $\hat{U}_\omega$ that approximates the Fourier decomposition of the target diagonal operator $\hat{D}\simeq P(\hat{U}_\omega)=\sum_{j=-p}^p c_j \hat{U}_{\omega}^j$, where $c_j$ is the associated $j$-th Fourier coefficient.  That article proves the existence of quantum circuits to implement Fourier series of diagonal operators. However, achieving a depth that scales logarithmically with the inverse of the error requires the Fourier series to converge exponentially fast. For this purpose, the diagonal operator $\hat{D}$ has to depend on a periodic analytic function $f$ with the property that the $L^1$-norms of its $p$-th derivatives behave asymptotically as  $\|\partial^{(p)}f\|_{L^1}=\mathcal{O}( p^p K^{-p})$ where $K$ is a real number strictly greater than one (see corollary \ref{cor:GQSP approximate} and corollary \ref{cor:GQSP approximate parallel} for details). Moreover, constructing the quantum circuit requires classical optimization to determine the angles of the single-qubit rotations. This step is resource-intensive and its convergence towards an accurate solution is not guaranteed. Additionally, this protocol uses one ancilla qubit, even when the diagonal operator is unitary. The Fourier-GQSP method with ancilla (last line of Table \ref{Table : approximate case}) combines the Fourier-GQSP method with the "copy" unitary method presented in \ref{copy lemma}: using $n$ additional ancilla qubits, one can reduce the depth of the Fourier-GQSP method by a factor $n$ by creating copies of the initial ancilla qubit to parallelize the control of the phase gates associated with the control $\hat{U}_\omega$ operator (more details in corollary \ref{cor:GQSP approximate parallel}).

\subsection{Sparse unitary and non-unitary diagonal operators}

Sparsity allows diagonal operators efficiently implemented. Sparse diagonal operators either possess $s$ non-trivial eigenvalues, i.e. $s$ eigenvalues different from one for diagonal unitaries and $s$ eigenvalues different from zero for non-unitary diagonal operators, or consist of a product of $s$ Walsh operators $\hat{W}_j$ defined in Eq. \eqref{walsh decomposition}. The methods for implementing sparse diagonal operators and their associated scalings are presented in Table \ref{Table : sparse case}. The first five lines consider diagonal operators with $s$ non-trivial eigenvalues. The scalings are directly given by the chosen method for the multi-control gates as explained in Appendix \ref{sec: sparse diagonal operator}. The first, second and third lines correspond directly to $s$ times one multi-control phase gate. The fourth and fifth lines are derived from the adjustable-depth Theorem \ref{thm : adjustable-depth with ancilla} and the full parallelization Theorem \ref{thm : full parallelization}. To implement a sparse non-unitary diagonal operator, one requires an additional ancilla qubit.

The last three lines describes the complexity scalings for diagonal operators that consist of a product of $s$ Walsh-Hadamard operators. The cost of implementation depends directly on a constant $k$, which represents the maximum number of ones in the binary decomposition of the order of the Walsh operators. This constant is at most $n$, indicating that the scalings are double-sparse: sparse in the number $s$ of operators and sparse in the number of ones $k$. Each of the three lines corresponds to a parallelization method: no parallelization, partial parallelization and full parallelization. The number of ancilla qubits also depends on the parallelization method and whether the diagonal operator is unitary or non-unitary.


\begin{table*}[ht]
    \centering
\begin{tabular}{ |p{7cm}||p{3.9cm}|p{3.5cm}|p{1.7cm}|}
 \hline
 \multicolumn{4}{|c|}{Efficient quantum circuit for sparse diagonal operators} \\
 \hline
 Method & Depth & Size & Ancilla \\
 \hline

 Sequential no ancilla exact (\ref{lemma : Sequential no ancilla : Sparse}, \ref{lemma : Sequential block-encoding no ancilla : Sparse}) & $\mathcal O (ns)$ & $\mathcal O (ns)$ & 0 or 1 \\
 Sequential no ancilla approximated  (\ref{lemma : Sequential no ancilla : Sparse approxime}, \ref{lemma : block-encoding Sequential no ancilla : Sparse approxime}) & $\mathcal O(s\log^3(n)\log(\frac{s}{\epsilon}))$ & $\mathcal O(ns\log^4(n)\log(\frac{s}{\epsilon}))$ & 0 or 1 \\
 Sequential one ancilla exact  (\ref{lemma : Sequential one ancilla : Sparse}, \ref{lemma : block-encoding Sequential one ancilla : Sparse}) & $\mathcal O (s\log^3(n))$ & $\mathcal O (ns\log^4(n))$ & 1 or 2\\
 Sequential Adjustable-depth  (\ref{cor : Sequential adjustable-depth : Sparse}, \ref{cor : block-encoding Sequential adjustable-depth : Sparse}) &   $\mathcal{O}(\frac{s\log^3(n)}{ (m/n)}+\log(m/n))$  & $\mathcal{O}(ns\log^4(n)+m)$& $m =\Omega(n)$\\
 Sequential fully parallelized  (\ref{cor : Sequential fully parallelized : Sparse}, \ref{cor : block-encoding Sequential fully parallelized : Sparse}) & $\mathcal{O}(\log(s)+\log^3(n))$ & $\mathcal{O}(ns)$ & $\mathcal{O}(ns)$\\
\hline
 
 Walsh no ancilla  (\ref{lemma : walsh no ancilla : Sparse}, \ref{lemma : block-encoding walsh no ancilla : Sparse})  & $\mathcal{O}(sk)$ & $\mathcal{O}(sk)$ & 0 or 1\\
 Walsh Adjustable-depth  (\ref{cor : walsh adjustable-depth : Sparse}, \ref{cor : block-encoded walsh adjustable-depth : Sparse})  &   $\mathcal{O}(\frac{sk}{(m/k)}+\log(m/k))$  & $\mathcal{O}(sk+m)$& $m =\Omega(k)$\\
 Walsh fully parallelized  (\ref{cor : walsh fully parallelized : Sparse}, \ref{cor : block-encoded walsh fully parallelized : Sparse}) & $\mathcal{O}(k+\log(s))$ & $\mathcal{O}(sk)$ & $\mathcal{O}(sk)$\\
\hline
 
  \hline
\end{tabular}  
\caption{Table of depth, size and $m$-ancilla qubits for different quantum circuits implementing unitary or non-unitary $n$-qubit diagonal operators with sparse decompositions. $s$ is the number of sequential or Walsh-Hadamard operators in the decomposition of the diagonal operator. All scalings are proven in Appendix \ref{sparse du} for diagonal unitaries and Appendix \ref{sparse nu diag} for non-unitary diagonal operators.}
\label{Table : sparse case}
\end{table*}

Sparse Walsh-Hadamard decompositions are particularly efficient to implement diagonal unitaries depending on real-valued functions with bounded first derivative see Fig. \ref{fig: space_time_tradeoff} and Fig. \ref{fig: tradeoff_time_accuracy} or for bounded piecewise continuous functions (see Appendix A 3 in \cite{PhysRevA.109.042401}). In particular, to reach a given accuracy, sparse Walsh-Hadamard decompositions are often more efficient than $M$-Walsh series (presented in Table \ref{Table : approximate case}) when the most significant Walsh coefficients are implemented before the others.

\section{Amplitude amplification and non-destructive Repeat-Until-Success protocols}
\label{section amplitude amplification}

Amplitude amplification enables to increase the probability of success of performing a non-unitary diagonal operator $\hat{D}$ on a given state $\ket{\psi}$. When the cost of preparing the state $\ket{\psi}$ from $\ket{0}^{\otimes n}$ is not prohibitively high, one can perform the usual amplitude amplification protocol \cite{Brassard_2002}, otherwise one can perform either the oblivious amplitude amplification protocol \cite{10.1145/2591796.2591854} or a repeat-until-success scheme where, at each failure, the “wrong” state is corrected to become the “good” state up to a new probability of success.

The amplitude amplification protocol increases the probability of success of measuring the ancilla qubit in state $\ket{1}$ by making repeated reflexions composed of the unitary $\hat{U}_\psi$ preparing the state $\ket{\psi}$ from $\ket{0}^{\otimes n}$ and of the unitary $\hat{U}_D$ which block-encodes $\hat{D}$. Considering the final state $\ket{\phi}=\hat{U}_D\ket{\psi}\otimes \ket{0}_{q_A}=\hat{U}_D(\hat{U}_\psi\otimes \hat{I}_2)\ket{0}^{\otimes n} \otimes \ket{0}_{q_A}$ defined in Eq. \eqref{eq:final_state_qsp}, the amplitude amplification protocol is composed of $\hat{U}_{\phi}=\hat{I}-2\ket{\phi}\bra{\phi}$ and $\hat{U}_P=-(\hat{I}-2\hat{P})$ with $\hat{P}$ the projector on the good subspace:  $\hat{P}=\hat{I}_2\otimes\cdots\otimes \hat{I}_2 \otimes \ket{1}_{q_A}\bra{1}_{q_A}$. The unitary $\hat{U}_{\phi}$ can be rewritten as $\hat{U}_{\phi}=\hat{U_D}^{-1}(\hat{U_ \psi}^{-1}\otimes \hat{I}_2)\Lambda_{1,...,n\rightarrow q_A}(-\hat{Z})(\hat{U_\psi}\otimes \hat{I}_2)\hat{U_D}$ with $\Lambda_{1,...,n\rightarrow q_A}(-\hat{Z})$ an anti-controlled minus $\hat{Z}$-Pauli gate and 
$\hat{U}_P=\hat{I}_2\otimes\cdots\otimes \hat{I}_2 \otimes (-\hat{Z})$. The scheme of the quantum circuit associated to a step of amplitude amplification is shown on Fig. \ref{fig:amplitude_amplification}. The number $k$ of steps needed to perform the amplitude amplification is given by $k=\lfloor \pi/(4\beta) \rfloor$ with $\beta=\arcsin(\sqrt{\mathbb{P}(1)})$. The number $k$ can be estimated from $\mathbb{P}(1)$ using one of the formula $(\ref{eq: proba success}, \ref{eq: proba of success continue}, \ref{eq: proba success sparse})$. In the case where the diagonal unitary and the qubit state are associated to continuous functions, the number of amplitude amplification steps needed to reach $\mathbb{P}(1)\simeq 1$ is asymptotically independent of the number of qubits $n$, as demonstrated in Fig. \ref{fig:second} for the quantum state preparation of Gaussian states. The amplitude amplification protocol is illustrated Fig. \ref{fig:third}, also for the quantum state preparation of Gaussian states, proving numerically that the probability of success reaches $\mathbb{P}(1)\simeq 1$ in a number $k=\lfloor \pi/(4\beta) \rfloor$ of steps.

 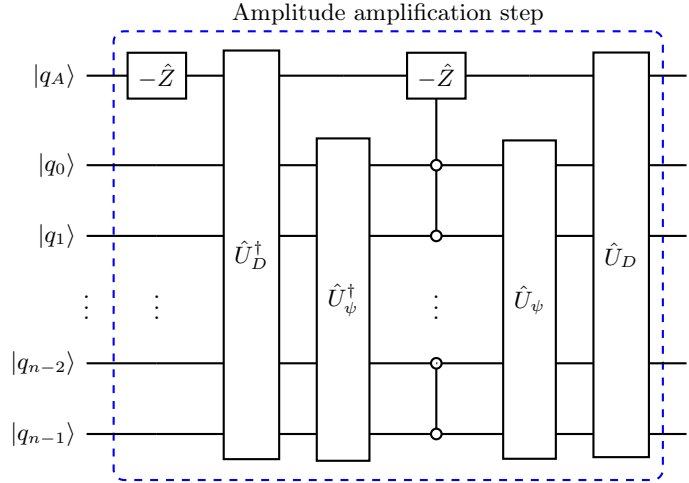
\begin{figure}[ht]
\centering
\begin{quantikz}
  \lstick{$\ket{q_A}$} & \gate{-\hat{Z}}\gategroup[6,steps=6,style={dashed,rounded corners,color=blue,inner xsep=2pt},background]{{Amplitude amplification step}} & \gate[6,nwires={4}]{\hat{U}_D^\dagger} &\qw & \gate{-\hat{Z}} & \qw& \gate[6,nwires={4}]{\hat{U}_D} & \qw \\
  \lstick{$\ket{q_0}$} & \qw &&\gate[5,nwires={3}]{\hat{U}_\psi^\dagger} & \octrl{-1} &\gate[5,nwires={3}]{\hat{U}_\psi}  && \qw \\
  \lstick{$\ket{q_1}$} & \qw &&& \octrl{-1} && & \qw \\
  \vdots & \vdots & \vdots & \vdots & \vdots & \vdots & \vdots \\
  \lstick{$\ket{q_{n-2}}$} & \qw &&& \octrl{0} && & \qw \\
  \lstick{$\ket{q_{n-1}}$} & \qw &&& \octrl{-1} && & \qw
\end{quantikz} 
\caption{Quantum circuit for amplitude amplification. Note that for QSP, operator $\hat{U}_\psi$ is an Hadamard tower on the register $\ket{q}$, {\sl i.e.}, $\hat{U}_\psi= \hat{H}^{\otimes n}$. Operator $\hat{U}_D$ is the block-encoding of the non-unitary diagonal operator $\hat{D}$.}
\label{fig:amplitude_amplification}
\end{figure}

Conversely, a repeat-until-success protocol is relevant when the cost of preparing the state $\ket{\psi}$ from $\ket{0}^{\otimes n}$ is prohibitively high.  Here we propose a non-destructive repeat-until-success protocol: consider the diagonal unitary  $e^{i\arcsin(\hat{D}/(\alpha  d_{\max}))}$ where a constant parameter $\alpha>1$ ensures that the failure operator  $\hat{D}'=\sqrt{\hat{I}-(\hat{D}/(\alpha d_{\max}))^2}$ is invertible. If a failure occurs, one can implement the non-unitary diagonal operator $\hat{D}(\hat{D}'^{-1})$ using the unitary operator  $e^{i\arcsin(\hat{D}(\hat{D}'^{-1})/(\alpha d_{\max}'))}$ with $d_{\max}'=\max_{i=0,...,N-1}|d_i|/(\sqrt{1-(d_i/(\alpha d_{\max}))^2})$. Applying this operator on the state $\frac{\hat{D}'\ket{\psi}}{\| \hat{D}'\ket{\psi} \|}$ leads to the target state $\frac{\hat{D}\ket{\psi}}{\| \hat{D}\ket{\psi} \|}$ with  a probability of success $\mathbb{P}'(1)=\frac{\| \hat{D}\ket{\psi}\|^2 }{\alpha^2 d_{\max}'^2 \| \hat{D'}\ket{\psi} \|^2}$. This protocol can be implemented until a success is reached, avoiding destroying and preparing again the state $\ket{\psi}$. Figure \ref{fig:fourth} illustrates this protocol for the quantum state preparation of Gaussian states, showing that the probability of success after $k$ failures reaches a constant value. Therefore, the average number of failures before a success is finite and can be estimated numerically.
 
\section{Applications}\label{sec:applications}

Unitary and non-unitary diagonal operators are fundamental for many applications ranging from the resolution of partial differential equations modeling phenomena in physics, chemistry, biology and finance \cite{van2004stochastic, freund2000stochastic,oksendal2013stochastic, BLACK1976167,pironneau2009partial} to the loading of classical data on qubit states \cite{sun2023asymptotically, PhysRevA.106.022414, PhysRevA.109.042401}, image processing \cite{gonzalez2009digital} and optimization \cite{farhi2001quantum,10.1145/3569095}. In particular, non-unitary operations have already been mentionned in many different quantum algorithms for quantum simulation \cite{gilyen2019quantum,low2019hamiltonian,seneviratne2024exact,mangin2024efficient} . As an illustration, we apply the methods previously presented to two problems. The first one is the quantum state preparation problem for which we improve the current state of the art of preparing qubit-state depending on continuous, differentiable function $f$ and we provide new methods with space-time-accuracy trade-offs. The second major application is the resolution of partial differential equation where the associated evolution operator is implementable using diagonal operators and the Quantum Fourier Transform. As an example, we solve the diffusion equation, also called heat equation, which possesses a non-unitary evolution.

\subsection{Quantum State Preparation}
\label{sec:qsp}
The quantum state preparation problem consists in loading a set of classical data $\{y_x\}_{x=0}^{2^n-1}$ into a $n$-qubit state $\ket{\psi}=\frac{1}{\mathcal{N}}\sum_{x=0}^{2^n-1}y_x\ket{x}$  up to a normalization factor $\mathcal{N}$. In the following we provide an efficient quantum state preparation algorithm for $n$-qubit states depending on a differentiable function $f$ is defined by:
\begin{equation}
\ket{f}=\frac{1}{||f||_{2,N}}\sum_{x=0}^{N-1} f(x/N)\ket{x},
\label{eq: target state}
\end{equation}
with $||f||_{2,N}=\sqrt{\sum_{x=0}^{N-1} |f(x/N)|^2}$ and $N=2^n$.
Now, consider the non-unitary operator $\hat{D}=\sum_{x=0}^{N-1}f(x/N)\ket{x}\bra{x}$. The quantum state preparation can be performed by implementing $\hat{D}$ with one of the method previously presented Table \ref{Table : approximate case} and applying $\hat{D}$ on the uniform superposition of all states $\ket{s}=\hat{H}^{\otimes n} \ket{0}^{\otimes n} = \frac{1}{\sqrt{N}}\sum_{x=0}^{N-1}\ket{x}$. The following Theorem summarizes this result:

\begin{theorem}[Quantum State Preparation]
\label{thm: qsp}
Let $f$ be a real-valued differentiable function in $L^2([0,1])$, $n$ the number of qubits and $\epsilon>0$. There is a quantum circuit that efficiently prepare the state $\ket{f}$ defined in Eq. \eqref{eq: target state} up to an infidelity $\epsilon>0$ using $p\geq 1$ ancilla qubits with a depth scaling at worst as $\mathcal{O}(1/(p\sqrt{\epsilon})+\log(p))$, size $\mathcal{O}(n+1/(p\sqrt{\epsilon}))$ and a probability of success $\mathbb{P}(1)=\|f\|_{2,[0,1]}^2/\|f\|^2_\infty=\Theta(1)$. 
\end{theorem}

\begin{proof} This theorem is a direct consequence of the scaling given in Table \ref{Table : approximate case}. The associated quantum circuits are the same as the one for the non-unitary diagonal operator. The non-unitary diagonal operator requires to be applied on the uniform superposition $\ket{s}=\frac{1}{\sqrt{N}}\sum_{x=0}^{N-1}\ket{x}$, which is prepared by a one-depth layer of $n$ Hadamard gates from the state $\ket{0}^{\otimes n}$. The infidelity between the target quantum state $\ket{f}$ and the prepared quantum state $\ket*{\tilde{f}}$ is defined as $1-F=1-|\braket*{\tilde{f}}{f}|^2$. If the $2$-norm of the difference between $\ket{f}$ and $\ket*{\tilde{f}}$ is smaller than $\epsilon$, then the associated infidelity is smaller than $\epsilon^2$. 

In particular, if the non-unitary diagonal operator $\hat{D}=\sum_{x=0}^{N-1}f(x/N)\ket{x}\bra{x}$ is approximated by $\tilde{D}$ up to an error $\sqrt{\epsilon}$ in spectral norm, and we define $\ket{f}=\hat{D}\ket{s}/\|\hat{D}\ket{s}\|_{2,N}$ and  $\ket*{\tilde{f}}=\hat{\tilde{D}}\ket{s}/\|\hat{\tilde{D}}\ket{s}\|_{2,N}$. Then, $\| \ket{f}- \ket*{\tilde{f}} \|_{2,N}\leq \frac{2\sqrt{N\epsilon}}{\|f\|_{2,N}}$ such that $1-F\leq \frac{4N\epsilon}{\|f\|_{2,N}^2}$. Note that $\|f\|_{2,N}^2/N$ converges toward $\|f\|_{2,[0,1]}^2$ in the large $N$ limit. Therefore, there exists $N_0\in \mathbb{N}$ such that $\forall N>N_0,  \frac{N}{\|f\|_{2,N}^2}\leq K_1$, implying that $1-F=O(\epsilon)$.  The probability of success is given by Eq. \eqref{eq: proba of success continue} for the uniform function $g=1$ and $\mathbb{P}(1)=\|f\|_{2,[0,1]}^2/\|f\|^2_\infty=\Theta(1)$. 
\end{proof}

In the case of a small probability of success, one can still perform a constant number $k=\lfloor \pi/(4\arcsin (\sqrt{\mathbb{P}(1)})) \rfloor$ of amplitude amplification steps to reach $\mathbb{P}(1)\simeq 1$. 

The quantum state preparation of Gaussian states is performed in Fig. \ref{fig:first} and the probability of success is presented as a function of the number of qubits in Fig. \ref{fig:second}, showing an independent of $n$ limit. The amplitude amplification protocol is also illustrated in Fig. \ref{fig:third}, proving numerically the efficiency of our method. The space-time trade-off is illustrated in Fig. \ref{fig: space_time_tradeoff} for an exact Walsh-Hadamard decomposition, an approximate one using a $M$-Walsh series and using a sparse Walsh series. The trade-off is particularly efficient for a small number of ancilla where one diminishes by several order the depth of the circuits. Note that increasing the number of ancilla qubits, increases the cost of the copy operations as well. As a consequence, for a large number of ancilla qubits, adding some ancilla qubits may not improve anymore the depth of the circuit. The trade-off between time and accuracy is illustrated in Fig. \ref{fig: tradeoff_time_accuracy} for the quantum state preparation of Gaussian states. The infidelity of the prepared states decreases with the number of operators one implements as stated in the approximation Theorem \ref{thm: Approximation theorem}. Furthermore, sparse Walsh-Hadamard decomposition allows to reach a given infidelity with smaller depth quantum circuits than $M$-Walsh Series.

In \cite{PhysRevA.109.042401}, $n$-qubit states are implemented using the Taylor expansion of the operator $\hat{I}-e^{i\hat{f}\epsilon_0} \simeq \hat{f}\epsilon_0$ with $\hat{f}=\sum_{x=0}^{2^n-1}f(x)\ket{x}\bra{x}$. In comparison, the methods presented here implement directly $\hat{f}$ using $e^{\pm i \arcsin(\hat{f}/(\alpha f_{\max}))}$ avoiding an additional error coming from the Taylor expansion. Theorem \ref{thm: qsp} also improves the probability of success which is scaling with $\epsilon$ in \cite{PhysRevA.109.042401} while here it reaches a constant value.

\begin{figure}
\centering
\begin{subfigure}{0.47\textwidth}
    \includegraphics[width=\textwidth]{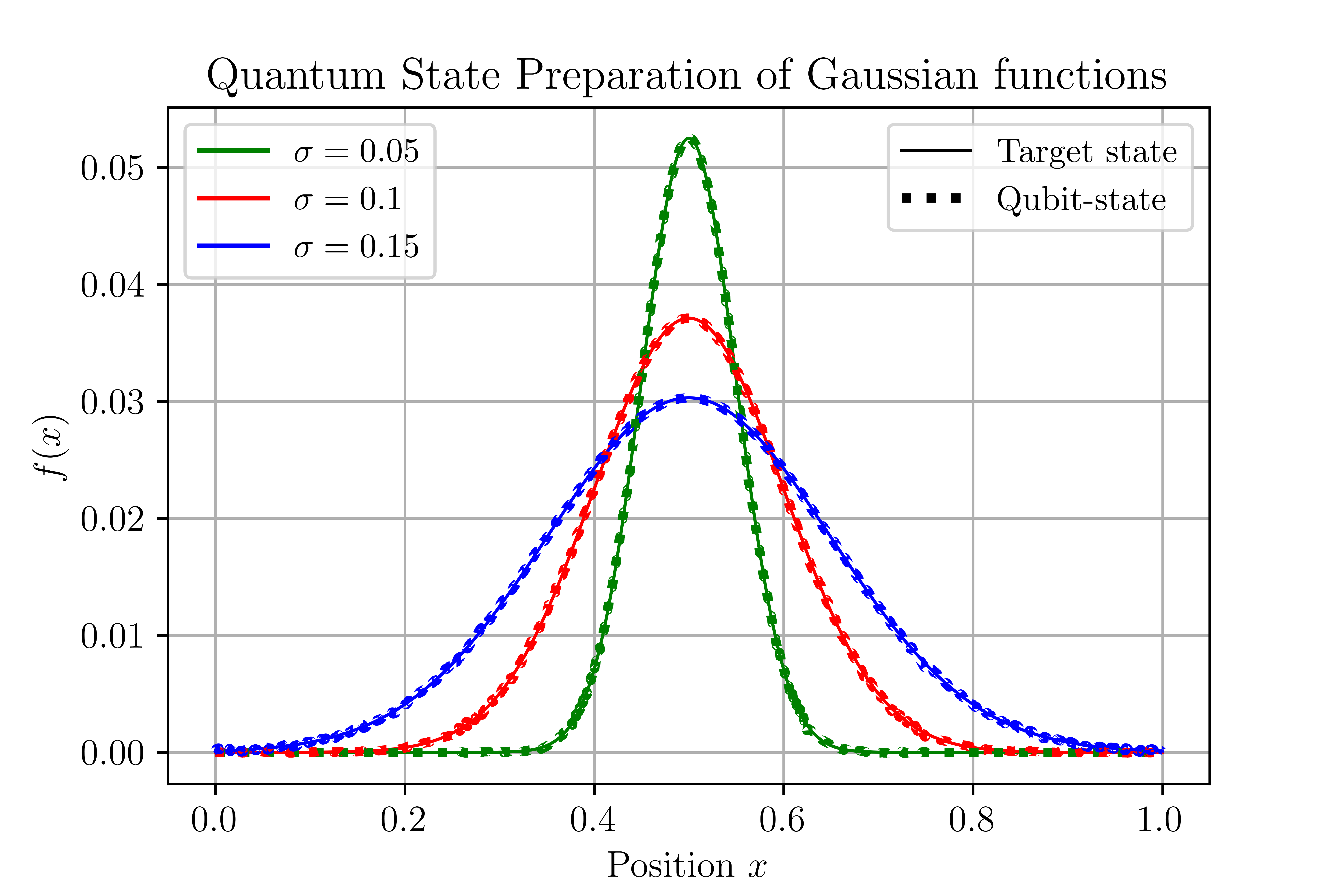}
    \caption{}
    \label{fig:first}
\end{subfigure}
\begin{subfigure}{0.47\textwidth}
    \includegraphics[width=\textwidth]{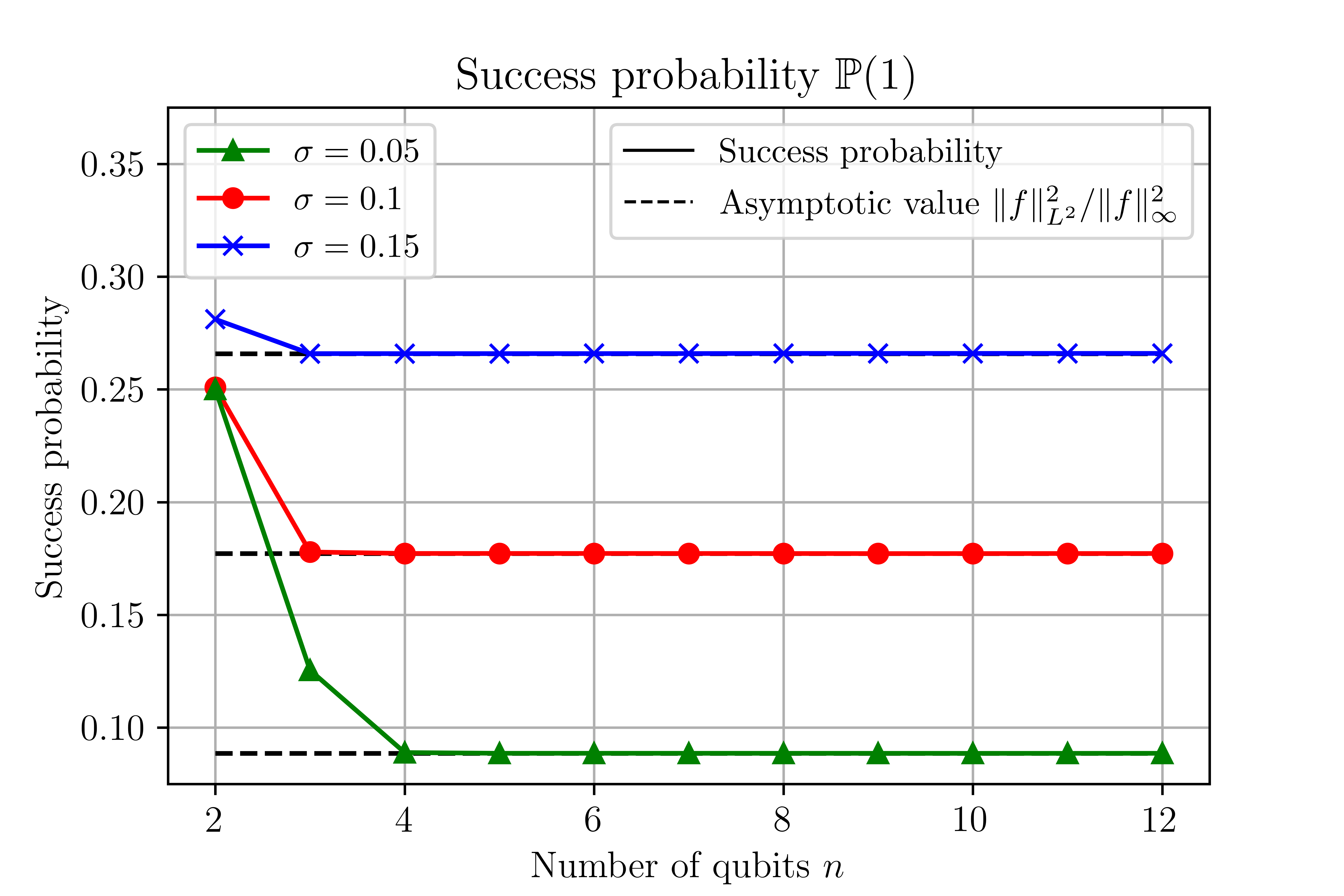}
    \caption{}
    \label{fig:second}
\end{subfigure}
\begin{subfigure}{0.47\textwidth}
    \includegraphics[width=\textwidth]{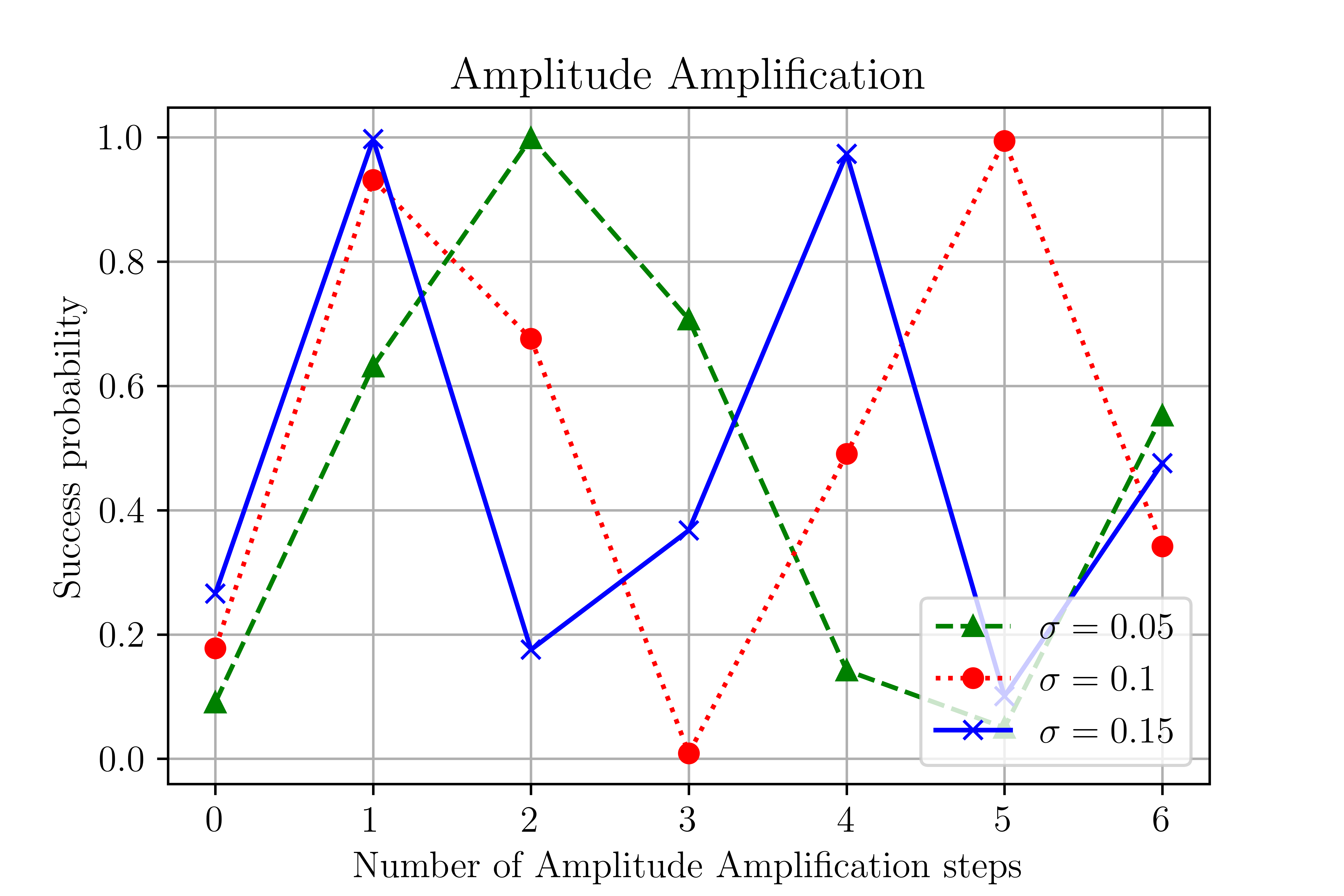}
    \caption{}
    \label{fig:third}
\end{subfigure}
\begin{subfigure}{0.47\textwidth}
    \includegraphics[width=\textwidth]{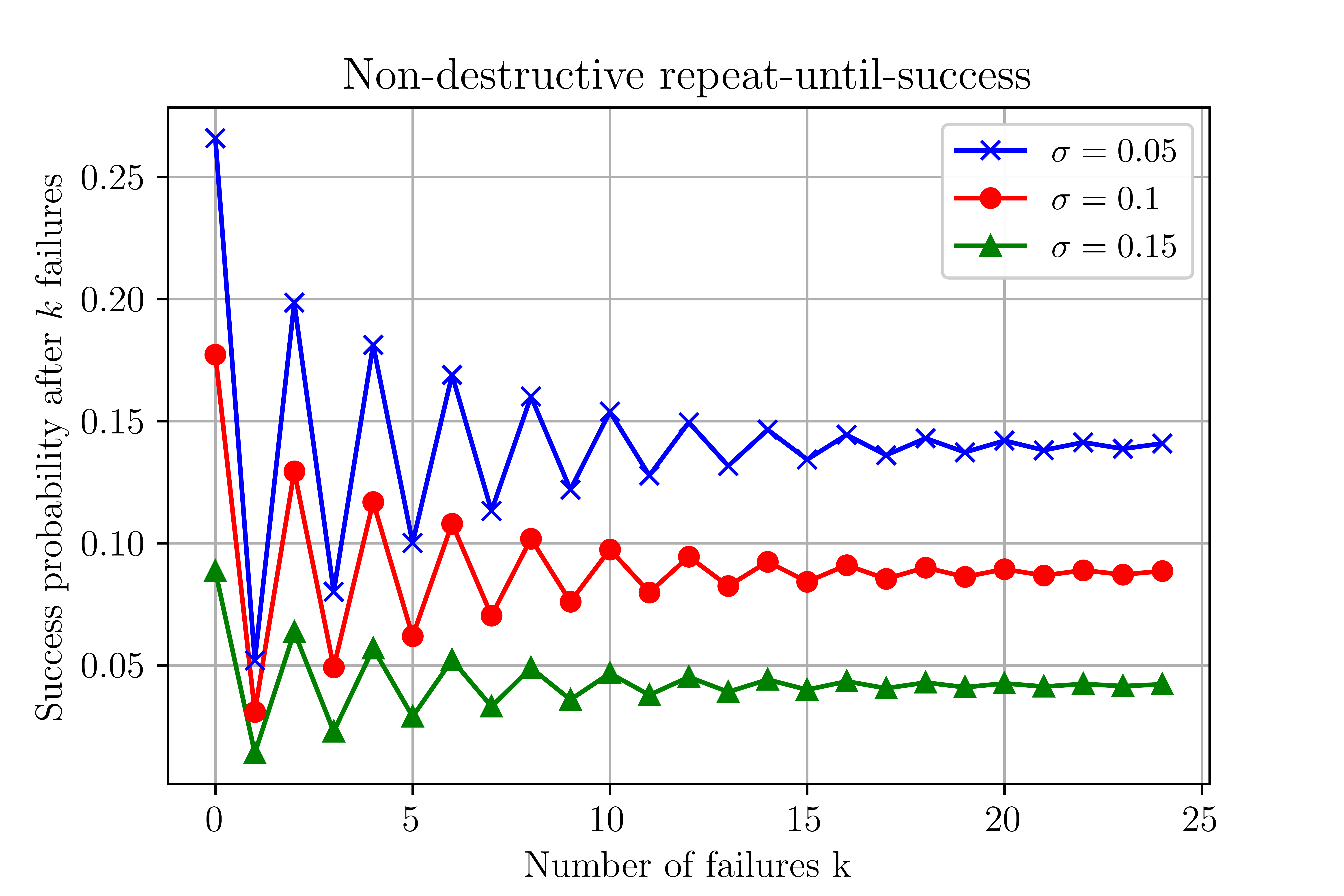}
    \caption{}
    \label{fig:fourth}
\end{subfigure}
\caption{Quantum state preparation (Fig. \ref{fig:first}) of Gaussian state $f(x)\propto e^{-0.5(x-0.5)^2/\sigma^2}$ for $\sigma=0.05, 0.1$ and $0.15$ on $n=12$ qubits using respectively $s=30, 45, 90$ of Walsh-Hadamard operators with 2-norm error between the qubit state and the target state of $0.0054, 0.0052, 0.0054$  respectively. Probability of success of the QSP protocol as a function of the number of qubits (Fig. \ref{fig:second}), as a function of the number of amplitude amplification steps (Fig. \ref{fig:third}) and probability of success after $k$ failures for the non-destructive repeat-until-success scheme (Fig. \ref{fig:fourth}). The numerical values of $\mathbb{P}(1)$ reach their theoretical asymptotic limit $\|f\|^2_{L^2}/ \|f\|_{\infty}^2$  even for a small number of qubits and the number of amplitude amplification steps needed to reach the first pic of probability verify $k=\lfloor \pi/(4\arcsin (\sqrt{\mathbb{P}(1)})) \rfloor = 1, 1, 2$ respectively. For the non-destructive repeat-until-success scheme, the average number of failures before reaching a success are $6, 10, 22$ respectively.
}
\label{fig:figures QSP}
\end{figure}

\begin{figure}
\centering
\includegraphics[width=0.4\textwidth]{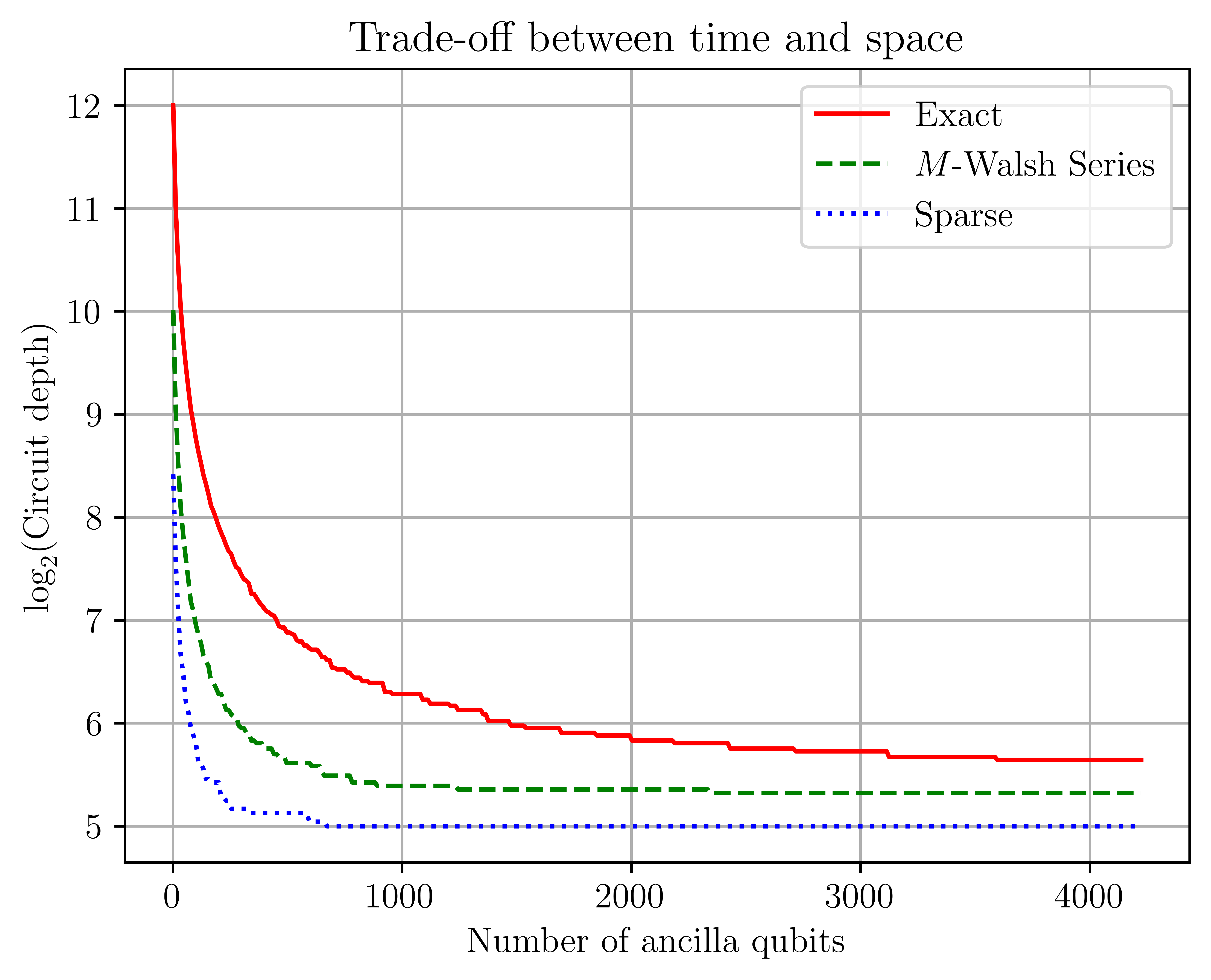}    
\caption{Depth as a function of the number of ancilla qubits for the quantum state preparation of a Gaussian state $f(x)\propto e^{-0.5(x-0.5)^2/\sigma^2}$ with $\sigma=0.1$ on $n=10$ qubits for three different Walsh-Hadamard decompositions: exact with $2{10}$ operators, approximate with $2^8$ operators and an infidelity $1-F=5.96\times10^{-5}$ and sparse with 70 Walsh-Hadamard operators and an infidelity $1-F=6.06\times10^{-5}$.}
\label{fig: space_time_tradeoff}
\end{figure}

\begin{figure}
\centering
\includegraphics[width=0.5\textwidth]{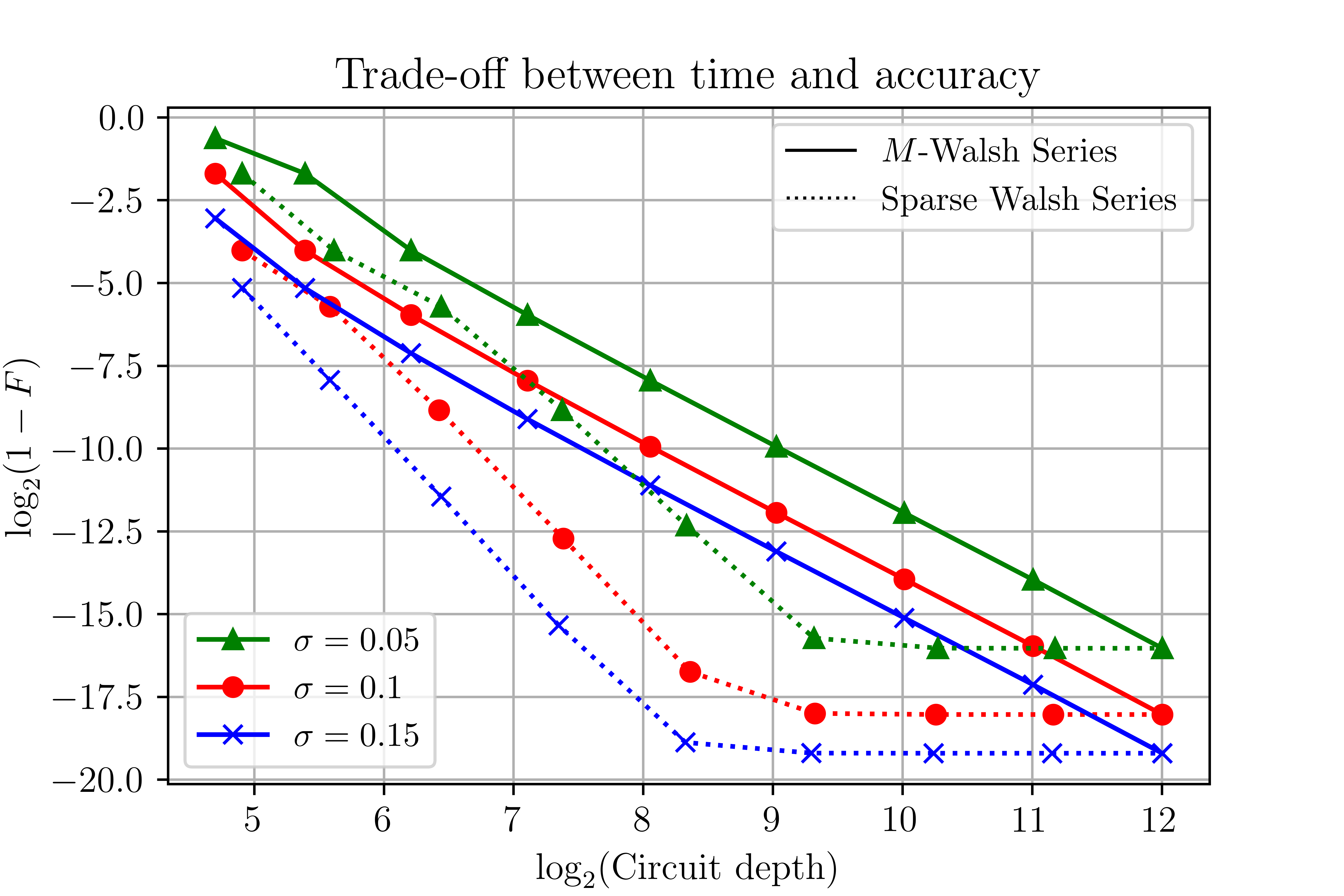}        
\caption{Infidelity as a function of the depth of the quantum circuit associated to the quantum state preparation of Gaussian states $f(x)\propto e^{-0.5(x-0.5)^2/\sigma^2}$ with $\sigma=0.05, 0.1, 0.15$ on $n=12$ qubits without ancilla qubits. The dots from upper left to lower right correspond to the $M$-Walsh Series with $M=2^m$ Walsh-Hadamard operators (full lines) and to the sparse Walsh Series composed of the $2^m$ Walsh-Hadamard operators with the largest coefficients in absolute value with $m=2,3,\hdots,10$. }
\label{fig: tradeoff_time_accuracy}
\end{figure}

\subsection{Real-space simulation of the Heat equation}

As an illustration of the methods presented in this article, we solve the following one-dimensional diffusion equation for a given initial condition $f_0$ defined on $[0,1]$ and a constant diffusion coefficient $\kappa$:

\begin{equation}
    \partial_t f = \kappa \partial_{xx} f,
\label{eq : heat equation}
\end{equation}
where $f$ is a function of the time $[0,T]$ and space $[0,1]$. For simplicity, we consider periodic boundary conditions. The numerical resolution is based on a real-space simulation where the space $[0,1]$ of the simulation is discretized into $N$ equi-length intervals. The spatial derivatives are discretized using a central finite difference $\partial_x f \rightarrow \frac{f(x+\Delta x)-f(x-\Delta x)}{2\Delta x}$ with $\Delta x=1/N$ the length of a space interval. 

The real space encoding consists in approximating the solution $f$ of Eq. \eqref{eq : heat equation} at the $N=2^n$ spatial points $\mathcal{X}_n=\{0,1/N,...,(N-1)/N\}$ and at each time $t>0$ by a state $\ket*{\tilde{f}}_t$, solution of the discretized diffusion equation :
\begin{equation}
\label{eq:Discrete ODE}
\begin{split}
    \partial_t \ket*{\tilde{f}}_t &= \kappa\frac{(\hat{S}-\hat{S}^\dagger)^2}{4\Delta x^2} \ket*{\tilde{f}}_t \\
\ket*{\tilde{f}}_{t=0} &=\ket{f_0}
\end{split},
\end{equation}
where $\frac{(\hat{S}-\hat{S}^\dagger)^2}{4\Delta x^2}$ is the discretized Laplacian operator with $\hat{S}=\sum_{x\in \mathcal{X}_n}\ket{x+1}\bra{x}$ the shift operator, also called increment operator\footnote{with $\ket{n+1}=\ket{0}$ due to the periodic boundary conditions.}. The solution $\ket*{\tilde{f}}_t$ is given by the evolution operator applied on the initial condition $\ket{f_0}$:
\begin{equation}
   \ket*{\tilde{f}}_t=e^{\frac{\kappa t }{4\Delta x^2} (\hat{S}-\hat{S}^\dagger)^2}\ket{f_0}.
\end{equation}

Note that $ \ket*{\tilde{f}}_t$ is not normalized at this stage of the computations. The evolution operator of the diffusion equation can be diagonalized in Fourier space using the Quantum Fourier Transform \cite{coppersmith2002approximate,nielsen2002quantum} thanks to the fact that the shift operators are circulant matrices which are diagonal in Fourier space \cite{10.5555/1202296}. The diagonal operator is non-unitary and depends on a smooth function $f(x)=e^{-\frac{\kappa}{\Delta x^2} \sin(2\pi x)^2 t}$, so the evolution is given by:

\begin{equation}
    \ket*{\tilde{f}}_t=\widehat{QFT}^{-1}e^{-\frac{\kappa t}{\Delta x^2} \sin(2\pi \hat{x})^2 }\widehat{QFT}\ket{f_0},
\label{eq: solution quantum numerical}
\end{equation}
where $\hat{x}=\sum_{x\in \mathcal{X}_n}x\ket{x}\bra{x}$ is the position operator.

The initial condition $\ket{f_0}$ is encoded in an $n$-qubit state thanks to the QSP protocol presented in the previous section (Theorem \ref{thm: qsp}) and the non-unitary diagonal operator is implemented using a sequential decomposition. Figure \ref{fig: heat equation resolution} presents an $n=8$-qubit simulation where the initial Gaussian state $f_0(x)=e^{-0.5(x-0.5)^2/\sigma^2}$ with $\sigma=0.1$  evolves through the heat equation. The amplitude of the numerical solution $\ket*{\tilde{f}}_t$ is shown as a function of the position for different times $t$ and compared with the analytical solution $f$ of Eq. \eqref{eq : heat equation} computed as a Fourier Series: 
\begin{equation}
f(x,t)=\sum_{q=0}^{+\infty}e^{-4\kappa q^2 \pi^2t} (\alpha_q \cos(2q\pi (x-0.5))+\beta_q \sin(2q\pi (x-0.5))),
\label{eq: analytical solution}
\end{equation}
with the Fourier coefficients, for $q\ge 1$:
\begin{equation}
\begin{split}
\alpha_q &= 2\int_0^1 f_0(x)\cos(2q\pi (x-0.5)) \textrm{d}x \\
\beta_q &= 2\int_0^1 f_0(x)\sin(2q\pi (x-0.5)) \textrm{d}x
\end{split},
\end{equation}
and $\alpha_0=\int_0^1 f_0(x)\textrm{d}x$, $\beta_0=0$.

To implement the resolution of the heat equation described in Eq. \eqref{eq: solution quantum numerical}, one needs to block-encode the operator $\hat{D}_t=\sum_{x=0}^{N-1}e^{-\frac{\kappa t}{\Delta x^2}\sin(2\pi x/N)^2}\ket{x}\bra{x}$, using the diagonal unitaries $e^{\pm i\arcsin{\hat{D}_t/(\alpha d_{\max})}}$. The eigenvalues of $\arcsin{(\hat{D}_t/(\alpha d_{\max}))}$ are shown in Fig. \ref{fig:second_diffusion} at different times of interest for $\alpha=1.1$. In this example, the number of non-trivial eigenvalues is very small, implying that a sequential decomposition is relevant\footnote{For a fixed $n$, the evolution operator becomes increasingly sparse in Fourier space as the time $t$ progresses.}. Fig. \ref{fig:first_diffusion} presents a comparison between the normalized analytical solution of the diffusion equation and the normalized states obtained from Eq. \eqref{eq: solution quantum numerical}, by implementing a very small number of sequential operators (14, 8, 6 and 1 respectively).

The error is defined in terms of the $2$-norm of the difference $\| \ket{f}_t/\|\ket{f}_t\|_{2,N}-\ket*{\tilde{\tilde{f}}}_t/\|\ket*{\tilde{\tilde{f}}}_t\|_{2,N} \|_{2,N}$, where $\ket{f}_t$ is a non-normalized state encoding the analytical solution at time $t$ as $\ket{f}_t=\sum_{x=0}^{N-1}f(x/N,t)\ket{x}$ and the associated normalization factor is $\|\ket{f}_t\|_{2,N}=\sqrt{\sum_{x=0}^{N-1}|f(x/N,t)|^2}$. $\ket*{\tilde{\tilde{f}}}_t$ is the prepared state $\ket*{\tilde{\tilde{f}}}_t=(\widehat{QFT}^{-1}) \tilde{\hat{D}}_t (\widehat{QFT})\ket*{\tilde{f}_0}$ with a normalization factor  $\|\ket*{\tilde{\tilde{f}}}_t\|_{2,N}= \| (\widehat{QFT}^{-1}) \tilde{\hat{D}}_t (\widehat{QFT})\ket*{\tilde{f}_0}\|_{2,N}$ where $\ket*{\tilde{f}_0}$ is the approximated initial condition and $\tilde{\hat{D}}_t $ is the approximation of the diagonal operator $\hat{D}_t$.

Three approximations contribute to the error: the approximate quantum state preparation of the initial condition, the space discretization with the centered finite difference approximation of the Laplacian operator, and the approximation in implementing the evolution operator.  The quantity  $\| \ket{f}_t/\|\ket{f}_t\|_{2,N}-\ket*{\tilde{\tilde{f}}}_t/\|\ket*{\tilde{\tilde{f}}}_t\|_{2,N} \|_{2,N}$ can be bounded by the sum of these three errors. The first one is the approximation of the initial Gaussian function. This approximation is discussed in detail in the previous section on quantum state preparation \ref{sec:qsp} and in Fig \ref{fig: tradeoff_time_accuracy}.  The second source of error arises from the spatial discretization of the Laplacian operator. The discretization error is the difference between $\ket{f}_t/\|\ket{f}_t\|_{2,N}$ and $\ket*{\tilde{f}}_t/ \| \ket*{\tilde{f}}_t \|_{2,N}$, where $\ket*{\tilde{f}}_t$ is the solution of Eq. \eqref{eq:Discrete ODE}. This difference can be bounded using the variation of parameter formula and the Taylor formula, resulting in $\|\ket{f}_t/\|\ket{f}_t\|_{2,N}-\ket*{\tilde{f}}_t/ \| \ket*{\tilde{f}}_t \|_{2,N} \|_{2,N}\leq Kt/N$, where $K$ is a constant that depends on the maximum value of the third spatial derivatives of $f$. A detailed computation of the discretization error is given in Appendix \ref{sec:discretization error}. The last source of error comes from the sparsity of the diagonal operator represented in Fig \ref{fig:second_diffusion}. Only a few eigenvalues are  implemented exactly, while the others are approximated as $0$.  

In practice, the numerical simulations demonstrate that the difference globally diminishes as the time $t$ increases. This is peculiar to the diffusion equation, as the Fourier components gradually vanish with the time $t$. Only the constant average value remains in the long-time limit. Therefore, for this particular example of a partial differential equation, the numerical scheme converges as the time $t$ increases.

\begin{figure}[ht]
\centering
\begin{subfigure}{0.47\textwidth}
    \includegraphics[width=\textwidth]{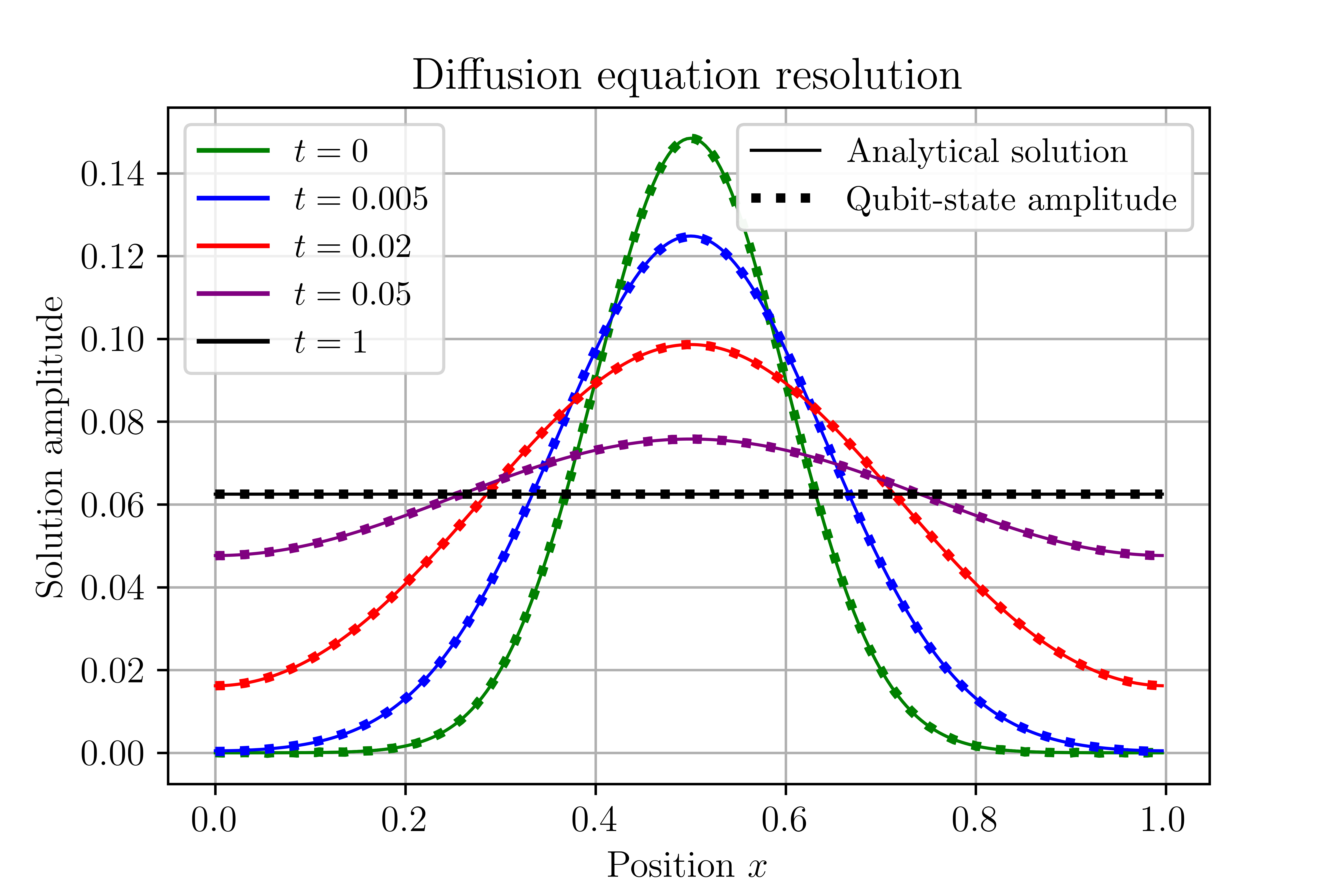}
    \caption{}
    \label{fig:first_diffusion}
\end{subfigure}
\begin{subfigure}{0.47\textwidth}
    \includegraphics[width=\textwidth]{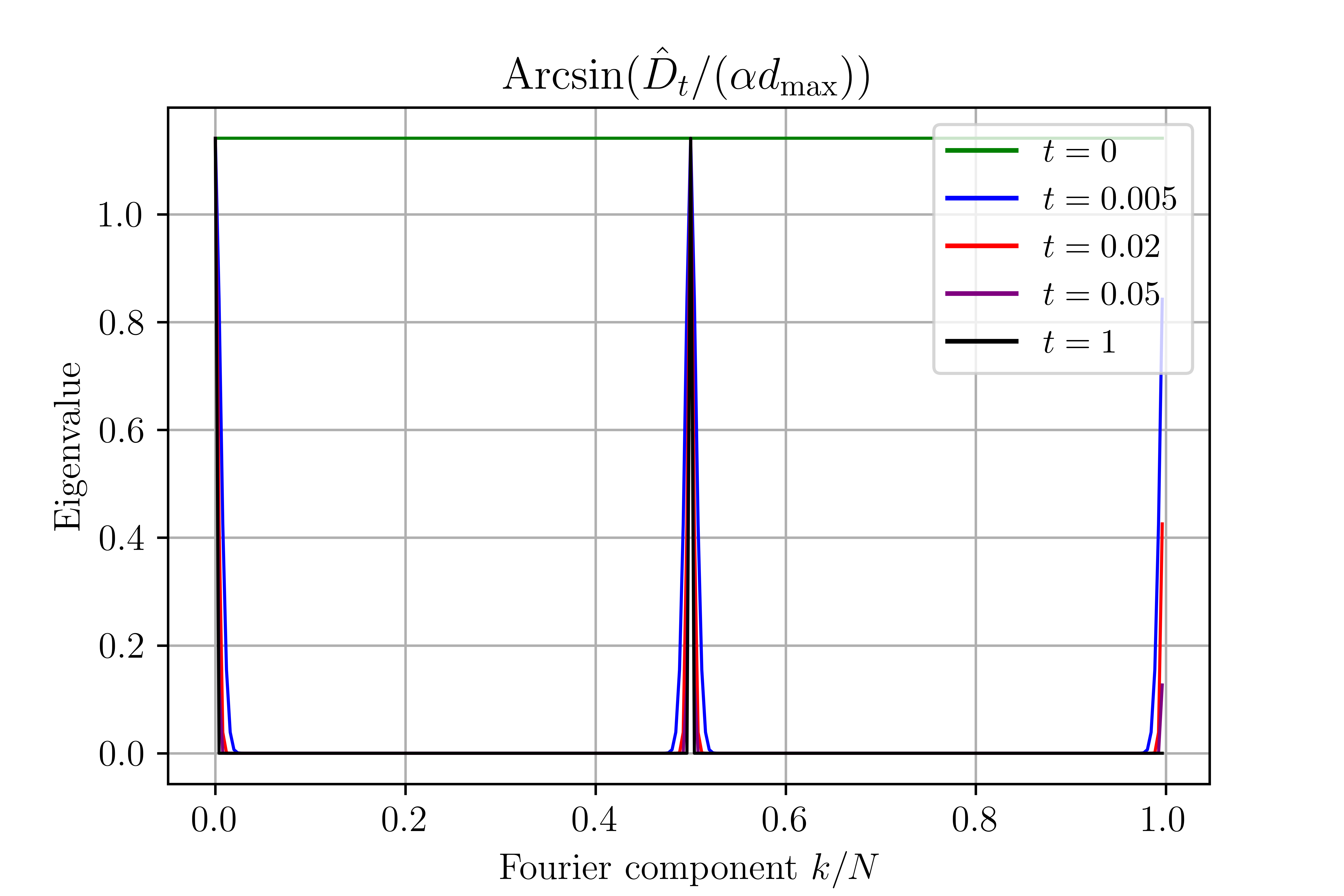}
    \caption{}
    \label{fig:second_diffusion}
\end{subfigure}
\caption{Resolution of the heat equation at different times $t$.  Fig. \ref{fig:first_diffusion} shows the $n=8$ qubit state given by Eq. \eqref{eq: solution quantum numerical} (dashed lines) and compares it to the analytical solution (full line) given by Eq. \eqref{eq: analytical solution} for an initial condition $f_0(x)=e^{-0.5(x-0.5)^2/\sigma^2}$, where $\sigma=0.1$. Fig. \ref{fig:second_diffusion} presents the eigenvalues $e^{-\frac{\kappa t}{\Delta x^2} \sin(2 \pi x)^2}$ of the evolution operator given in Eq. \eqref{eq: solution quantum numerical} at different times $t$, with parameters $\alpha=1.1$ and $d_{\max}=1$, highlighting its sparsity. The non-unitary diagonal operator is implemented using $s=14,8,6,1$ sequential operators for times $t=0.005,0.02,0.05,1$ respectively. The errors between the analytical solution and the implemented qubit states are given by the $2$-norm of the difference of the normalized vectors: $\| \ket{f}_t/\|\ket{f}_t\|_{2,N}-\ket*{\tilde{\tilde{f}}}_t/\|\ket*{\tilde{\tilde{f}}}_t\|_{2,N} \|_{2,N}$ with values $2.2\times10^{-3}, 1.8\times10^{-3}, 2.2\times10^{-4}, 2.7\times10^{-4}$ and $3.2\times10^{-15}$ for times $t=0.005,0.02,0.05,1$ respectively. }
\label{fig: heat equation resolution}
\end{figure}

\section{Discussion}

The methods presented in this work are the first to enable the implementation of unitary and non-unitary diagonal operators with trade-offs between the depth, the number of ancilla and the error at the same time. Previous methods focused on exact synthesis of diagonal operators \cite{barenco1995elementary,bullock2004asymptotically,zhang2024depth,sun2023asymptotically}, looking for optimal scalings either in terms of size \cite{bullock2004asymptotically,zhang2024depth} or depth \cite{sun2023asymptotically},  or considering trade-offs only between depth and the level of approximation \cite{welch2014efficient,10.5555/3179320.3179326}.  

Trade-offs between depth and width can always be achieved for sequential or Walsh-Hadamard decomposition since the associated building-blocks are parallelizable. However, parallelization is not straightforward for the Fourier decomposition proposed by Motlagh et al. \cite{motlagh2024generalized} because the sum of the Fourier components rely on a quantum signal processing protocol. This protocol does not consist of a product of commuting operators and only the control phase gates can be parallelized  using at most $n$ ancilla qubits (see corollary \ref{cor:GQSP approximate parallel}).

Trade-offs between depth and accuracy are only possible for a specific class of diagonal operators, whether unitary or non-unitary, namely the diagonal operators depending on continuously-differentiable functions such as $\hat{U}=\sum_{x=0}^{N-1} e^{if(x)}\ket{x}\bra{x}$ or $\hat{D}=\sum_{x=0}^{N-1} f(x/N)\ket{x}\bra{x}$. The bound on the error given in the approximation Theorem \ref{thm: Approximation theorem} depends directly on the maximum of the derivative of the associated function. When the derivative of $f$ is unbounded, or if it scales with the number of qubits, the error is not guarranted to decrease as the depth increases. Numerical evidence demonstrates that the trade-off remain possible for square-root functions and some bounded piecewise differentiable functions \cite{PhysRevA.109.042401}.  Approximating diagonal operators with Fourier series \cite{motlagh2024generalized} is possible and may converge exponentially fast, but the convergence criterion is stronger than the Walsh-Hadamard condition on the derivative of $f$.

The quantum state preparation Theorem \ref{thm: qsp} provides an efficient and adjustable-depth method to prepare quantum states associated with continuously-differentiable functions. The provided scalings are upper bounds on the computational resources needed to reach a given accuracy. It is often possible to construct shallower quantum circuits by considering sparse Walsh-Hadamard series in which only the dominant terms of the series are implemented. In particular, compared to the Walsh Series Loader presented in  \cite{PhysRevA.109.042401}, Theorem \ref{thm: qsp} improves the quantum-state-preparation probability of success from $\Theta(\epsilon)$ to $\Theta(1)$ at the cost of implementing two controlled diagonal unitaries instead of one. Preparing sparse quantum states, where only $s$ components are non-zero, can be achieved by applying a non-unitary diagonal operator to the uniform superposition state $\ket{s}=(1/N)\sum_{x=0}^N\ket{x}$. The computational costs associated with the sequential sparse methods are summarized in Table \ref{Table : sparse case}. These methods may be compared to other sparse quantum state preparation algorithms that are not based on diagonal operators \cite{feniou2024sparse,de2022double}. In particular, the sparse sequential methods presented here do not exploit the double-sparsity property of some quantum states, as detailed in \cite{de2022double}.

Additionally, the 'dilation' method presented in FIG. 1 of \cite{PhysRevA.106.022414} for implementing non-unitary diagonal operators is very similar to the one presented in FIG. \ref{quantum circuit scheme for non-unitary diagonal}. Their "dilation" method is identical to a block-encoding method in which one additional qubit is used. A diagonal operator with eigenvalues $\sigma_{ii}\in \mathbb{C}$ can be implemented with two diagonal unitaries denoted by $\hat{\Sigma}_{\pm}$ in Eq. 2 of  \cite{PhysRevA.106.022414}. These diagonal unitaries correspond to the $e^{\pm i \arcsin(\hat{D}/d_{\max})}$ of FIG. \ref{quantum circuit scheme for non-unitary diagonal}: $\hat{\Sigma}_{ii \pm}=\sigma_{ii}\pm i \sqrt{1-|\sigma_{ii}|^2}(\sigma_{ii}/|\sigma_{ii}|)=(\sigma_{ii}/|\sigma_{ii}|)e^{\pm i \arccos(|\sigma_{ii}|)}$ with $|\sigma_{ii}|\leq 1$. Two differences arise from the conventions. First, we separate the modulus part and the phase part of the complex eigenvalues: any non-unitary diagonal operator with complex eigenvalues can be written as a product of a diagonal unitary (encoding the phases) with a non-unitary diagonal operator (encoding the modulus). This latter operator is then block-encoded using $e^{\pm i \arcsin(|\hat{D}|/d_{\max})}$. Instead of implementing two diagonal operators, one may implement a single block-encoding given by the controlled diagonal unitaries  $e^{i (\arg(\hat{D}) \pm \arcsin(|\hat{D}|/d_{\max}))}$, where $\arg(\hat{D})$ is the diagonal operator encoding the argument of the eigenvalues of $\hat{D}$. The second difference is that we chose the success condition to correspond to the ancilla state $\ket{1}$, whereas in \cite{PhysRevA.106.022414} it corresponds to $\ket{0}$. A bit-flip gate $\hat{X}$ can be used to switch between these conventions.

\section{Conclusion}

In conclusion, we have presented a versatile framework to implement efficiently unitary and non-unitary diagonal operators with adjustable-depth circuits. We have demonstrated that ancilla qubits can be used to parallelize the implementation of diagonal operators, making it possible to customize the depth of quantum circuits with respect to the width. This framework is based on two key tunable features: parallelization and approximation. We have established the theoretical foundations through several theorems, detailling how the full parallelization and the adjustable-depth techniques can be applied to any diagonal unitary exactly decomposed into simpler operators.

We have also extended these new methods in two different ways. First, we have shown how to implement in an approximate manner unitary operators which depend on continuously differentiable functions. Second, we have adapted the above framework to non-unitary operators.

Applications as diverse as quantum state preparation, Hamiltonian simulations and numerical resolution of partial differential equations require implementing non-unitary diagonal operators efficiently. For example, non-unitary operators appear in PDEs generated by stochastic processes \cite{oksendal2013stochastic} and are thus essential in many contexts ranging from physics, chemistry and biology \cite{van2004stochastic, freund2000stochastic} to image processing \cite{gonzalez2009digital} and finance \cite{BLACK1976167,pironneau2009partial}.  However, implementing efficiently non-unitary operators is a significant challenge in the context of circuit-based quantum computing. Our approach involves the block-encoding of non-unitary diagonal operators into larger unitary ones, retaining the adjustable properties of the quantum circuits. We have presented strategies to successfully implement these block-encoded operators, such as non-destructive repeat-until-success schemes and amplitude amplification techniques, which allow for depth adjustment based on success probability.

As an illustration, we have presented the practical implementations of our methods to efficient quantum state preparation and high-fidelity simulation of the diffusion equation associated to Brownian motion. By preparing an initial Gaussian distribution and simulating its evolution, we demonstrate the capability of our framework to implement non-unitary operations with space-time-accuracy trade-offs.

Overall, this work provides a comprehensive set of tools for implementing both unitary and non-unitary diagonal operators in quantum circuits, with the flexibility to adjust depth as needed. This paves the way for more efficient and adaptable quantum algorithms, broadening the horizons for quantum computing applications. Additional work could consider to use these quantum circuits to solve more complex multiscale-multidimensional partial differential equations modeling physics, chemistry of finance phenomena or could consider that approximate methods may be generalizable to diagonal operators depending on bounded, continuous functions even when the functions are not differentiable.

\section*{Acknowledgements}

The authors would like to thank G.Di Molfetta, B.Claudon and C.Feniou for their useful feedbacks on the form and content of this manuscript. The quantum circuit diagrams were generated with the quantikz package \cite{kay2018tutorial}. U.Nzongani acknowledges support from the PEPR EPiQ ANR-22-PETQ-0007, by the ANR
JCJC DisQC ANR-22-CE47-0002-01. 

Code files for unitary diagonal operators, non-unitary diagonal operators, and quantum state preparation using the Walsh-Hadamard decomposition are available at \cite{github_ugo}.

\bibliography{maintext.bib}
\bibliographystyle{unsrt}
\appendix

\section{Examples of quantum circuits}\label{app:circuits}

In this section, illustrative examples of quantum circuits are presented for both the sequential and Walsh-Hadamard decompositions with $n=3$ qubits. First, the non-optimized circuits are displayed, then the optimize one with a Gray code ordering to reduce the number of gates. Finally, the quantum circuit associated to the parallelization scheme, shown on Fig. \ref{fig:adjustable_depth_qc}, with $p=3$ ancilla qubits is presented. 
Consider the $3$-qubit diagonal unitary:

\begin{equation}
\hat{U}_\theta=e^{i\hat{\theta}}=
\begin{pmatrix}
e^{i\theta_0} & 0 & 0 & 0 & 0 & 0 & 0 & 0 \\
0 & e^{i\theta_1} & 0 & 0 & 0 & 0 & 0 & 0 \\
0 & 0 & e^{i\theta_2} & 0 & 0 & 0 & 0 & 0 \\
0 & 0 & 0 & e^{i\theta_3} & 0 & 0 & 0 & 0 \\
0 & 0 & 0 & 0 & e^{i\theta_4} & 0 & 0 & 0 \\
0 & 0 & 0 & 0 & 0 & e^{i\theta_5} & 0 & 0 \\
0 & 0 & 0 & 0 & 0 & 0 & e^{i\theta_6} & 0 \\
0 & 0 & 0 & 0 & 0 & 0 & 0 & e^{i\theta_7}
\end{pmatrix},
\end{equation}
where each $e^{i\theta_x}$ corresponds to the eigenvalue of basis state eigenvector $\ket{x}$. The main register of the following quantum circuits is $\ket{q=q_{n-1}\dots q_1q_0}$ where $q_i\in\{0,1\}$ and $q=\sum_{j=0}^{n-1}q_j2^j$.

\subsection{Sequential decomposition}
\label{example sequential decomposition}
The sequential decomposition of diagonal unitaries leads to quantum circuits composed of multi-controlled phase and NOT gates. Figure \ref{fig:naive_sequential} gives an example of an exact decomposition of an arbitrary $3$-qubit diagonal unitary where the multi-controlled gate are implemented in decimal ordering.

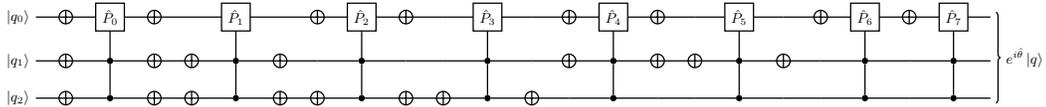
\begin{figure}[ht]
    \centering
\scalebox{0.6}{
\begin{quantikz}
  \lstick{$\ket{q_0}$} & \targ{} & \gate{\hat{P}_0} & \targ{} & \qw & \gate{\hat{P}_1} & \qw & \targ{} & \gate{\hat{P}_2} & \targ{} & \qw & \gate{\hat{P}_3} & \qw & \targ{} & \gate{\hat{P}_4} & \targ{} & \qw & \gate{\hat{P}_5} & \qw & \targ{} & \gate{\hat{P}_6} & \targ{} & \gate{\hat{P}_7} & \rstick[3]{$e^{i\hat{\theta}}\ket{q}$}\qw \\
  \lstick{$\ket{q_1}$} & \targ{} & \ctrl{-1} & \targ{} & \targ{} & \ctrl{-1} & \targ{} & \qw & \ctrl{-1} & \qw & \qw & \ctrl{-1} & \qw & \targ{} & \ctrl{-1} & \targ{} & \targ{} & \ctrl{-1} & \targ{} & \qw & \ctrl{-1} & \qw & \ctrl{-1} & \qw  \\
  \lstick{$\ket{q_2}$} & \targ{} & \ctrl{-1}& \targ{} & \targ{} & \ctrl{-1} & \targ{} & \targ{} & \ctrl{-1} & \targ{} & \targ{} & \ctrl{-1} & \targ{} & \qw & \ctrl{-1} & \qw & \qw & \ctrl{-1} & \qw & \qw & \ctrl{-1} & \qw & \ctrl{-1} & \qw 
\end{quantikz} 
}
\caption{Sequential decomposition without Gray code for $n=3$ qubits with the phase gates $\hat{P}_x=\hat{P}(\theta_x)=\begin{pmatrix}
    1 & 0 \\ 0 & e^{i\theta_x}
\end{pmatrix}$ with $x\in [0,7]$.}
\label{fig:naive_sequential}
\end{figure}

A significant number of NOT gates cancel out. One can use a Gray code \cite{gray-pulse-code-communication-1953} to maximise the number of cancellation. The Gray code, also called reflected binary code, was first introduced by Frank Gray in 1953. It is a way to enumerate a set of binary element such that the Hamming distance of two neighboring element is equal to 1, {\sl i.e.}, there is only a one bit difference between two neighboring bit strings. For instance, the 3-bits Gray code is shown on Table \ref{fig:gray}.

\begin{table}[ht]
\centering
\begin{tabular}{|c|c|c|}
\hline
\textit{Decimal with Gray code ordering} & \textit{Binary representation} \\
\hline
0 & 000 \\
\hline
1 & 001 \\
\hline
3 & 011 \\
\hline
2 & 010 \\
\hline
6 & 110 \\
\hline
7 & 111 \\
\hline
5 & 101 \\
\hline
4 & 100 \\
\hline
\end{tabular}
\caption{3-bits Gray code.}
\label{fig:gray}
\end{table}

The quantum circuit implementing an arbitrary $3$-qubit diagonal unitary using a Gray code ordering is represented Fig. \ref{fig:gray_sequential}.

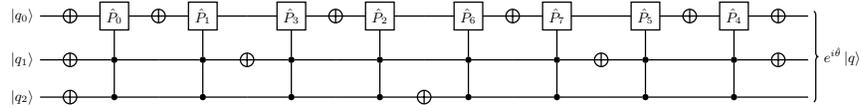
\begin{figure}[ht]
    \centering
\scalebox{0.6}{
\begin{quantikz}
  \lstick{$\ket{q_0}$} & \targ{} & \gate{\hat{P}_0} & \targ{} & \gate{\hat{P}_1} & \qw & \gate{\hat{P}_3} & \targ{} & \gate{\hat{P}_2} & \qw & \gate{\hat{P}_6} & \targ{} & \gate{\hat{P}_7} & \qw & \gate{\hat{P}_5} & \targ{} & \gate{\hat{P}_4} & \targ{} & \rstick[3]{$e^{i\hat{\theta}}\ket{q}$}\qw \\
  \lstick{$\ket{q_1}$} & \targ{} & \ctrl{-1} & \qw & \ctrl{-1} & \targ{} & \ctrl{-1} & \qw & \ctrl{-1} & \qw & \ctrl{-1} & \qw & \ctrl{-1} & \targ{} & \ctrl{-1} & \qw & \ctrl{-1} & \targ{} & \qw &  \\
  \lstick{$\ket{q_2}$} & \targ{} & \ctrl{-1} & \qw & \ctrl{-1} & \qw & \ctrl{-1} & \qw & \ctrl{-1} & \targ{} & \ctrl{-1} & \qw & \ctrl{-1} & \qw & \ctrl{-1} & \qw & \ctrl{-1} & \qw & \qw
\end{quantikz} 
}
\caption{Sequential decomposition with Gray code for $n=3$ qubits the phase gates $\hat{P}_x=\hat{P}(\theta_x)=\begin{pmatrix}
    1 & 0 \\ 0 & e^{i\theta_x}
\end{pmatrix}$ with $x\in [0,7]$.}
\label{fig:gray_sequential}
\end{figure}

Lastly, we show the quantum circuit implementing the adjustable-depth sequential decomposition of $e^{i\hat{f}}$ on Fig. \ref{fig:adjustable_sequential_example}.

\begin{figure}[ht]
\centering
\scalebox{0.6}{
\begin{quantikz}
  \lstick{$\ket{q_0}$} & \ctrl{3}\gategroup[6,steps=3,style={dashed,rounded corners,color=blue,inner xsep=2pt},background]{{$\widehat{\text{copy}}$}} & \qw & \qw & \targ{} & \gate{\hat{P}_0} & \targ{} & \gate{\hat{P}_1} & \qw & \gate{\hat{P}_3} & \targ{} & \gate{\hat{P}_2} & \targ{} & \qw\gategroup[6,steps=3,style={dashed,rounded corners,color=blue,inner xsep=2pt},background]{{$\widehat{\text{copy}}^{-1}$}}& \qw & \ctrl{3} & \rstick[3]{$e^{i\hat{f}}\ket{q}$}\qw \\
  \lstick{$\ket{q_1}$} & \qw & \ctrl{3} & \qw & \targ{} & \ctrl{-1} & \qw & \ctrl{-1} & \targ{} & \ctrl{-1} & \qw & \ctrl{-1} & \qw & \qw & \ctrl{3} & \qw & \qw \\
  \lstick{$\ket{q_2}$} & \qw & \qw & \ctrl{3} & \targ{} & \ctrl{-1} & \qw & \ctrl{-1} & \qw & \ctrl{-1} & \qw & \ctrl{-1} & \targ{} & \ctrl{3} & \qw & \qw & \qw \\
  \lstick{$\ket{q'_0=0}$} & \targ{} & \qw & \qw & \targ{} & \gate{\hat{P}_6} & \targ{} & \gate{\hat{P}_7} & \qw & \gate{\hat{P}_5} & \targ{} & \gate{\hat{P}_4} & \targ{} & \qw & \qw & \targ{} & \rstick{$\ket{q'_0=0}$}\qw \\
  \lstick{$\ket{q'_1=0}$} & \qw & \targ{} & \qw & \qw & \ctrl{-1} & \qw & \ctrl{-1} & \targ{} & \ctrl{-1} & \qw & \ctrl{-1} & \targ{} & \qw & \targ{} & \qw & \rstick{$\ket{q'_1=0}$}\qw \\
  \lstick{$\ket{q'_2=0}$} & \qw & \qw & \targ{} & \qw & \ctrl{-1} & \qw & \ctrl{-1}& \qw & \ctrl{-1}& \qw & \ctrl{-1} & \qw & \targ{} & \qw & \qw & \rstick{$\ket{q'_2=0}$}\qw
\end{quantikz} 
}
\caption{Adjustable-depth quantum circuit implementing the sequential decomposition with Gray code for $n=3$ qubits with $p=3$ ancilla qubits.}
\label{fig:adjustable_sequential_example}
\end{figure}
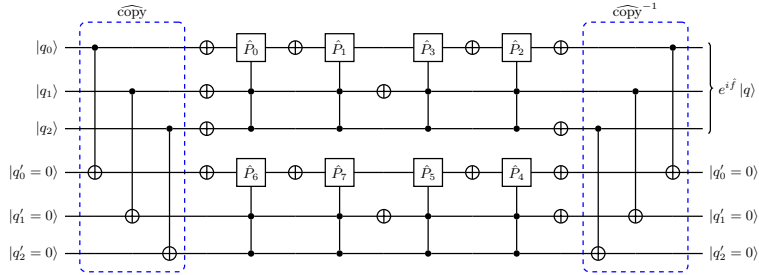

\subsection{Walsh-Hadamard decomposition}
\label{example Walsh-Hadamard decomposition}
Diagonal unitaries can be exactly implemented as a product of Walsh-Hadamard decomposition $\hat{U}=\prod \hat{W}_j$ is implemented with a quantum circuit composed of CNOT and $\hat{R}_Z$ gates. Its naïve form is shown on Fig. \ref{fig:naive_walsh}. 

\begin{figure}[ht]
    \centering
\scalebox{0.6}{
\begin{quantikz}
  \lstick{$\ket{q_0}$} & \gate{e^{ia_0}\hat{I}} & \gate{\hat{R}_1} & \qw & \ctrl{1} & \qw & \ctrl{1} & \qw & \ctrl{2} & \qw & \ctrl{2} & \qw & \qw & \qw & \ctrl{2} & \qw & \qw & \qw & \ctrl{2} & \qw & \rstick[3]{$e^{i\hat{\theta}}\ket{q}$}\qw\\
  \lstick{$\ket{q_1}$} & \qw & \qw & \gate{\hat{R}_2} & \targ{} & \gate{\hat{R}_3} & \targ{} & \qw & \qw & \qw & \qw & \ctrl{1} & \qw & \ctrl{1} & \qw & \ctrl{1} & \qw & \ctrl{1} & \qw & \qw & \qw \\
  \lstick{$\ket{q_2}$} & \qw & \qw & \qw & \qw & \qw & \qw & \gate{\hat{R}_4} & \targ{} & \gate{\hat{R}_5} & \targ{} & \targ{} & \gate{\hat{R}_6} & \targ{} & \targ{} & \targ{} & \gate{\hat{R}_7} & \targ{} & \targ{} & \qw & \qw
\end{quantikz} 
}
\caption{Walsh decomposition with no Gray code for $n=3$ qubits with $\hat{R}_x\equiv \hat{R}_Z(-2a_x)$ with $x\in [1,7]$. This circuit is shown on Fig. 5 of \cite{welch2014efficient}. The first gate corresponds to the global phase gate $e^{ia_0} \hat{I}=\begin{pmatrix}
    e^{ia_0} & 0 \\ 0 & e^{ia_0}
\end{pmatrix}$.}
\label{fig:naive_walsh}
\end{figure}
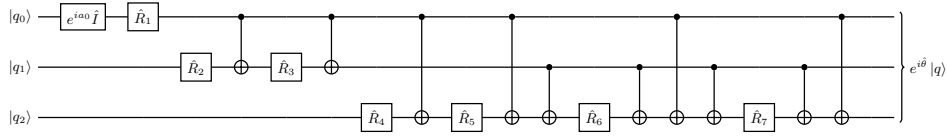

As with the sequential decomposition, the number of CNOTs can be reduced by changing the order of the operator using Gray code \cite{gray-pulse-code-communication-1953,welch2014efficient}. The obtained circuit is shown on Fig. \ref{fig:gray_walsh}.

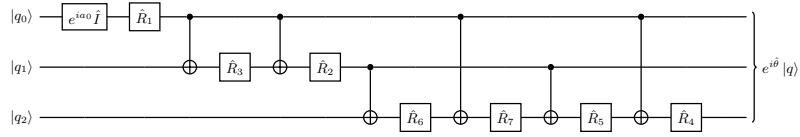
\begin{figure}[ht]
    \centering
\scalebox{0.6}{
\begin{quantikz}
  \lstick{$\ket{q_0}$} & \gate{e^{ia_0} \hat{I}} & \gate{\hat{R}_1} & \ctrl{1} & \qw & \ctrl{1} & \qw & \qw & \qw & \ctrl{2} & \qw & \qw & \qw & \ctrl{2} & \qw & \qw & \rstick[3]{$e^{i\hat{\theta}}\ket{q}$}\qw \\
  \lstick{$\ket{q_1}$} & \qw & \qw & \targ{} & \gate{\hat{R}_3} & \targ{} & \gate{\hat{R}_2} & \ctrl{1} & \qw & \qw & \qw & \ctrl{1} & \qw & \qw & \qw & \qw & \qw \\
  \lstick{$\ket{q_2}$} & \qw & \qw & \qw & \qw & \qw & \qw & \targ{} & \gate{\hat{R}_6} & \targ{} & \gate{\hat{R}_7} & \targ{} & \gate{\hat{R}_5} & \targ{} & \gate{\hat{R}_4} & \qw & \qw
\end{quantikz} 
}
\caption{Walsh decomposition with Gray code for $n=3$ qubits with $\hat{R}_x\equiv \hat{R}_Z(-2a_x)$ with $x\in [1,7]$. This circuit is shown on Fig. 5 of \cite{welch2014efficient}. The first gate corresponds to the global phase gate $e^{ia_0} \hat{I}=\begin{pmatrix}
    e^{ia_0} & 0 \\ 0 & e^{ia_0}
\end{pmatrix}$.}
\label{fig:gray_walsh}
\end{figure}

Finally, we show the quantum circuit implementing the adjustable-depth Walsh-Hadamard decomposition of $e^{i\hat{f}}$ on Fig. \ref{fig:adjustable_walsh_example}.

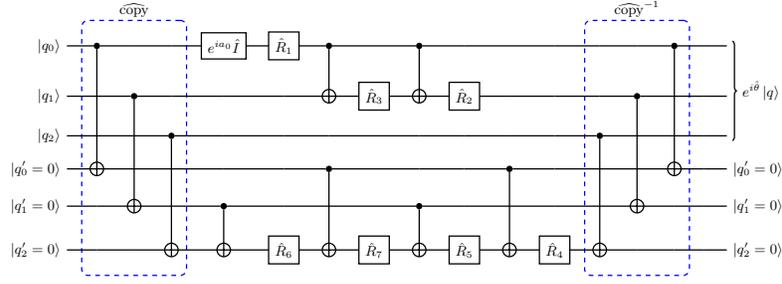
\begin{figure}[ht]
\centering
\scalebox{0.6}{
\begin{quantikz}
  \lstick{$\ket{q_0}$} & \ctrl{3}\gategroup[6,steps=3,style={dashed,rounded corners,color=blue,inner xsep=2pt},background]{{$\widehat{\text{copy}}$}} & \qw & \qw & \gate{e^{ia_0}\hat{I}} & \gate{\hat{R}_1} & \ctrl{1} & \qw & \ctrl{1} & \qw & \qw & \qw & \qw\gategroup[6,steps=3,style={dashed,rounded corners,color=blue,inner xsep=2pt},background]{{$\widehat{\text{copy}}^{-1}$}} & \qw & \ctrl{3} & \qw & \rstick[3]{$e^{i\hat{\theta}}\ket{q}$}\qw \\
  \lstick{$\ket{q_1}$} & \qw & \ctrl{3} & \qw & \qw & \qw & \targ{} & \gate{\hat{R}_3} & \targ{} & \gate{\hat{R}_2} & \qw & \qw & \qw & \ctrl{3} & \qw & \qw & \qw \\
  \lstick{$\ket{q_2}$} & \qw & \qw & \ctrl{3} & \qw & \qw & \qw & \qw & \qw & \qw & \qw & \qw & \ctrl{3} & \qw & \qw & \qw & \qw \\
  \lstick{$\ket{q'_0=0}$} & \targ{} & \qw & \qw & \qw & \qw & \ctrl{2} & \qw & \qw & \qw & \ctrl{2} & \qw & \qw & \qw & \targ{} & \qw & \rstick{$\ket{q'_0=0}$}\qw \\
  \lstick{$\ket{q'_1=0}$} & \qw & \targ{} & \qw & \ctrl{1} & \qw & \qw & \qw & \ctrl{1} & \qw & \qw & \qw & \qw & \targ{} & \qw & \qw & \rstick{$\ket{q'_1=0}$}\qw\\
  \lstick{$\ket{q'_2=0}$} & \qw & \qw & \targ{} & \targ{} & \gate{\hat{R}_6} & \targ{} & \gate{\hat{R}_7} & \targ{} & \gate{\hat{R}_5} & \targ{} & \gate{\hat{R}_4} & \targ{} & \qw & \qw & \qw & \rstick{$\ket{q'_2=0}$}\qw
\end{quantikz} 
}
\caption{Adjustable-depth quantum circuit implementing the Walsh decomposition with Gray code for $n=3$ qubits with $p=3$ ancilla qubits.}
\label{fig:adjustable_walsh_example}
\end{figure}

\section{Proof of Theorem \ref{thm: block-encoding}}
\label{proof thm block encoding}

\begin{proof}
Consider a decomposition $\tilde{U}=\prod_{j=0}^{p-1} \hat{U}_j$ where each $\hat{U}_j$ acts on $k_j\leq n$ qubits and consider a number $k+p$ ancilla qubits, with $k=\sum_{j=1}^{p-1}k_j$. The registers of ancilla qubits are denoted by $R_1,\hdots, R_{p-1}$ and each $R_j $ contains $k_j$ qubits that will serve to copy only the qubits on which $\hat{U}_j$ is acting non-trivially. The register $R_{q_A}$ contains $p-1$ qubits on which $q_A$ will be copied.
The proof of correctness of the quantum circuit Fig. \ref{Adjustable depth framework for non unitary diagonal operator} starts by considering a state $\ket{x}$ of the computational basis of the main register. For simplicity, tensor products between the qubit states and tensor products with identity operators are omitted:
\begin{equation}
\begin{split}
&\ket{x} \ket{0}^{\otimes k} \ket{0}_{q_A} \ket{0}^{\otimes (p-1)}  \\ &\xrightarrow{\hat{H}}  \ket{x} \ket{0}^{\otimes k} \frac{1}{\sqrt{2}}(\ket{0}_{q_A}+\ket{1}_{q_A}) \ket{0}^{\otimes (p-1)} \\
&\xrightarrow{\widehat{\text{copy}}_{R\rightarrow R_1\cup \hdots \cup R_{p-1}} \otimes \widehat{\text{copy}}_{q_A\rightarrow R_{q_A} }}
\ket{x} \ket{\tilde{x}_1}\hdots \ket{\tilde{x}_{p-1}}  \frac{1}{\sqrt{2}}(\ket{0}^{\otimes p}+\ket{1}^{\otimes p}) \\
&\xrightarrow{\bigotimes_{j=0}^{p-1} \hat{C}_{\bar{q}_j\rightarrow R_j}(\hat{U}_j)} \frac{1}{\sqrt{2}}(\hat{U}_0\ket{x} \hat{U}_1 \ket{\tilde{x}_1}\hdots \hat{U}_{p-1}\ket{\tilde{x}_{p-1}}\ket{0}^{\otimes p}\\&+\ket{x} \ket{\tilde{x}_1}\hdots \ket{\tilde{x}_{p-1}}\ket{1}^{\otimes p}) \\
&\xrightarrow{\bigotimes_{j=0}^{p-1} \hat{C}_{q_j\rightarrow R_j}(\hat{U}_j^{\dagger})}
 \frac{1}{\sqrt{2}}(\hat{U}_0\ket{x} \hat{U}_1 \ket{\tilde{x}_1}\hdots \hat{U}_{p-1}\ket{\tilde{x}_{p-1}}\ket{0}^{\otimes p}\\&+\hat{U}_0^{\dagger}\ket{x} \hat{U}_1^{\dagger}\ket{\tilde{x}_1}\hdots \hat{U}_{p-1}^{\dagger}\ket{\tilde{x}_{p-1}}\ket{1}^{\otimes p}) \\
&
= \frac{1}{\sqrt{2}}(e^{i\tilde{\theta}_x}\ket{x} \ket{\tilde{x}_1}\hdots\ket{\tilde{x}_{p-1}}\ket{0}^{\otimes p}+e^{-i\tilde{\theta}_x}\ket{x}  \ket{\tilde{x}_1}\hdots \ket{\tilde{x}_{p-1}}\ket{1}^{\otimes p}) \\
& \xrightarrow{\widehat{\text{copy}}_{R\rightarrow R_1\cup \hdots \cup R_{p-1}}^{-1} \otimes \widehat{\text{copy}}_{q_A\rightarrow R_{q_A}}^{-1}}  \frac{1}{\sqrt{2}}(e^{i\tilde{\theta}_x}\ket{x} \ket{0}^{\otimes k} \ket{0}_{q_A}+e^{-i\tilde{\theta}_x}\ket{x}  \ket{0}^{\otimes k}\ket{1}_{q_A})\ket{0}^{\otimes(p-1)} \\
& \xrightarrow{\hat{P}\hat{H}} (-i\cos(\theta_x)\ket{x} \ket{0}^{\otimes k} \ket{0}_{q_A}+\sin(\theta_x)\ket{x}  \ket{0}^{\otimes k}\ket{1}_{q_A})\ket{0}^{\otimes(p-1)}
\end{split}
\end{equation} 

with $\tilde{\theta}_x=\sum_{j=0}^{p-1} \theta_j(x)$ and $\hat{U}_j=\sum_{x=0}^{N-1}e^{i\theta_j(x)}\ket{x}\bra{x}$. The operator $\hat{C}_{\bar{q}_j\rightarrow R_j}(\hat{U}_j)$ is defined as the unitary $\hat{U}_j$ applied on the register $R_j$ and anti-controlled by the $j$-th qubit of $R_{q_A}$ ($q_0$ denotes $q_A$) and the operator $\hat{C}_{q_j\rightarrow R_j}(\hat{U}_j^{\dagger})$ is defined as the unitary $\hat{U}_j^\dagger$ applied on the register $R_j$ and controlled by the $j$-th qubit of $R_{q_A}$. The unitary $\widehat{\text{copy}}_{R\rightarrow R_1\cup \hdots \cup R_{p-1}}$ corresponds to the copy of main register $R$ into the copy registers $R_1,\hdots, R_{p-1}$ and $\widehat{\text{copy}}_{q_A\rightarrow R_{q_A}}$ corresponds to the copy of $q_A$ into the register $R_{q_A}$.

The unitary $\hat{U}_D$ corresponding to the previous operations ({\sl i.e.} Fig. \ref{Adjustable depth framework for non unitary diagonal operator}) is a $(\alpha d_{\max},k+p,\epsilon)$-block encoding of $\hat{D}=\sum_{x=0}^N d_x\ket{x}\bra{x}$:
\begin{equation}
\begin{split}
&\|\frac{1}{\alpha d_{\max}}\hat{D}-(\bra{0}^{\otimes k}\bra{1}_{q_A}\bra{0}^{\otimes(p-1)}\otimes \hat{I}_R)\hat{U}_D(\ket{0}^{\otimes k}\ket{1}_{q_A}\ket{0}^{\otimes(p-1)}\otimes \hat{I}_R) \|_2 \\
&=\|\sum_{x=0}^{N-1}\frac{d_x}{\alpha d_{\max}}\ket{x}\bra{x}-\sum_{x=0}^{N-1}\sin(\tilde{\theta}_x)\ket{x}\bra{x} \|_2 \\
&=\max_x \| \frac{d_x}{\alpha d_{\max}}-\sin(\tilde{\theta}_x) \| \leq \epsilon
\end{split}
\end{equation}
The inequality comes from the hypothesis that $\tilde{U}=\sum_{x=0}^N e^{i\tilde{\theta}_x}\ket{x}\bra{x}$ is an $\epsilon$ approximation of $\hat{U}=e^{i\arcsin(\hat{D}/(\alpha d_{\max}))}$ in spectral norm.
\end{proof}

\section{Scalings of the different quantum circuits}\label{sec:approximation_methods}

First, the results for the exact quantum circuit synthesis of unitary and non-unitary diagonal operators are presented in section \ref{sec:exact methods}. Then, these exact methods are combined with the approximation Theorem \ref{thm: Approximation theorem} to proove the scalings of the approximative methods in section \ref{sec:approximate methods}. A special subsection has been attributed to the fundamentally different Fourier GQSP methods  \ref{sec: Fourier GQSP methods}. The different methods for sparse diagonal operators are presented in section \ref{sec: sparse diagonal operator}.

\subsection{Exact methods}
\label{sec:exact methods}
\subsubsection{Diagonal unitaries}
The following Lemmas and corollaries summarize the complexity of implementing exactly any $n$-qubit diagonal unitary.

\begin{lemma}[Exact sequential decomposition without ancilla]
\label{Exact_sequential_decomposition_without_ancilla}
    Any $n$-qubit diagonal unitary is implementable through its sequential decomposition with a quantum circuit of depth $\mathcal{O}(n2^n)$ and size $\mathcal{O}(n2^n)$ without using ancilla qubits.
\end{lemma}
\begin{proof}
    The quantum circuit is composed of $2^n$ operator $\hat{U}_j$ defined in Eq. \eqref{sequential decomposition}. Thanks to the fact that the unitaries $\hat{U}_j$ commute with each other, one can choose the order of implementation of the $\hat{U}_j$ which cancels a maximum number of $\hat{X}$ gates. This is achieved using a Gray code \cite{gray-pulse-code-communication-1953,nielsen2002quantum, beauchamp1984applications}: between two following $\hat{U}_j$, $\hat{U}_{j'}$, only one bit changes in the binary decomposition of $j$ and $j'$. This leads to an ordering where only one $\hat{X}$ gate remains between two consecutive $\Lambda_{\{0,...,n-2\}}(\hat{P}(\theta_j))$, $\Lambda_{\{0,...,n-2\}}(\hat{P}(\theta_{j'})$. Then, each of the $\Lambda_{\{0,...,n-2\}}(\hat{P}(\theta_j))$ is implemented using the scheme of \cite{Craig} which is an exact scheme with size and depth linear in $n$ without using ancilla qubits.
\end{proof}

\begin{lemma}[Exact sequential decomposition with one ancilla]
\label{lemma: Exact sequential decomposition with one ancilla}
    Any $n$-qubit diagonal unitary is implementable through its sequential decomposition with a quantum circuit of depth $\mathcal{O}(\log(n)^3 2^{n})$ and size $\mathcal{O}(n\log(n)^4 2^{n})$ using one ancilla qubit.
\end{lemma}
\begin{proof}
Similarly than previously, the $2^n$ sequential operator $\hat{U}_j$ are implemented in a gray-code order. Then, each of them is implemented with scheme presented in \cite{claudon2024polylogarithmic} which uses one ancilla qubit to achieve a polylogarithmic depth $\mathcal{O}(\log(n)^3)$ to implement any $n$-controlled operations at the cost of a size $\mathcal{O}(n\log(n)^4)$ (see Corollary 1 in \cite{claudon2024polylogarithmic}). Therefore, the overall depth for an exact implementation of a given arbitrary diagonal unitary is $\mathcal{O}(\log(n)^3 2^n)$ and size $\mathcal{O}(n\log(n)^4 2^n)$ using one ancilla qubit.
\end{proof}

\begin{lemma}[Exact Walsh-Hadamard decomposition without ancilla (Theorem 1.3 in \cite{bullock2004asymptotically})]
\label{lemma: Exact walsh-hadamard decomposition}
    Any $n$-qubit diagonal unitary is implementable through its Walsh-Hadamard decomposition with a quantum circuit of depth $\mathcal{O}(2^n)$ and size $\mathcal{O}(2^n)$ without using ancilla qubits.
\end{lemma}

\begin{corollary}[Fully parallelized exact sequential decomposition]
\label{lemma: fully parallelized sequential decomposition with one ancilla}
    Any $n$-qubit diagonal unitary is implementable through its sequential decomposition with a fully parallelized quantum circuit of depth $\mathcal{O}(n)$ and size $\mathcal{O}(n2^n)$ using $\mathcal{O}(n2^n)$ ancilla qubits.
\end{corollary}
\begin{proof}
Lemma \ref{Exact_sequential_decomposition_without_ancilla} or lemma \ref{lemma: Exact sequential decomposition with one ancilla} associated to the full parallelization Theorem \ref{thm : full parallelization}.
\end{proof}

\begin{corollary}[Fully parallelized exact Walsh-Hadamard decomposition]
\label{lemma: fully parallelized walsh-hadamard decomposition}
    Any $n$-qubit diagonal unitary is implementable through its Walsh-Hadamard decomposition with a fully parallelized quantum circuit of depth $\mathcal{O}(n)$ and size $\mathcal{O}(n2^n)$ without using $\mathcal{O}(n2^n)$ ancilla qubits.
\end{corollary}
\begin{proof}
Lemma \ref{lemma: Exact walsh-hadamard decomposition} with the full parallelization Theorem \ref{thm : full parallelization}.
\end{proof}

\begin{corollary}[Partially parallelized exact sequential decomposition using linear depth decomposition]
\label{lemma: partially parallelized sequential decomposition}
    Any $n$-qubit diagonal unitary is implementable through its sequential decomposition with a partially parallelized quantum circuit of depth $\mathcal{O}(n^2 2^n/m+\log(m/n))$ and size $\mathcal{O}(n2^n+m)$ using $ m \in [n, \mathcal{O}(n2^n)]$ ancilla qubits. 
\end{corollary}
\begin{proof}
Lemma \ref{Exact_sequential_decomposition_without_ancilla} with the partial parallelization Theorem \ref{thm : adjustable-depth with ancilla}.
\end{proof}

\begin{corollary}[Partially parallelized exact sequential decomposition using polylogarithmic depth decomposition ]
\label{lemma: partially parallelized sequential decomposition bis}
    Any $n$-qubit diagonal unitary is implementable through its sequential decomposition with a partially parallelized quantum circuit of depth $\mathcal{O}(n \log(n)^3 2^n/m+\log(m/n))$ and size $\mathcal{O}(n\log(n)^4 2^n+m)$ using $ m \in [n+2, \mathcal{O}(n2^n)]$ ancilla qubits. 
\end{corollary}
\begin{proof}
Lemma \ref{lemma: Exact sequential decomposition with one ancilla} with the partial parallelization Theorem \ref{thm : adjustable-depth with ancilla}.
\end{proof}

\begin{corollary}[Partially parallelized exact Walsh-Hadamard decomposition]
\label{lemma: partially parallelized walsh hadamard decomposition with one ancilla}
    Any $n$-qubit diagonal unitary is implementable through its Walsh-Hadamard decomposition with a partially parallelized quantum circuit of depth $\mathcal{O}(n^2 2^n/m+\log(m/n))$ and size $\mathcal{O}( n2^n+m)$  using $ m \in [n, \mathcal{O}(n2^n)]$ ancilla qubits. 
\end{corollary}
\begin{proof}
Lemma \ref{lemma: Exact walsh-hadamard decomposition} with the partial parallelization Theorem \ref{thm : adjustable-depth with ancilla}.
\end{proof}

\begin{lemma}[Lemma 11 \cite{sun2023asymptotically}]
Any $n$-qubit diagonal unitary can be implemented by a quantum circuit of depth $\mathcal{O}(2^n/n)$ and size $\mathcal{O}(2^n)$ without ancillary qubits.
\end{lemma}

\begin{lemma}[Lemma 20 \cite{sun2023asymptotically}]
For any $m\in [2n, 2^n/n]$, any $n$-qubit diagonal unitary can be implemented by a quantum circuit of depth $\mathcal{O}(2^n/m +\log(m))$ and size $\mathcal{O}(2^n+nm)$ with $m$ ancillary qubits.
\end{lemma}

\begin{lemma}[Lemma 20 \cite{sun2023asymptotically} with a maximum number of ancilla qubits]
Any $n$-qubit diagonal unitary can be implemented by a quantum circuit of depth $\mathcal{O}(n)$ and size $\mathcal{O}(2^n)$ with $m=2^n/n$ ancillary qubits.
\end{lemma}

\subsubsection{Non-unitary diagonal operators}

The following corollaries summarize the complexity of implementing exactly a block-encoding $\hat{U}_D$ of any $n$-qubit non-unitary diagonal operator $\hat{D}$ through the controlled-diagonal unitaries $e^{\pm i\arcsin(\hat{D}/(\alpha d_{\max}))}$ with $\alpha \geq 1$.

\begin{corollary}[Exact block-encoding using a sequential decomposition and one ancilla qubit]
\label{Exact block-encoding using a sequential decomposition and one ancilla}
     For any $\alpha \geq 1$, any $n$-qubit non-unitary diagonal operator $\hat{D}$ can be $(\alpha d_{\max},1,0)$-block-encoded with a quantum circuit of depth $\mathcal{O}(n2^{n})$ and size $\mathcal{O}(n2^{n})$ using one ancilla qubit.
\end{corollary}
\begin{proof}
    Lemma \ref{Exact_sequential_decomposition_without_ancilla} and the block-encoding Theorem \ref{thm: block-encoding}.
\end{proof}

\begin{corollary}[Exact sequential block-encoding with two ancilla qubits]
\label{lemma: Exact block-encoding sequential decomposition with one ancilla}
 For any $\alpha \geq 1$, any $n$-qubit non-unitary diagonal operator $\hat{D}$ can be $(\alpha d_{\max},2,0)$-block-encoded with a quantum circuit of depth $\mathcal{O}(\log(n)^3 2^{n})$ and size  $\mathcal{O}(n\log(n)^4 2^{n})$ using two ancilla qubits.
\end{corollary}
\begin{proof}
Lemma \ref{lemma: Exact sequential decomposition with one ancilla} and the block-encoding Theorem \ref{thm: block-encoding}: the $(n-1)$-controlled phase gates are controlled by one ancilla qubit, becoming $n$-controlled phase gates. The last ancilla qubit is used to implement the $n$-controlled phase gates exactly with Corollary 1 of \cite{claudon2024polylogarithmic}.
\end{proof}

\begin{corollary}[Exact Walsh-Hadamard block-encoding with one ancilla]
\label{lemma: Exact walsh-hadamard block-encoding with one ancilla }
For any $\alpha \geq 1$, any $n$-qubit non-unitary diagonal operator $\hat{D}$ can be $(\alpha d_{\max},1,0)$-block-encoded with a quantum circuit of depth $\mathcal{O}(2^n)$ and size  $\mathcal{O}(2^n)$ using one ancilla qubit.

\end{corollary}
\begin{proof}
Lemma \ref{lemma: Exact walsh-hadamard decomposition} and the block-encoding Theorem \ref{thm: block-encoding}
\end{proof}

\begin{corollary}[Fully parallelized block-encoding with exact sequential decomposition]
\label{cor: Fully parallelized block-encoding with exact sequential decomposition}
  For any $\alpha \geq 1$, any $n$-qubit non-unitary diagonal operator $\hat{D}$ can be $(\alpha d_{\max},m,0)$-block-encoded with a quantum circuit of depth $\mathcal{O}(n)$ and size $\mathcal{O}(n2^{n})$ using $m=\mathcal{O}(n2^n)$ ancilla qubits.
 
\end{corollary}
\begin{proof}
Lemma \ref{Exact_sequential_decomposition_without_ancilla} or Lemma \ref{lemma: Exact sequential decomposition with one ancilla} associated to the full parallelization Theorem \ref{thm : full parallelization}  and the block-encoding Theorem \ref{thm: block-encoding}
\end{proof}

\begin{corollary}[Fully parallelized block-encoding with exact Walsh-Hadamard decomposition]
\label{lemma: fully parallelized Walsh-Hadamard decomposition}
For any $\alpha \geq 1$, any $n$-qubit non-unitary diagonal operator $\hat{D}$ can be $(\alpha d_{\max},m,0)$-block-encoded with a quantum circuit of depth $\mathcal{O}(n)$ and size $\mathcal{O}(n2^{n})$ using $m=\mathcal{O}(n2^n)$ ancilla qubits.
\end{corollary}
\begin{proof}
Lemma \ref{lemma: Exact walsh-hadamard decomposition} with the full parallelization Theorem \ref{thm : full parallelization} and the block-encoding Theorem \ref{thm: block-encoding}.
\end{proof}

\begin{corollary}[Partially parallelized block-encoding with exact sequential decomposition]
\label{corollary: partially parallelized block-encoding with exact sequential decomposition}
For any $\alpha \geq 1$, any $n$-qubit non-unitary diagonal operator $\hat{D}$ can be $(\alpha d_{\max},m,0)$-block-encoded with a quantum circuit of depth $\mathcal{O}(n^2 2^n/m+\log(m/n))$ and size $\mathcal{O}(n2^n+m)$ using $ m \in [\Omega(n), \mathcal{O}(n2^n)]$ ancilla.
\end{corollary}
\begin{proof}
Lemma \ref{Exact_sequential_decomposition_without_ancilla} with the partial parallelization Theorem \ref{thm : adjustable-depth with ancilla} and the block-encoding Theorem \ref{thm: block-encoding}
\end{proof}

\begin{corollary}[Partially parallelized block-encoding with exact sequential decomposition using polylogarithmic depth decomposition]
\label{lemma: partially parallelized block-encoding sequential decomposition using polylog MCphase}
For any $\alpha \geq 1$, any $n$-qubit non-unitary diagonal operator $\hat{D}$ can be $(\alpha d_{\max},m,0)$-block-encoded with a quantum circuit of depth $\mathcal{O}(n \log(n)^3 2^n/m+\log(m/n))$ and size $\mathcal{O}(n\log(n)^4 2^n+m)$ using $m$ ancilla qubits with $m=\Omega(n)$ and $m=\mathcal{O}(2^n/n)$. 
\end{corollary}
\begin{proof}
Lemma \ref{lemma: Exact sequential decomposition with one ancilla} with the partial parallelization Theorem \ref{thm : adjustable-depth with ancilla} and the block-encoding Theorem \ref{thm: block-encoding}
\end{proof}

\begin{corollary}[Partially parallelized block-encoding with exact Walsh-Hadamard decomposition]
\label{Partially parallelized block-encoding with exact Walsh-Hadamard decomposition}
For any $\alpha \geq 1$, any $n$-qubit non-unitary diagonal operator $\hat{D}$ can be $(\alpha d_{\max},m,0)$-block-encoded with a quantum circuit of depth $\mathcal{O}(n^2 2^n/m+\log(m/n))$ and size $\mathcal{O}(n2^n+m)$ using $m$ ancilla qubits with $m=\Omega(n)$ and $m=\mathcal{O}(2^n/n)$. 
\end{corollary}
\begin{proof}
Lemma \ref{lemma: Exact walsh-hadamard decomposition} with the partial parallelization Theorem \ref{thm : adjustable-depth with ancilla} and the block-encoding Theorem \ref{thm: block-encoding}.
\end{proof}

\begin{corollary}[Lemma 11 \cite{sun2023asymptotically} + block-encoding]
For any $\alpha \geq 1$, any $n$-qubit non-unitary diagonal operator $\hat{D}$ can be $(\alpha d_{\max},m,0)$-block-encoded with a quantum circuit of depth $\mathcal{O}(2^n/n)$ and size $\mathcal{O}(2^n)$ using $n$ ancilla qubits. 
\end{corollary}
\begin{proof}
The Lemma 11 \cite{sun2023asymptotically} associated with the block-encoding Theorem \ref{thm: block-encoding} where each $\hat{U}_j$ is not a sequential or Walsh-Hadamard operator, but a primitive quantum gates of the decomposition given by the quantum circuits of \cite{sun2023asymptotically}. Indeed, any $n$-qubit unitary operation with size $s(n)$ and depth $d(n)$ can be controlled by a qubit $q_A$ with size $\mathcal{O}(s(n)+n)$ and depth $d(n)+\mathcal{O}(log(n))$ using $n-1$ copies of $q_A$ to control in parallel the different gates. 
\end{proof}

\begin{corollary}[Lemma 20  \cite{sun2023asymptotically} + block-encoding]
For any $\alpha \geq 1$, any $n$-qubit non-unitary diagonal operator $\hat{D}$ can be $(\alpha d_{\max},m,0)$-block-encoded with a quantum circuit of depth $\mathcal{O}(2^n/m +\log(m))$ and size $\mathcal{O}(n2^n)$ using $m$ ancilla qubits with $m=\Omega(n)$ and $m=\mathcal{O}(2^n/n)$. 
\end{corollary}
\begin{proof}
Lemma 20  \cite{sun2023asymptotically}  associated to the block-encoding Theorem \ref{thm: block-encoding}.
\end{proof}

\begin{corollary}[Lemma 20  \cite{sun2023asymptotically} with a maximum number of ancilla qubits + block-encoding]
For any $\alpha \geq 1$, any $n$-qubit non-unitary diagonal operator $\hat{D}$ can be $(\alpha d_{\max},m,0)$-block-encoded with a quantum circuit of depth $\mathcal{O}(n)$ and size $\mathcal{O}(n2^n)$ using $m=\mathcal{O}(2^n/n)$ ancilla qubits. 
\end{corollary}
\begin{proof}
Lemma 20  \cite{sun2023asymptotically} associated with the block-encoding Theorem \ref{thm: block-encoding}.
\end{proof}

\subsection{Approximate methods}
\label{sec:approximate methods}

\subsubsection{For diagonal unitaries depending on differentiable functions}
\label{For diagonal unitaries depending on differentiable functions} 
In the following, we consider a $n$-qubit diagonal unitary $\hat{U}=\sum_{x=0}^{N-1}e^{if(x/N)}\ket{x}\bra{x}$, with $N=2^n$ depending on a real-valued function $f$ defined on $[0,1]$ with bounded first derivative.

\begin{corollary}[Approximate sequential decomposition without ancilla]
\label{cor : Sequential no ancilla : approximé}
Any $n$-qubit diagonal unitary depending on a real-valued function $f$ defined on $[0,1]$ with bounded first derivative is implementable up to an error $\epsilon>0$ in spectral norm through its sequential decomposition with a quantum circuit of depth $\mathcal{O}(\log(1/\epsilon)/\epsilon)$ and size $\mathcal{O}(\log(1/\epsilon)/\epsilon)$ without ancilla qubits.
\end{corollary}
\begin{proof}
This corollary is a direct consequence of Lemma \ref{Exact_sequential_decomposition_without_ancilla} and the approximate Theorem \ref{thm: Approximation theorem}.
\end{proof}

\begin{corollary}[Approximate sequential decomposition with one ancilla]
\label{cor:  approximate sequential decomposition with one ancilla}
Any $n$-qubit diagonal unitary depending on a real-valued function $f$ defined on $[0,1]$ with bounded first derivative is implementable up to an error $\epsilon>0$ in spectral norm through its sequential decomposition with a quantum circuit of depth $\mathcal{O}(\log(\log(1/\epsilon))^3/\epsilon)$ and size $\mathcal{O}(\log(1/\epsilon)\log(\log(1/\epsilon))^4/\epsilon)$ using one ancilla qubit.
\end{corollary}
\begin{proof}
This corollary is a direct consequence of Lemma \ref{lemma: Exact sequential decomposition with one ancilla} and the approximate Theorem \ref{thm: Approximation theorem}.
\end{proof}

\begin{corollary}[Approximate Walsh-Hadamard decomposition without ancilla \cite{welch2014efficient}]
\label{cor: approximate walsh-hadamard decomposition}
Any $n$-qubit diagonal unitary depending on a real-valued function $f$ defined on $[0,1]$ with bounded first derivative is implementable up to an error $\epsilon>0$ in spectral norm through its Walsh-Hadamard decomposition with a quantum circuit of depth $\mathcal{O}(1/\epsilon)$ and size $\mathcal{O}(1/\epsilon)$ without ancilla qubit.
\end{corollary}
\begin{proof}
This corollary is a direct consequence of Lemma \ref{lemma: Exact walsh-hadamard decomposition} and the approximate Theorem \ref{thm: Approximation theorem}.
\end{proof}

\begin{corollary}[Fully parallelized approximate sequential decomposition]
\label{lemma: fully parallelized approximate sequential decomposition}
Any $n$-qubit diagonal unitary depending on a real-valued function $f$ defined on $[0,1]$ with bounded first derivative is implementable up to an error $\epsilon>0$ in spectral norm through its sequential decomposition with a quantum circuit of depth $\mathcal{O}(\log(1/\epsilon))$ and size $\mathcal{O}(\log(1/\epsilon)/\epsilon)$ using $\mathcal{O}(\log(1/\epsilon)/\epsilon)$ ancilla qubits.
\end{corollary}
\begin{proof}
This corollary is a direct consequence of corollary \ref{lemma: fully parallelized sequential decomposition with one ancilla} associated with the approximate Theorem \ref{thm: Approximation theorem}.
\end{proof}

\begin{corollary}[Fully parallelized approximate Walsh-Hadamard decomposition]
\label{lemma: fully parallelized approximate walsh-hadamard decomposition}
Any $n$-qubit diagonal unitary depending on a real-valued function $f$ defined on $[0,1]$ with bounded first derivative is implementable up to an error $\epsilon>0$ in spectral norm through its Walsh-Hadamard decomposition with a quantum circuit of depth $\mathcal{O}(\log(1/\epsilon))$ and size $\mathcal{O}(\log(1/\epsilon)/\epsilon)$ using $\mathcal{O}(\log(1/\epsilon)/\epsilon)$ ancilla qubits.
\end{corollary}
This corollary is a direct consequence of corollary \ref{lemma: fully parallelized walsh-hadamard decomposition} with the approximate Theorem \ref{thm: Approximation theorem}.

\begin{corollary}[Partially parallelized approximate sequential decomposition using linear depth decomposition]
\label{cor: partially parallelized approximate sequential decomposition}
Any $n$-qubit diagonal unitary depending on a real-valued function $f$ defined on $[0,1]$ with bounded first derivative is implementable up to an error $\epsilon>0$ in spectral norm through its sequential decomposition with a quantum circuit of depth $\mathcal{O}(\log(1/\epsilon)^2/(\epsilon m)+\log(m/\log(1/\epsilon)))$ and size $\mathcal{O}(\log(1/\epsilon)/\epsilon+m)$ using $m$ ancilla qubits with  $m=\Omega(\log(1/\epsilon))$ and $m=\mathcal{O}(\log(1/\epsilon)/\epsilon)$.
\end{corollary}
\begin{proof}
This corollary is a direct consequence of corollary \ref{lemma: partially parallelized sequential decomposition} with the approximate Theorem \ref{thm: Approximation theorem}.
\end{proof}

\begin{corollary}[Partially parallelized approximate sequential decomposition using polylogarithmic depth decomposition]
\label{cor: partially parallelized approximate sequential decomposition bis}
Any $n$-qubit diagonal unitary depending on a real-valued function $f$ defined on $[0,1]$ with bounded first derivative is implementable up to an error $\epsilon>0$ in spectral norm through its sequential decomposition with a quantum circuit of depth $\mathcal{O}(\log(1/\epsilon)\log(\log(1/\epsilon))^3/(\epsilon m)+\log(m/\log(1/\epsilon)))$ and size $\mathcal{O}(\log(1/\epsilon)\log(\log(1/\epsilon))^4/\epsilon+m)$ using $m$ ancilla qubits with  $m=\Omega(\log(1/\epsilon))$ and $m=\mathcal{O}(\log(1/\epsilon)/\epsilon)$.
\end{corollary}
\begin{proof}
This corollary is a direct consequence of corollary \ref{lemma: partially parallelized sequential decomposition bis} with the approximate Theorem \ref{thm: Approximation theorem}.
\end{proof}

\begin{corollary}[Partially parallelized approximate Walsh-Hadamard decomposition]
\label{cor: partially parallelized approximate walsh-hadamard decomposition with one ancilla}
Any $n$-qubit diagonal unitary depending on a real-valued function $f$ defined on $[0,1]$ with bounded first derivative is implementable up to an error $\epsilon>0$ in spectral norm through its Walsh-Hadamard decomposition with a quantum circuit of depth $\mathcal{O}(\log(1/\epsilon)^2/(\epsilon m)+\log(m/\log(1/\epsilon)))$ and size $\mathcal{O}(\log(1/\epsilon)/\epsilon+m)$ using $m$ ancilla qubits with  $m=\Omega(\log(1/\epsilon))$ and $m=\mathcal{O}(\log(1/\epsilon)/\epsilon)$.
\end{corollary}
\begin{proof}
This corollary is a direct consequence of corollary  \ref{lemma: partially parallelized walsh hadamard decomposition with one ancilla} with the approximate Theorem \ref{thm: Approximation theorem}.
\end{proof}

\begin{corollary}[Walsh-recursive]
\label{cor: lemma11+approximate}
Any $n$-qubit diagonal unitary depending on a real-valued function $f$ defined on $[0,1]$ with bounded first derivative is implementable up to an error $\epsilon>0$ in spectral norm with a quantum circuit of depth $\mathcal{O}(1/(\epsilon \log(1/\epsilon)))$ and size $\mathcal{O}(1/\epsilon)$ without ancilla qubits.
\end{corollary}
\begin{proof}
This corollary is a direct consequence of Lemma 11 \cite{sun2023asymptotically} with the approximate Theorem \ref{thm: Approximation theorem}.
\end{proof}

\begin{corollary}[Walsh-optimized adjustable-depth]
\label{cor:lemma20+approximate}
For any $\epsilon>0$, any $n$-qubit diagonal unitary depending on a real-valued function $f$ defined on $[0,1]$ with bounded first derivative is implementable up to an error $\epsilon$ in spectral norm with a quantum circuit of depth $\mathcal{O}(1/(\epsilon m)+\log(m))$ and size $\mathcal{O}(1/\epsilon+\log(1/\epsilon)m)$ using $m$ ancilla qubits with $m =\Omega(\log(1/\epsilon))$ and $m=\mathcal{O}(1/(\epsilon \log(1/\epsilon)))$.
\end{corollary}
\begin{proof}
This corollary is a direct consequence of Lemma 20 \cite{sun2023asymptotically} with the approximate Theorem \ref{thm: Approximation theorem}.
\end{proof}

\begin{corollary}[Walsh-optimized fully parallelized]
\label{cor: lemma20 max qubit+approximate}
For any $\epsilon>0$, any $n$-qubit diagonal unitary depending on a real-valued function $f$ defined on $[0,1]$ with bounded first derivative is implementable up to an error $\epsilon$ in spectral norm with a quantum circuit of depth $\mathcal{O}(\log(1/\epsilon))$ and size $\mathcal{O}(1/\epsilon)$ with $m=\mathcal{O}(1/(\epsilon \log(1/\epsilon)))$ ancilla qubits.
\end{corollary}
\begin{proof}
This corollary is a direct consequence of Lemma 20 \cite{sun2023asymptotically} for a maximum number of ancilla qubits with the approximate Theorem \ref{thm: Approximation theorem}.
\end{proof}

\subsubsection{For non-unitary diagonal operators depending on differentiable functions}
\label{sec: non-unitary diagonal operators depending on differentiable functions}

The following corollaries summarize the complexity of implementing approximately a block-encoding $\hat{U}_D$ of a $n$-qubit non-unitary diagonal operator $\hat{D}=\sum_{x=0}^{N-1}f(x/N)\ket{x}\bra{x}$, depending on a real-valued function $f$ defined on $[0,1]$ with bounded first derivative, through the controlled-diagonal unitaries $e^{\pm i\arcsin(\hat{D}/(\alpha d_{\max}))}$ with $\alpha > 1$.

\begin{corollary}[Approximate block-encoding using a sequential decomposition and one ancilla qubit]
\label{cor: approximate block-encoding using a sequential decomposition and one ancilla}
     For any $\alpha > 1$, any $n$-qubit non-unitary diagonal operator $\hat{D}$  depending on a real-valued function $f$ defined on $[0,1]$ with bounded first derivative can be $(\alpha d_{\max},1,\epsilon)$-block-encoded with a quantum circuit of depth $\mathcal{O}(\log(1/\epsilon)/\epsilon)$ and size $\mathcal{O}(\log(1/\epsilon)/\epsilon)$ using one ancilla qubit.
\end{corollary}
\begin{proof}
This corollary is a direct consequence of corollary \ref{cor : Sequential no ancilla : approximé} applied on $\hat{U}=\sum_{x=0}^{N-1}e^{ig(x/N)}\ket{x}\bra{x}$ with $g(x)=\arcsin(f(x)/(\alpha \|f\|_{\infty}))$ such that $\|g'\|_{\infty}\leq \|f'\|_{\infty}/(\|f\|_{\infty} \sqrt{\alpha^2-1})$ and the block-encoding Theorem \ref{thm: block-encoding}.
\end{proof}

\begin{corollary}[Approximate sequential block-encoding with two ancilla qubits]
\label{cor: approximate  block-encoding sequential decomposition  with two ancilla}
For any $\alpha > 1$, any $n$-qubit non-unitary diagonal operator $\hat{D}$  depending on a real-valued function $f$ defined on $[0,1]$ with bounded first derivative can be $(\alpha d_{\max},1,\epsilon)$-block-encoded with a quantum circuit of depth $\mathcal{O}(\log(\log(1/\epsilon))^3/\epsilon)$ and size $\mathcal{O}(\log(1/\epsilon)(\log(\log(1/\epsilon))^4/\epsilon)$ using two ancilla qubits.
\end{corollary}
\begin{proof}
This corollary is a direct consequence of corollary \ref{cor:  approximate sequential decomposition with one ancilla} and the block-encoding Theorem \ref{thm: block-encoding}.
\end{proof}

\begin{corollary}[Approximate Walsh-Hadamard block-encoding with one ancilla]
\label{cor: approximate walsh-hadamard block-encoding with one ancilla}
For any $\alpha > 1$, any $n$-qubit non-unitary diagonal operator $\hat{D}$  depending on a real-valued function $f$ defined on $[0,1]$ with bounded first derivative can be $(\alpha d_{\max},1,\epsilon)$-block-encoded with a quantum circuit of depth $\mathcal{O}(1/\epsilon)$ and size $\mathcal{O}(1/\epsilon)$ using one ancilla qubit.
\end{corollary}
\begin{proof}
This corollary is a direct consequence of corollary \ref{cor: approximate walsh-hadamard decomposition} and the block-encoding Theorem \ref{thm: block-encoding}.
\end{proof}

\begin{corollary}[Fully parallelized block-encoding with approximate sequential decomposition]

\label{cor: Fully parallelized block-encoding with approximate sequential decomposition}
For any $\alpha > 1$, any $n$-qubit non-unitary diagonal operator $\hat{D}$  depending on a real-valued function $f$ defined on $[0,1]$ with bounded first derivative can be $(\alpha d_{\max},1,\epsilon)$-block-encoded with a quantum circuit of depth $\mathcal{O}(\log(1/\epsilon))$ and size $\mathcal{O}(\log(1/\epsilon)/\epsilon)$ using $m=\mathcal{O}(\log(1/\epsilon)/\epsilon)$ ancilla qubits.
\end{corollary}
\begin{proof}
This corollary is a direct consequence of corollary \ref{lemma: fully parallelized approximate sequential decomposition} and the block-encoding Theorem \ref{thm: block-encoding}.
\end{proof}

\begin{corollary}[Fully parallelized block-encoding with approximate Walsh-Hadamard decomposition]
\label{lemma: fully parallelized approximate Walsh-Hadamard decomposition}
For any $\alpha > 1$, any $n$-qubit non-unitary diagonal operator $\hat{D}$  depending on a real-valued function $f$ defined on $[0,1]$ with bounded first derivative can be $(\alpha d_{\max},1,\epsilon)$-block-encoded with a quantum circuit of depth $\mathcal{O}(\log(1/\epsilon))$ and size $\mathcal{O}(\log(1/\epsilon)/\epsilon)$ using $m=\mathcal{O}(\log(1/\epsilon)/\epsilon)$ ancilla qubits.
\end{corollary}
\begin{proof}
This corollary is a direct consequence of corollary  \ref{lemma: fully parallelized approximate walsh-hadamard decomposition} and the block-encoding Theorem \ref{thm: block-encoding}.
\end{proof}

\begin{corollary}[Partially parallelized block-encoding with approximate sequential decomposition]
\label{corollary: partially parallelized block-encoding with approximate sequential decomposition}
For any $\alpha > 1$, any $n$-qubit non-unitary diagonal operator $\hat{D}$  depending on a real-valued function $f$ defined on $[0,1]$ with bounded first derivative can be $(\alpha d_{\max},1,\epsilon)$-block-encoded with a quantum circuit of depth $\mathcal{O}(\log(1/\epsilon)^2/(\epsilon m) +\log(m/\log(1/\epsilon)) )$ and size $\mathcal{O}(\log(1/\epsilon)/\epsilon+m)$ using $m$ ancilla qubits with $m=\Omega(\log(1/\epsilon))$ and $ m=\mathcal{O}(\log(1/\epsilon)/\epsilon)$.
\end{corollary}
\begin{proof}
This corollary is a direct consequence of corollary  \ref{cor: partially parallelized approximate sequential decomposition} and the block-encoding Theorem \ref{thm: block-encoding}.
\end{proof}

\begin{corollary}[Partially parallelized block-encoding with approximate sequential decomposition using polylogarithmic depth decomposition]
\label{cor: partially parallelized approximate block-encoding sequential decomposition using polylog MCphase}
For any $\alpha > 1$, any $n$-qubit non-unitary diagonal operator $\hat{D}$  depending on a real-valued function $f$ defined on $[0,1]$ with bounded first derivative can be $(\alpha d_{\max},1,\epsilon)$-block-encoded with a quantum circuit of depth $\mathcal{O}(\log(1/\epsilon)\log(\log(1/\epsilon))^3/(\epsilon m) +\log(m/\log(1/\epsilon)))$ and size $\mathcal{O}(\log(1/\epsilon)\log(\log(1/\epsilon))^4/\epsilon+m)$ using $m$ ancilla qubits with  $m=\Omega(\log(1/\epsilon))$ and $ m=\mathcal{O}(\log(1/\epsilon)/\epsilon)$.
\end{corollary}
\begin{proof}
This corollary is a direct consequence of corollary  \ref{cor: partially parallelized approximate sequential decomposition bis} and the block-encoding Theorem \ref{thm: block-encoding}.
\end{proof}

\begin{corollary}[Partially parallelized block-encoding with approximate Walsh-Hadamard decomposition]
\label{Partially parallelized block-encoding with approximate Walsh-Hadamard decomposition}
For any $\alpha > 1$, any $n$-qubit non-unitary diagonal operator $\hat{D}$  depending on a real-valued function $f$ defined on $[0,1]$ with bounded first derivative can be $(\alpha d_{\max},1,\epsilon)$-block-encoded with a quantum circuit of depth $\mathcal{O}(\log(1/\epsilon)^2/(\epsilon m) +\log(m/\log(1/\epsilon)))$ and size $\mathcal{O}(\log(1/\epsilon)/\epsilon+m)$ using $m$ ancilla qubits with $m=\Omega(\log(1/\epsilon))$ and $ m=\mathcal{O}(\log(1/\epsilon)/\epsilon)$.
\end{corollary}
\begin{proof}
This corollary is a direct consequence of corollary  \ref{cor: partially parallelized approximate walsh-hadamard decomposition with one ancilla} and the block-encoding Theorem \ref{thm: block-encoding}.
\end{proof}

\begin{corollary}[Walsh-recursive]
\label{walsh-recursive nonunitary}
For any $\alpha > 1$, any $n$-qubit non-unitary diagonal operator $\hat{D}$  depending on a real-valued function $f$ defined on $[0,1]$ with bounded first derivative can be $(\alpha d_{\max},1,\epsilon)$-block-encoded with a quantum circuit of depth $\mathcal{O}(1/(\epsilon \log(1/\epsilon) ))$ and size $\mathcal{O}(1/\epsilon)$ using $m=\mathcal{O}(\log(1/\epsilon))$ ancilla qubits.
\end{corollary}
\begin{proof}
This corollary is a direct consequence of corollary \ref{cor: lemma11+approximate} and the block-encoding Theorem \ref{thm: block-encoding}.
\end{proof}

\begin{corollary}[Walsh-optimized adjustable-depth]
\label{Walsh-optimized adjustable-depth non-unitary}
For any $\alpha > 1$, any $n$-qubit non-unitary diagonal operator $\hat{D}$  depending on a real-valued function $f$ defined on $[0,1]$ with bounded first derivative can be $(\alpha d_{\max},1,\epsilon)$-block-encoded with a quantum circuit of depth $\mathcal{O}(1/(\epsilon m)+\log(m))$ and size $\mathcal{O}(\log(1/\epsilon)/\epsilon)$ using $m$ ancilla qubits with $m=\Omega(\log(1/\epsilon))$ and $ m=\mathcal{O}(1/(\epsilon \log(1/\epsilon)))$.
\end{corollary}
\begin{proof}
This corollary is a direct consequence of corollary \ref{cor:lemma20+approximate} and the block-encoding Theorem \ref{thm: block-encoding}.
\end{proof}

\begin{corollary}[Walsh-optimized fully parallelized]
\label{Walsh-optimized fully parallelized non-unitary}
For any $\alpha > 1$, any $n$-qubit non-unitary diagonal operator $\hat{D}$  depending on a real-valued function $f$ defined on $[0,1]$ with bounded first derivative can be $(\alpha d_{\max},1,\epsilon)$-block-encoded with a quantum circuit of depth $\mathcal{O}(\log(1/\epsilon))$ and size $\mathcal{O}(1/\epsilon)$ using $m=\mathcal{O}(1/(\epsilon \log(1/\epsilon)))$ ancilla qubits.
\end{corollary}
\begin{proof}

This corollary is a direct consequence of corollary \ref{cor: lemma20 max qubit+approximate} and the block-encoding Theorem \ref{thm: block-encoding}.
\end{proof}

\subsubsection{Fourier GQSP methods}
\label{sec: Fourier GQSP methods}

First, we recall a result established by Motlagh et al. \cite{motlagh2024generalized}: 

\begin{corollary}[Corollary 14 of \cite{motlagh2024generalized}]
\label{cor: gqsp_diag}
Let $f (x) = \sum_{k=0}^{d}\alpha_ke^{2i\pi kx/N}$ be a function such that $\|f\|_{\infty}\leq 1$, it is possible to $(1,1,0)$-block encode the $n$-qubit operator $\sum_{k=0}^d f(x)\ket{x}\bra{x}$ with $O(dn)$ single- and two-qubit gates. 
\end{corollary}

A direct consequence of this corollary is that, when a function $f$ can be approximated by a Fourier series that converges exponentially fast, then it is possible to reach an accuracy $\epsilon>0$ with only a logarithmic number of terms $\mathcal{O}(\log(1/\epsilon))$. 
\begin{corollary}
\label{cor:GQSP approximate}
Any $n$-qubit non-unitary diagonal operator $\hat{D}$ depending on a real-valued function $f$ defined on $[0,1]$, which is periodic $f(0)=f(1)$, analytic and there exist two constants $C>0$ and $R>1$ such that for all $M\in\mathbb{N}$, $\|\partial_{x}^{(M)}f\|_{L^1}\leq C M!/R^M$, can be $(d_{\max},1,\epsilon)$-block-encoded with a quantum circuit of depth $\mathcal{O}(n\log(1/\epsilon))$ and size $\mathcal{O}(n\log(1/\epsilon))$ using one ancilla qubit.
\end{corollary}
\begin{proof}
The proof comes from the previous corollary \ref{cor: gqsp_diag} and the error between the function and its truncated Fourier series. We present an alternative version of Lemma 19 from \cite{childs2022quantum}: when a function $f$ is periodic, analytic and there exist two constants $C>0$ and $R>1$ such that for all $M\in\mathbb{N}$, $\|\partial_{x}^{(M)}f\|_{L^1}\leq C M!/R^M$, then for all $M\in\mathbb{N}$:
\begin{equation}
\begin{split}
     \|f-\tilde{f}_M\|_\infty& \leq2\frac{\|\partial_{x}^{(M)}f\|_{L^1}}{(2\pi)^{M}(M-1) M^{(M-1)}} \leq2C\frac{M!}{(2\pi R)^{M}(M-1) M^{(M-1)}} \\& \leq 2C\frac{\sqrt{2\pi M^3} }{(2\pi e R)^{M}(M-1)} e^{\frac{1}{12M}}
\end{split}
\label{Fourier convergence}
    \end{equation}
where $\tilde{f}_M(x)=\sum_{k=-M}^M  a_{k}^g e^{2 i\pi kx}$ and $a_{k}^f=\int_{0}^1f(x)e^{-2i\pi kx}dx$ the $k$-th Fourier coefficient of $g$. Notice that $\|f-\tilde{f}_M\|_\infty \leq \epsilon/2$ implies that $M$ scales as $\mathcal{O}(\log(1/\epsilon))$.

Then, the generalized quantum signal processing protocol \cite{motlagh2024generalized}  allows the implementation of polynomial of operators without additional restrictions on the polynomial other than $|P(e^{i\theta})|\leq1$ for all $\theta\in[0,2\pi[$. We consider the normalized function $\tilde{\tilde{f}}_M=\tilde{f}_M/(\|f\|_{\infty}+\epsilon/2)$ such that $\|\tilde{\tilde{f}}_M\|_{\infty}\leq 1$, and $\|\frac{f}{\|f\|_{\infty}}-\tilde{\tilde{f}}_M\|_{\infty}\leq\frac{\epsilon}{\|f\|_{\infty}}$. The GQSP protocol, using the polynomial $P_M(X)=\sum_{k=0}^{2M}(\alpha_{k-M}/(\|f\|_{\infty}+\epsilon/2))X^k$ and the unitary $\hat{U}_{\omega}=\sum_{j=0}^{N-1}e^{2i\pi j/N}\ket{j}\bra{j}$, produces a $(1,1,0)$-block encoding $\hat{W}$ of $\tilde{D}_M=\sum_{k=0}^d \tilde{\tilde{f}}_M(x)\ket{x}\bra{x}$. This operator $\hat{W}$ is a $(\|f\|_{\infty},1,\epsilon)$-block-encoding of $\hat{D}$: 
\begin{equation}
\|\hat{D}-\|f\|_{\infty}(\bra{0}\otimes \hat{I}_n)\hat{W}(\ket{0} \otimes \hat{I}_n)\|_2\leq\|f\|_{\infty}\|\frac{f}{\|f\|_{\infty}}-\tilde{\tilde{f}}_M\|_{\infty}\leq \epsilon
\end{equation}
The associated quantum circuit has a size and depth of $O(Mn)$, using one ancilla qubit. 
\end{proof}

\begin{corollary}
\label{cor:GQSP approximate parallel}
Any $n$-qubit non-unitary diagonal operator $\hat{D}$ depending on a real-valued function $f$ defined on $[0,1]$, which is periodic $f(0)=f(1)$, analytic and there exist two constants $C>0$ and $R>1$ such that for all $M\in\mathbb{N}$, $\|\partial_{x}^{(M)}f\|_{L^1}\leq C M!/R^M$, can be $(f_{\max},n,\epsilon)$-block-encoded with a quantum circuit of depth $\mathcal{O}(\log(1/\epsilon)\log(n))$ and size $\mathcal{O}(n\log(1/\epsilon))$ using $n$ ancilla qubit.
\end{corollary}
\begin{proof}
The previous corollary \ref{cor:GQSP approximate} enables to implement diagonal operators using a polynomial of  $\hat{U}_{\omega}=\sum_{j=0}^{N-1}e^{2i\pi j/N}\ket{j}\bra{j}$. The operator  $\hat{U}_{\omega}=\sum_{j=0}^{N-1}e^{2i\pi j/N}\ket{j}\bra{j}$ is composed of $n$ phase gates, and the control $\hat{U}_{\omega}$ is made of $n$ controlled phase gates. For each controlled $\hat{U}_{\omega}$, it is possible to parallelize the implementation of the control phase gate by using $n-1$ additional ancilla qubits and the copy unitary presented in Lemma \ref{copy lemma}. 
\end{proof}

\subsection{Sparse diagonal operators}
\label{sec: sparse diagonal operator}

\subsubsection{Sparse diagonal unitaries}\label{sparse du}
The following Lemmas and corollaries summarize the complexity of implementing any $n$-qubit diagonal unitary with a decomposition containing only $s$ sequential operators or $s$ Walsh-Hadamard operators.

\begin{lemma}[Sparse sequential decomposition without ancilla]
\label{lemma : Sequential no ancilla : Sparse}
   Any $n$-qubit diagonal unitary with a $s$-sparse sequential decomposition is exactly implementable with a quantum circuit of depth $\mathcal{O}(ns)$ and size $\mathcal{O}(ns)$ without using ancilla qubits.
\end{lemma}
\begin{proof}
    The quantum circuit is composed of $s$ operator $\hat{U}_j$ defined in Eq. \eqref{sequential decomposition}. Each of the $s$ sequential operator contains at worst $2n$ $\hat{X}$-Pauli gates and one $\Lambda_{\{0,...,n-2\}}(\hat{P}(\theta_j))$ which is implemented using the scheme of C.Gidney \cite{Craig}. It is an exact method for multi-controlled gates with size and depth linear in $n$ without using ancilla qubits.
\end{proof}

\begin{lemma}[Approximate sparse sequential decomposition without ancilla]
\label{lemma : Sequential no ancilla : Sparse approxime}
   Any $n$-qubit diagonal unitary with a $s$-sparse sequential decomposition is implementable up to an error $\epsilon>0$ with a quantum circuit of depth $\mathcal{O}(s\log(n)^3\log(s/\epsilon))$ and size $\mathcal{O}(sn\log(n)^4\log(s/\epsilon))$ without using ancilla qubits.
\end{lemma}
\begin{proof}
    The quantum circuit is composed of $s$ operator $\hat{U}_j$ defined in Eq. \eqref{sequential decomposition}. Each of the $s$ sequential operator contains at worst $2n$ $\hat{X}$-Pauli gates and one $\Lambda_{\{0,...,n-2\}}(\hat{P}(\theta_j))$. Each of the $s$ multi-controlled phase gate is implemented using the approximative method of Claudon et al. (Proposition 2 in \cite{claudon2024polylogarithmic}) with error $\epsilon'=\epsilon/s$,  depth $\mathcal{O}(\log(n)^3\log(s/\epsilon))$ and size $\mathcal{O}(n\log(n)^4\log(s/\epsilon))$.
\end{proof}

\begin{lemma}[Sparse sequential decomposition with one ancilla]
\label{lemma : Sequential one ancilla : Sparse}
   Any $n$-qubit diagonal unitary with a $s$-sparse sequential decomposition is exactly implementable with a quantum circuit of depth $\mathcal{O}(s\log(n)^3)$ and size $\mathcal{O}(sn\log(n)^4)$ using one ancilla qubit.
\end{lemma}
\begin{proof}
    The quantum circuit is composed of $s$ operator $\hat{U}_j$ defined in Eq. \eqref{sequential decomposition}. Each of the $s$ sequential operator contains at worst $2n$ $\hat{X}$-Pauli gates and one $\Lambda_{\{0,...,n-2\}}(\hat{P}(\theta_j))$. Each of the $s$ multi-controlled phase gate is implemented using the exact method of Claudon et al. (Corollary 1 in \cite{claudon2024polylogarithmic}) with  depth $\mathcal{O}(\log(n)^3)$, size $\mathcal{O}(n\log(n)^4)$ using one ancilla qubit.
\end{proof}

\begin{corollary}[Adjustable-depth sparse sequential decomposition]
\label{cor : Sequential adjustable-depth : Sparse}
 Any $n$-qubit diagonal unitary with a $s$-sparse sequential decomposition is exactly implementable with a quantum circuit of depth $\mathcal{O}(s\log(n)^3/(m/n) + \log(m/n))$ and size $\mathcal{O}(sn\log(n)^4+m)$ using $m\geq n+2$ ancilla qubits.
\end{corollary}
\begin{proof}
This corollary is a direct consequence of Lemma \ref{lemma : Sequential one ancilla : Sparse} and the adjustable-depth Theorem \ref{thm : adjustable-depth with ancilla}.
\end{proof}

\begin{corollary}[Fully parallelized sparse sequential decomposition ]
\label{cor : Sequential fully parallelized : Sparse}
 Any $n$-qubit diagonal unitary with a $s$-sparse sequential decomposition is exactly implementable with a quantum circuit of depth $\mathcal{O}(\log(n)^3+\log(s))$ and size $\mathcal{O}(sn\log(n)^4)$ using $m=\mathcal{O}(ns)$ ancilla qubits.
\end{corollary}
\begin{proof}
This corollary is a direct consequence of lemma \ref{lemma : Sequential one ancilla : Sparse} and the fully parallelized Theorem \ref{thm : full parallelization}.
\end{proof}

\begin{lemma}[Sparse Walsh-Hadamard decomposition without ancilla]
\label{lemma : walsh no ancilla : Sparse}
Any $n$-qubit diagonal unitary with a $s$-sparse Walsh-Hadamard decomposition is exactly implementable with a quantum circuit of depth $\mathcal{O}(sk)$ and size $\mathcal{O}(sk)$ without using ancilla qubits, where $k$ is the maximum number of $1$ in the binary decomposition of the indexes $j$ of the Walsh-Hadamard operators.
\end{lemma}
\begin{proof}
  The quantum circuit is composed of $s$ operator $\hat{W}_j$ defined in Eq. \eqref{eq: exp of walsh operator} for $j\in S \subseteq \{0,1,...,2^n-1\}$, with $\|S\|=s$. Each Walsh-Hadamard operator $\hat{W}_j$ for $j=\sum_{i=0}^{n-1}j_i 2^i \in S$ contains exactly $2k_j$ controlled-NOT gates and one $\hat{R}_Z$ gate where $k_j$ is the number of $1$ in the binary decomposition of $j$, a.k.a. $k_j=\sum_{i=0}^{n-1}j_i$. By defining $k=\max_{j\in S} k_j$, the size of the quantum circuit is bounded by $\mathcal{O}(sk)$ and the depth is also bounded by $\mathcal{O}(sk)$ where $k\leq n$ is independant of $n$.
\end{proof}

\begin{corollary}[Adjustable-depth sparse Walsh-Hadamard decomposition]
\label{cor : walsh adjustable-depth : Sparse}
Any $n$-qubit diagonal unitary with a $s$-sparse Walsh-Hadamard decomposition is exactly implementable with a quantum circuit of depth $\mathcal{O}(sk/(m/k)+\log(m/k))$ and size $\mathcal{O}(sk+m)$ without using $m\geq k$ ancilla qubits, where $k$ is the maximum number of $1$ in the binary decomposition of the indexes $j$ of the Walsh-Hadamard operators.
\end{corollary}
\begin{proof}
This corollary is a direct consequence of lemma \ref{lemma : walsh no ancilla : Sparse} and the adjustable-depth Theorem \ref{thm : adjustable-depth with ancilla}.
\end{proof}

\begin{corollary}[Fully-parallelized sparse Walsh-Hadamard decomposition]
\label{cor : walsh fully parallelized : Sparse}
Any $n$-qubit diagonal unitary with a $s$-sparse Walsh-Hadamard decomposition is exactly implementable with a quantum circuit of depth $\mathcal{O}(k+\log(s))$ and size $\mathcal{O}(sk)$ without using $m=\mathcal{O}(sk)$ ancilla qubits, where $k$ is the maximum number of $1$ in the binary decomposition of the indexes $j$ of the Walsh-Hadamard operators.
\end{corollary}
\begin{proof}
This corollary is a direct consequence of lemma \ref{lemma : walsh no ancilla : Sparse} and the fully parallelized Theorem \ref{thm : full parallelization}.
\end{proof}

\subsubsection{Sparse non-unitary diagonal operators}\label{sparse nu diag} 
The following corollaries summarize the complexity of implementing any $n$-qubit non-unitary diagonal operators $\hat{D}$ with an associated $s$-sparse operator $e^{\pm i\arcsin(\hat{D}/(\alpha d_{\max}))}$, with $\alpha > 1$.

\begin{corollary}[Block-encoding of sparse sequential decomposition with one ancilla]
\label{lemma : Sequential block-encoding no ancilla : Sparse}
   For any $\alpha\geq 1$, any $n$-qubit non-unitary diagonal operator with a $s$-sparse sequential decomposition can be $(\alpha d_{\max},1,0)$-block-encoded with a quantum circuit of depth $\mathcal{O}(ns)$ and size $\mathcal{O}(ns)$ using one ancilla qubit.
\end{corollary}
\begin{proof}
This corollary is a direct consequence of lemma \ref{lemma : Sequential no ancilla : Sparse} and the block-encoding Theorem \ref{thm: block-encoding}. Remark that if $\hat{D}$ has $s$ non-zero real eigenvalues, then $\hat{U}^=e^{\pm i\arcsin(\hat{D}/(\alpha d_{\max}))}$ has at most $s$ eigenvalues different than $1$.
\end{proof}

\begin{corollary}[Block-encoding of approximate sparse sequential decomposition with one ancilla]
\label{lemma : block-encoding Sequential no ancilla : Sparse approxime}
For any $\alpha\geq 1$, any $n$-qubit non-unitary diagonal operator with a $s$-sparse sequential decomposition can be $(\alpha d_{\max},1,\epsilon)$-block-encoded with a quantum circuit of depth $\mathcal{O}(s\log(n)^3\log(s/\epsilon))$, size $\mathcal{O}(sn\log(n)^4\log(s/\epsilon))$ using one ancilla qubit.
\end{corollary}
\begin{proof}
This corollary is a direct consequence of lemma \ref{lemma : Sequential no ancilla : Sparse approxime} and the block-encoding Theorem \ref{thm: block-encoding}.
\end{proof}

\begin{corollary}[Block-encoding of sparse sequential decomposition with two ancilla qubits]
\label{lemma : block-encoding Sequential one ancilla : Sparse}
For any $\alpha\geq 1$, any $n$-qubit non-unitary diagonal operator with a $s$-sparse sequential decomposition can be $(\alpha d_{\max},2,0)$-block-encoded with a quantum circuit of depth $\mathcal{O}(s\log(n)^3)$, size $\mathcal{O}(sn\log^4(n))$ using two ancilla qubit.
\end{corollary}
\begin{proof}
This corollary is a direct consequence of lemma \ref{lemma : Sequential one ancilla : Sparse} and the block-encoding Theorem \ref{thm: block-encoding}.
\end{proof}

\begin{corollary}[Adjustable-depth sparse sequential decomposition ]
\label{cor : block-encoding Sequential adjustable-depth : Sparse}
For any $\alpha\geq 1$, any $n$-qubit non-unitary diagonal operator with a $s$-sparse sequential decomposition can be $(\alpha d_{\max},m,0)$-block-encoded with a quantum circuit of depth $\mathcal{O}(s\log(n)^3/(m/n) + \log(m/n))$, size $\mathcal{O}(sn\log^4(n)+m)$ using $m$ ancilla qubits with $m=\Omega(n)$ and $m=\mathcal{O}(ns)$.

\end{corollary}
\begin{proof}
This corollary is a direct consequence of corollary \ref{cor : Sequential adjustable-depth : Sparse} and the block-encoding Theorem \ref{thm: block-encoding}.
\end{proof}

\begin{corollary}[Fully parallelized sparse sequential decomposition]
\label{cor : block-encoding Sequential fully parallelized : Sparse}
For any $\alpha\geq 1$, any $n$-qubit non-unitary diagonal operator with a $s$-sparse sequential decomposition can be $(\alpha d_{\max},m,0)$-block-encoded with a quantum circuit of depth $\mathcal{O}(\log(n)^3+ \log(s))$, size $\mathcal{O}(sn\log^4(n))$ using $m=\mathcal{O}(ns)$ ancilla qubits.
\end{corollary}
\begin{proof}
This corollary is a direct consequence of Corollary \ref{cor : Sequential fully parallelized : Sparse} and the block-encoding Theorem \ref{thm: block-encoding}.
\end{proof}

\begin{corollary}[Sparse Walsh-Hadamard decomposition with one ancilla]
\label{lemma : block-encoding walsh no ancilla : Sparse}
Any $n$-qubit non-unitary diagonal $\hat{D}$ with real eigenvalues such that $\hat{U}^=e^{\pm i\arcsin(\hat{D}/(\alpha d_{\max}))}$ has an $\epsilon>0$ approximative $s$-sparse Walsh-Hadamard decomposition can be $(\alpha d_{\max},1,\epsilon)$-block-encoded with a quantum circuit of depth $\mathcal{O}(sk)$, size $\mathcal{O}(sk)$ using one ancilla qubit.
\end{corollary}
\begin{proof}
 This corollary is a direct consequence of lemma \ref{lemma : walsh no ancilla : Sparse} and the block-encoding Theorem \ref{thm: block-encoding}.
\end{proof}

\begin{corollary}[Adjustable-depth sparse Walsh-Hadamard decomposition]
\label{cor :  block-encoded walsh adjustable-depth : Sparse}
Any $n$-qubit non-unitary diagonal $\hat{D}$ with real eigenvalues such that $\hat{U}^=e^{\pm i\arcsin(\hat{D}/(\alpha d_{\max}))}$ has an $\epsilon>0$ approximative $s$-sparse Walsh-Hadamard decomposition can be $(\alpha d_{\max},m,\epsilon)$-block-encoded with a quantum circuit of depth $\mathcal{O}(sk/(m/k)+\log(m/k))$, size $\mathcal{O}(sk+m)$ using $m=\Omega(k)$ ancilla qubits.
\end{corollary}
\begin{proof}
 This corollary is a direct consequence of corollary \ref{cor : walsh adjustable-depth : Sparse} and the block-encoding Theorem \ref{thm: block-encoding}.
\end{proof}

\begin{corollary}[Fully-parallelized sparse Walsh-Hadamard decomposition]
\label{cor : block-encoded walsh fully parallelized : Sparse}
Any $n$-qubit non-unitary diagonal $\hat{D}$ with real eigenvalues such that $\hat{U}^=e^{\pm i\arcsin(\hat{D}/(\alpha d_{\max}))}$ has an $\epsilon>0$ approximative $s$-sparse Walsh-Hadamard decomposition can be $(\alpha d_{\max},m,\epsilon)$-block-encoded with a quantum circuit of depth $\mathcal{O}(k+\log(s))$, size $\mathcal{O}(sk)$ using $m=\mathcal{O}(sk)$ ancilla qubits.
\end{corollary}
\begin{proof} 
This corollary is a direct consequence of corollary \ref{cor : walsh fully parallelized : Sparse} and the block-encoding Theorem \ref{thm: block-encoding}.
\end{proof}

\section{Discretization error of the diffusion equation}
\label{sec:discretization error}

In this Appendix, the bound on the discretization error of the diffusion equation is derived. The space discretization of the diffusion equation and the centered finite difference approximation of the one-dimensional Laplacian operator introduce an error in the resolution of the diffusion equation. We define $\ket{f}_t=\sum_{x=0}^{N-1}f(x/N,t)\ket{x}$ as the non-normalized state encoding the solution of the diffusion Equation (\ref{eq : heat equation}) and $\ket*{\tilde{f}}_t$ the solution of the ordinary differential equation (\ref{eq:Discrete ODE}) obtained after space discretization. Notice that
\begin{equation}
\begin{split}
\partial_t \ket{f}_t &=\kappa\frac{(\hat{S}-\hat{S}^\dagger)^2}{4\Delta x^2} \ket{f}_t+\ket{r}_t \\
\ket{f}_{t=0}&=\ket{f_0}
\end{split}
\end{equation}
with $\ket{r}_t=\kappa \ket{\partial_{xx}f}_t-\kappa\frac{(\hat{S}-\hat{S}^\dagger)^2}{4\Delta x^2} \ket{f}_t$ and $\ket{\partial_{xx}f}_t=\sum_{x=0}^{N-1}\partial_{xx}f(x/N,t)\ket{x}$. The variation of parameter formula implies:
\begin{equation}
\ket{f}_t=\ket*{\tilde{f}}_t+\int_0^t\hat{A}(t,s)\ket{r(s)}ds
\end{equation}
where $\hat{A}(t,s)=e^{\kappa(t-s)\frac{(\hat{S}-\hat{S}^\dagger)^2}{4\Delta x^2}}$. Therefore, by noticing that $\|\hat{A}\|_2\leq1$, the difference can be bounded as:
\begin{equation}
\| \ket{f}_t-\ket*{\tilde{f}}_t\|_{2,N} \leq t \kappa \max_{s\in[0,t]}\| \ket{\partial_{xx}f}_s-\frac{(\hat{S}-\hat{S}^\dagger)^2}{4\Delta x^2} \ket{f}_s \|_{2,N}
\end{equation}
Supposing that $f$ is three times differentiable, one can use the Taylor formula $f(b)=\sum_{j=0}^{p-1}\frac{f^{(j)}(a)}{j!}(b-a)^j+\int_a^b\frac{f^{(p)}(s)}{(p-1)!}(b-s)^{p-1}ds$ to bound the last term as:
\begin{equation}
\| \ket{\partial_{xx}f}_s-\frac{(\hat{S}-\hat{S}^\dagger)^2}{4\Delta x^2} \ket{f}_s \|_{2,N} \leq \frac{2}{3\sqrt{N}} \| \partial_{xxx}f \|_{\infty}
\end{equation}
Using the triangular inequality, one can now bound the difference of the normalized states as:
\begin{equation}
\| \frac{\ket{f}_t}{\| \ket{f}_t \|_{2,N}}-\frac{\ket*{\tilde{f}}_t}{\| \ket*{\tilde{f}}_t \|_{2,N}} \|_{2,N} \leq t \kappa \frac{2}{3\sqrt{N} \| \ket{f}_t \|_{2,N} } \| \partial_{xxx}f \|_{\infty}
\end{equation}

Remark that $\frac{\sqrt{N}}{\|f_s\|_{2,N}}$ converges toward $\frac{1}{\sqrt{\int_0^1 |f(x,s)|^2 dx}}$. Then, since the solution $f_t$ of the diffusion equation is bounded for all time $t$, there is a constant $K>0$, depending on the maximum over space and time of the third derivative of $f$ such that:
\begin{equation}
\| \frac{\ket{f}_t}{\| \ket{f}_t \|_{2,N}}-\frac{\ket*{\tilde{f}}_t}{\| \ket*{\tilde{f}}_t \|_{2,N}} \|_{2,N}  \leq t K/N
\end{equation}

\end{document}